\documentclass[11pt]{article}

\usepackage[T1]{fontenc}

\usepackage[left=0.9in,right=0.9in,top=1in,bottom=1in]{geometry}
\usepackage{newtxtext,newtxmath}

\usepackage{amsmath,amssymb,amsthm}
\usepackage{mathtools}
\usepackage{bm}
\usepackage{mathrsfs}
\usepackage{siunitx}
\usepackage{authblk}

\usepackage{graphicx}
\usepackage{float}
\usepackage{subcaption}
\usepackage{booktabs}
\usepackage{array}
\usepackage{multirow}

\usepackage[ruled,vlined]{algorithm2e}

\usepackage{hyperref}
\hypersetup{
  colorlinks=true,
  linkcolor=blue,
  citecolor=blue,
  urlcolor=blue
}

\theoremstyle{definition}
\newtheorem{theorem}{Theorem}[section]
\newtheorem{lemma}[theorem]{Lemma}
\newtheorem{corollary}[theorem]{Corollary}

\newcommand{\E}{{\rm I\kern-.3em E}}

\newcommand{\lmax}{l_{\text{max}}}

\DeclareMathOperator{\argmin}{\arg\min}
\DeclareMathOperator{\argmax}{\arg\max}
\def\dmodel{d_{\mbox{\small model}}}
\def\dtypebytes{\mbox{dtype\_bytes}}
\title{\Large\bfseries
Optimizing Resource Allocation for Geographically-Distributed Inference by Large Language Models
}

\author[1]{Tingyang Sun}
\author[1]{Ting He\thanks{This work was partly supported by the National Science Foundation under award CNS-2106294. Parimal was supported by Qualcomm Inc. under Qualcomm University Relations 6G India, Anusandhan National Research Foundation (ANRF) under Grant CRG/2023/008854, and the Office of International Relations at IISc under the IISc--PSU collaboration research grant. Corresponding author: \texttt{tinghe@psu.edu}.}}
\author[2]{Bo Ji}
\author[3]{Parimal Parag}

\affil[1]{\small Department of Computer Science and Engineering, Pennsylvania State University, University Park, PA, USA}
\affil[2]{\small Department of Computer Science, Virginia Tech, Blacksburg, VA, USA}
\affil[3]{\small Department of Electrical Communication Engineering, Indian Institute of Science, Bangalore, India}

\date{}

\begin{document}
\maketitle

\begin{abstract}
Large language models (LLMs) have demonstrated extraordinary performance in many artificial intelligence (AI) tasks but are expensive to use, even after training, due to their requirement of high-end GPUs. Recently, a distributed system called PETALS was developed to lower the barrier for deploying LLMs by splitting the model blocks across multiple servers with low-end GPUs distributed over the Internet, which was much faster than swapping the model parameters between the GPU memory and other cheaper but slower local storage media. However, the performance of such a distributed system critically depends on the resource allocation, and how to do so optimally remains unknown. In this work, we present the first systematic study of the resource allocation problem in distributed LLM inference, with focus on two important decisions: block placement and request routing. Our main results include: (i) experimentally validated performance models that can predict the inference performance under given block placement and request routing decisions, (ii) a formulation of the offline optimization of block placement and request routing as a mixed integer linear programming (MILP) problem together with the NP-hardness proof and a polynomial-complexity algorithm with  guaranteed performance, and (iii) an adaptation of the offline algorithm for the online setting with the same performance guarantee under bounded load. Through both experiments and experimentally-validated simulations, we have verified that the proposed solution can substantially reduce the inference time compared to the state-of-the-art solution in diverse settings with geographically-distributed servers. As a byproduct, we have also developed a light-weighted CPU-only simulator capable of predicting the performance of distributed LLM inference on GPU servers, which can evaluate large deployments and facilitate future research for researchers with limited GPU access. 
\end{abstract}


\section{Introduction}\label{sec:Introduction}

Large language models (LLMs) have yielded exceptional performance in many artificial intelligence (AI) tasks. Modern LLMs pre-trained on large datasets have demonstrated promising utility for diverse applications~\cite{Brown20NeurIPS}. 
However, to achieve state-of-the-art accuracy, such models often need to contain over 50 billion (50B) parameters, which makes it expensive to work with these models due to the requirement of expensive hardware such as high-end GPUs. 

To broaden the accessibility to LLMs, several research groups have open-sourced their pre-trained LLMs \cite{Zhang2022OPT,BigScience23BLOOM,Touvron2023llama,Touvron2023llama2}. However, the sheer size of these models makes it challenging to adopt them, even if no training is needed. Prior studies have shown that the bottleneck of running LLMs is not in the computation speed, as even a consumer-grade GPU like GeForce RTX 3070 has enough processing power to run a complete inference step of a 176B-parameter model within a second \cite{Nvidia2020}. Instead, the bottleneck is in the GPU memory. A 100B-parameter model will require 200 GB of GPU memory to load the model parameters at the standard half precision, and many LLMs are even larger (e.g., GPT-3 has 175B parameters and GPT-4 has over $1.7$ trillion parameters). Providing this much GPU memory is financially challenging for many  users. Straightforward approaches to address this bottleneck such as compressing the model or offloading the model parameters to cheaper storage media will introduce undesirable side effects. For example, a compressed LLM can still be too large for consumer-grade GPUs and too much compression will lower the inference accuracy~\cite{chen2024comprehensive}, while offloading the model parameters to larger but slower storage media like RAM or SSD 
will be very slow due to the limited bandwidth between such storage media and the GPU~\cite{Borzunov23NeurIPS}. 

Recently, a model-parallel approach has been proposed to address the above challenge through resource pooling across distributed devices. This approach distributes a model to multiple devices at the granularity of transformer \emph{blocks} (a.k.a. \emph{pipeline parallelism}~\cite{Huang19NeurIPS,Narayanan19SOSP,Yang21MLsys}) or neurons (a.k.a. \emph{tensor parallelism}~\cite{Krizhevsky17CommACM,Ben-Nun19ACMCompSurv,Tang20arXiv}) to run large models on devices with small GPUs. In particular, \cite{Borzunov23NeurIPS} has shown that pipeline parallelism can be used to run LLM inference tasks over \emph{geographically-distributed} servers, each with only a few GB of GPU memory, at a much faster rate than local parameter offloading. This is achieved by letting each server host a subset of consecutive blocks, and each inference \emph{request} (for autoregressive sequence generation) routed through a \emph{chain of servers} that collectively host the entire model, as illustrated in Fig.~\ref{fig:system_architecture}. To reduce the communication cost, caching is used pervasively in such systems. For example, in a state-of-the-art system called PETALS~\cite{Borzunov23NeurIPS}, the client of each request uses client-side caches to store the input history for each invoked server, so that once a server fails, the client can use the cached input to bring up a backup server without repeating the processing at other servers; meanwhile, each server uses server-side caches (a.k.a. \emph{attention caches}) to store the past computation results for each ongoing request (i.e., the key-value pairs for the past tokens~\cite{Vaswani17NeurIPS}), so that at each inference step (after the prefill phase), the client only needs to send the partially processed embedding of a single token to the next server and receive the processed embedding back, which only exchanges a few KB of data~\cite{Borzunov23NeurIPS}. 
The client-side caching implies a hub-spoke communication pattern as shown in Fig.~\ref{fig:system_architecture}, where all the communications during an {inference session} are anchored at the client, who will forward the partially processed embeddings between consecutive servers in order to update the input cache for each server.  

\begin{figure}[!t]
   \centerline{\includegraphics[width=0.55\linewidth]{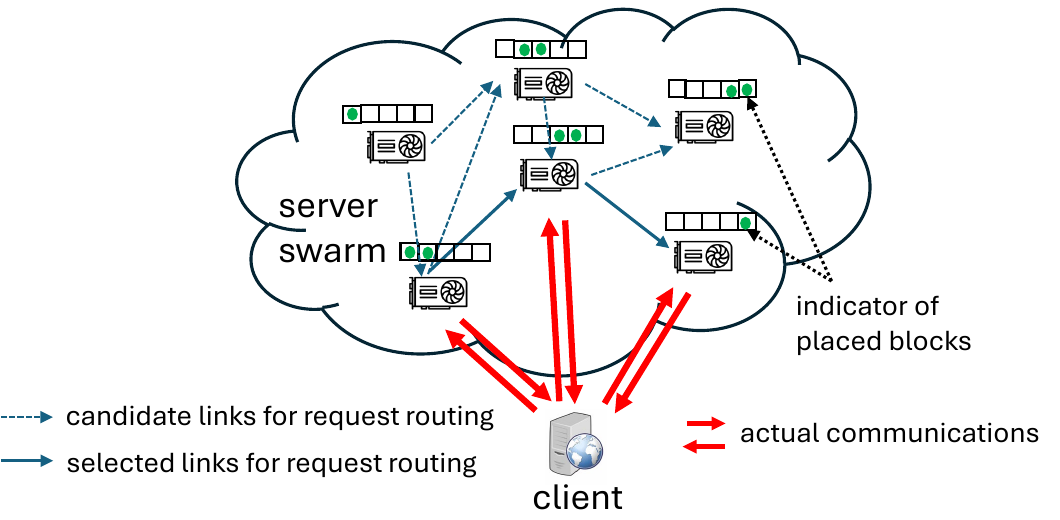}}
   \vspace{-1em}
    \caption{System architecture for pipeline-parallel LLM inference using client-centric communication. }
    \label{fig:system_architecture}
    \vspace{-.05em}
\end{figure}

In this work, we consider LLM inference over geographically-distributed servers using \emph{pipeline parallelism} and \emph{client-centric communication} as illustrated in Fig.~\ref{fig:system_architecture}. Although the feasibility of such systems has been validated in \cite{Borzunov23NeurIPS}, there is a lack of fundamental understanding of how to optimally manage their performance. 
In particular, the current solution in \cite{Borzunov23NeurIPS} relies on heuristics for resource allocation, e.g., selecting the blocks for each server based on a heuristic ``throughput'' metric and selecting the server chain for each request based on a graph with heuristic ``edge weights''~\cite{PetalsGitHub}.  
It remains open how to optimize the performance of such systems, while taking into account the unique characteristics of GPUs and LLM inference tasks. This work aims at filling this gap by rigorously formulating and tackling the \emph{block placement and request routing} problem for a distributed LLM inference system with the architecture in Fig.~\ref{fig:system_architecture}, using PETALS~\cite{Borzunov23NeurIPS} as a concrete example. 
While our system model and evaluations are based on PETALS, our approach is applicable to any pipeline-parallel LLM inference systems  as explained later. \looseness=-1  

\subsection{Related Work}\label{subsec:Related Work}

Below we will briefly review existing approaches to overcome the GPU memory limitation in running large models and related resource allocation problems. 

\textbf{Model parallelism:}
Model parallelism overcomes the GPU memory limitation by distributing the model parameters to multiple devices. 
It is further divided into \emph{pipeline parallelism} that assigns model parameters at the granularity of layers~\cite{Huang19NeurIPS,Narayanan19SOSP,Yang21MLsys}, and \emph{tensor parallelism} that assigns model parameters at the granularity of neurons~\cite{Krizhevsky17CommACM,Ben-Nun19ACMCompSurv,Tang20arXiv}. Tensor parallelism provides more flexibility in splitting the model at the cost of a higher communication overhead due to the all-to-all communication between devices hosting adjacent layers, which makes it more suitable for distributed inference within a single multi-GPU server or a cluster of highly connected servers~\cite{vLLM_distributed}. In contrast, pipeline parallelism only requires communication between pairs of devices at some loss of flexibility, which makes it more suitable for distributed inference across nodes~\cite{Borzunov23NeurIPS}.  PETALS~\cite{Borzunov23NeurIPS} uses pipeline parallelism.    
Other than PETALS, there are a few more systems that support distributed LLM inference across nodes with different system characteristics. For example, vLLM with Ray Serve~\cite{vLLM_Ray} uses tensor parallelism within each node and pipeline parallelism across nodes by letting servers processing adjacent blocks directly communicate partially processed tokens, Nvidia Dynamo~\cite{Nvidia_Dynamo} implements several LLM-specific execution capabilities such as disaggregating prefill\&decoding phases and offloading attention caches to memory hierarchies, and Amazon EKS~\cite{Amazon_EKS} simplifies the execution of multi-node inference workloads through integration with AWS services like Amazon EFS and Elastic Fabric Adapter.    
These systems are mainly designed for datacenter environments with high-bandwidth low-latency connections. 
In contrast, PETALS~\cite{Borzunov23NeurIPS} is the leading system for LLM inference in geographically-distributed environments. 
In this work, \emph{we focus on geographically-distributed LLM inference and thus build our system model after PETALS}, but our solution is generally applicable to any pipeline-parallel LLM inference systems, possibly with minor adaptations (e.g., see Remark after \eqref{eq:request processing time}). While the extension of our solution to tensor parallelism is nontrivial, we note that tensor parallelism is generally considered impractical for geographically-distributed settings.\looseness=-1
%
%

\textbf{Parameter offloading:}
Parameter offloading overcomes the GPU memory limitation by swapping model parameters between the GPU memory and a slower but larger storage, typically RAM or SSD~\cite{Pudipeddi20arXiv,Ren21UsenixATC,Rajbhandari21SC}. When using the model for inference, model parameters are loaded into the GPU just before being used, which in theory allows a server with a small GPU to process any large models that fit into its storage. The drawback of parameter offloading is the overhead in loading and unloading the model parameters, which can take a substantial amount of time (e.g., at least 11 seconds per token for a 175B-parameter model~\cite{Borzunov23NeurIPS}). For large models (with 50B+ parameters), \emph{parameter offloading was shown to be at least an order of magnitude slower than model parallelism}~\cite{Borzunov23NeurIPS}.

\textbf{Service function chaining:}
Under pipeline parallelism, each inference request is served by a chain of servers as illustrated in Fig.~\ref{fig:system_architecture}, which makes the corresponding resource allocation problem (e.g., how to place blocks and how to route requests) analogous to the problem of \emph{service function (SF)} placement and flow routing in the context of \emph{service function chaining (SFC)}~\cite{Medhat17CommMag}. SFC provides application-specific services for in-transit traffic  through a chain of SFs. 
SFC placement and routing has been actively studied in recent years, typically with an objective of minimizing operation cost, maximizing flow rate, optimizing Quality of Service (QoS), or optimizing a combination of these metrics~\cite{Chen20AIS}. The problems are usually NP-hard and solved by polynomial-time heuristics without guaranteed performance~\cite{Jang17JSAC,Chen20AIS}, except for a few special cases such as \cite{Zhang17ICDCS,Guo18INFOCOM,Shang19ICPP} with approximation guarantee. 
Despite the conceptual similarity, our problem is fundamentally different from SFC placement and routing. \emph{First of all}, the two domains have very different properties, e.g., different SFs may have heterogeneous resource requirements~\cite{Jang17JSAC,Zhang17ICDCS,Jalalitabar19NFV-SDN} and some SFs may change the flow rate~\cite{Ma17INFOCOM} or branch an input flow into multiple output flows~\cite{Jalalitabar19NFV-SDN}, while each LLM typically has identically structured transformer blocks that require the same amount of resource and input data size. 
%
\emph{Moreover}, the bottleneck resource in LLM inference 
is GPU memory, which differs from the bottleneck resource in SFC (i.e., CPU cycles) in that: (i) the need of GPU memory is not elastic, making congestion minimization approaches like \cite{Shang19ICPP} inapplicable, (ii) the GPU memory used to host a block is shared by all the requests processed by the same block, differing from the additive resource consumption for CPU cycles~\cite{Zhang17ICDCS,Guo18INFOCOM}, and (iii) nontrivial amounts of GPU memory are needed for both hosting blocks and serving inference sessions (due to the need of holding attention caches), causing resource contention between the placement decision and the routing decision that did not exist in service function chaining or other traditional ``placement and routing'' problems. Additionally, the number of blocks in an LLM can be much larger than the number of SFs in a typical SFC (e.g., BLOOM-176B has 70 blocks and GPT-3 has 96 blocks), making the formulations for SFC that enumerate the processing units \cite{Jang17JSAC,Zhang17ICDCS,Guo18INFOCOM,Shang19ICPP}  highly inefficient. Some LLM inference systems also have unique characteristics such as the hub-spoke communication pattern shown in Fig.~\ref{fig:system_architecture}. 
\emph{There is thus a need of studying resource allocation problems tailored to distributed LLM inference systems}.

\subsection{Summary of Contributions}

We perform a systematic study of the resource allocation problem in distributed LLM inference, with the following contributions: \begin{enumerate}
    \item We formulate a joint optimization of the placement of model blocks at servers (called \emph{block placement}) and the selection of server chains for inference requests (called \emph{request routing}) based on performance models experimentally validated on a real distributed LLM inference system called PETALS~\cite{Borzunov23NeurIPS}. Our formulation features a compact representation of the decision variables that allows the problem to be formulated as a mixed integer linear programming (MILP) problem with a polynomial number of variables/constraints. 
    \item We prove the NP-hardness of the formulated optimization problem, and develop a three-step algorithm (Alg.~\ref{Alg:CG-BPRR}) by decomposing the joint optimization into three subproblems, each optimizing one type of decision variables, which achieves a polynomial complexity and a guaranteed average inference time per token. 
    \item We then consider the online setting with dynamically arriving requests, for which we show that the offline algorithm can be adapted into a two-time-scale solution that solves the block placement problem via robust optimization and the (online) request routing problem via a variation of shortest-path routing, with link costs designed to approximate the total inference time (including waiting, communication, and computation) incurred at each server. The online algorithm inherits the same performance guarantee as the offline algorithm as long as the number of concurrent requests is within a design limit, which can be tuned to trade off the waiting time and the per-token inference time after waiting. 
    \item We compare our algorithms against the state-of-the-art algorithms in \cite{Borzunov23NeurIPS} through both controlled experiments and experimentally-validated simulations in a variety of settings with geographically distributed servers. The results not only show substantial performance improvements by our algorithms ($60$--$80\%$ smaller inference times), but also provide insights on the reasons for such improvements as well as the potential room for further improvements. A byproduct of our evaluation is a CPU-only simulator that is validated to produce reasonable predictions of the actual performance of distributed LLM inference on GPU servers, which can enable the evaluation of large deployments and facilitate future research for researchers with limited GPU access.  
\end{enumerate}

\textbf{Roadmap.} Section~\ref{sec:Background and Formulation} introduces the problem, Section~\ref{sec:Joint Block Placement and Request Routing} presents the optimization formulation and the proposed algorithms together with their analysis, Section~\ref{sec:Performance Evaluation} presents the evaluation results, and Section~\ref{sec:Conclusion} concludes the paper. The response to previous reviews, proofs, and other supporting materials are provided in the Appendix. 

\section{Problem Formulation}\label{sec:Background and Formulation}

\subsection{Target Application Scenario} 

We consider geographically-distributed inference by a pre-trained LLM with a decoder-only architecture~\cite{Radford2018ImprovingLU}, which is commonly adopted by state-of-the-art LLMs (e.g., GPT and its variants~\cite{Radford2018ImprovingLU,Radford2019LanguageMA,Brown20NeurIPS}). The LLM is composed of thin input/output layers to map each input token to an embedding and each output embedding to a probability distribution of the next token, as well as multiple layers of identically structured \emph{transformer blocks} (hereafter referred to as \emph{blocks}), each containing a multi-headed self-attention sub-layer, a fully connected feed-forward sub-layer, and corresponding layer normalizations~\cite{Radford2018ImprovingLU}. The input/output layers only contain a small percentage of the model parameters (e.g., $<3\%$ for BLOOM-176B~\cite{Borzunov23NeurIPS}) and can thus be hosted by the client. The blocks contain most of the model parameters and are thus delegated to the servers~\cite{Borzunov23NeurIPS}. 
Given a pre-trained LLM with $L$ blocks, we consider a distributed system with client-side and server-side caching like PETALS~\cite{Borzunov23NeurIPS}, which uses this LLM to serve autoregressive sequence generation requests (hereafter referred to as \emph{requests}) via a hub-spoke communication pattern as shown in Fig.~\ref{fig:system_architecture}. 
In this work, we focus on supporting \emph{short-prompt long-response queries}, where users ask concise questions to obtain detailed and comprehensive answers. This is a common type of queries, particularly in educational and technical contexts, and adheres to the best practices in using LLMs~\cite{Impact_prompt_length,Shorter_prompt}.  

\subsection{System Model}\label{subsec:System Model}

Suppose that the system contains a set of servers $V_s$ and a set of clients $V_c$, communicating through TCP connections according to the architecture shown in Fig.~\ref{fig:system_architecture}. Each client represents an ingress point for inference requests, which can be a host owned by end users or a proxy server submitting requests on behalf of end users. 
While it is possible for a server to directly forward its output to the next server to continue processing, PETALS~\cite{Borzunov23NeurIPS} lets the client perform such forwarding as in Fig.~\ref{fig:system_architecture}, in order to maintain client-side caches of the input history for each invoked server to achieve fault tolerance. \looseness=0 

\textbf{Request model:}
For each inference request, the client initiates an \emph{inference session}. Due to the one-one correspondence between inference requests and inference sessions, hereafter we will use ``request'' and ``session'' interchangeably. 
We will first consider the offline case with a given set of requests to understand the structure of the problem, and then address the online case with sequentially arriving requests. 
To formulate the offline case, we use $\mathcal{R}$ to denote the total set of inference requests and $\mathcal{R}_c$ to denote the set of requests from client $c\in V_c$. Each request has up to $\lmax^I$ input tokens and generates up to $\lmax$ output tokens\footnote{The actual number of generated tokens may be smaller if the end-of-sequence token is detected, but the maximum number of output tokens specified in the request is used in resource allocation.}, where $\lmax^I$ and $\lmax$ are system parameters with $\lmax^I+\lmax$ upper-bounded by the maximum sequence length supported by the LLM. For short-prompt long-response queries, 
we have $\lmax^I\ll \lmax$. 
To formulate the online case, we need to additionally model the state of each server as detailed in Section~\ref{subsubsec:Online Request Routing}. 

\textbf{Inference time model:}
Each server $j\in V_s$ has a \emph{per-block processing time} of up to $\tau^I_j(\lmax^I)$ during the prefill phase (i.e., phase for processing the input to generate the first token) and a \emph{per-token-per-block processing time} of $\tau_j$ during the decoding phase (i.e., phase for autoregressively generating the remaining tokens). Our experiments have validated $\tau^I_j(\lmax^I)$ to be independent of the output length $\lmax$ and $\tau_j$ to be independent of both the input length $\lmax^I$ and the output length $\lmax$ (see Fig.~\ref{fig:time_input_length}--\ref{fig:time_output_length}). 
Moreover, the connection between each client $c$ and each server $j$ has a \emph{per-input round trip time (RTT)} of up to $t^I_{cj}(\lmax^I)$, representing the communication time (including serialization/deserialization, propagation, queueing, and transmission delays) for transmitting one input sequence  from the client to the server and the processed input sequence from the server back to the client. Similarly, the connection also has a \emph{per-token RTT} of $t_{cj}$, representing the communication time for transmitting one token from the client to the server and the processed token from the server back to the client. 
%
This means that if a request from client $c$ with input length $\lmax^I$ and output length $\lmax$ is processed by a chain of servers $p$ with each server $j\in p$ processing $k_j$ blocks, it will incur a \emph{total inference time} (including all the client-server communication time and all the server processing time) of 
\begin{align}
 \sum_{j\in p}\left(t^I_{cj}(\lmax^I) + k_j \tau^I_j(\lmax^I) \right) + (\lmax-1) \sum_{j\in p}(t_{cj}+k_j \tau_j), \label{eq:request processing time}
\end{align}
where $\sum_{j\in p}\left(t^I_{cj}(\lmax^I) + k_j \tau^I_j(\lmax^I) \right)$ is the time to query the servers for the first token, and $\sum_{j\in p}(t_{cj}+k_j \tau_j)$ is the time to query the servers for each of the remaining tokens. In the case of exceptions (e.g., out-of-memory error), additional delays will apply; see Section~\ref{subsubsec:Online Request Routing} for details. 
In addition, each session incurs some delays at the client due to local computations such as tokenizing the input sequence and extracting the next token from each embedding processed by the servers, but since these delays are not affected by the resource allocation at servers (see Section~\ref{subsec:Resource Allocation Problem}), they only contribute a constant shift which is ignored in our model. 
\emph{Remark:} The above model assumes client-centric communication as in Fig.~\ref{fig:system_architecture}. If adjacently traversed servers directly communicate as in \cite{vLLM_Ray}, the communication times will become $\sum_{(i,j)\in p}t^I_{ij}(\lmax^I)$ for the first token and $\sum_{(i,j)\in p}t_{ij}$ for each of the subsequent tokens, where $t^I_{ij}(\lmax^I)$/$t_{ij}$ is the per-input/per-token one-way delay from node $i$ to node $j$. Our optimization formulation and solution are easily amendable to allow direct server-server communications by simply replacing the client-server RTT $t_{cj}$ with the server-server one-way delay $t_{ij}$.

\textbf{Memory consumption model:}
Each server $j\in V_s$ has a GPU memory of $M_j$, which denotes the effective memory capacity for storing model parameters and attention caches (see Remark after \eqref{eq:model of total memory consumption}). 
The \emph{size of each block} is $s_m$ (bytes), given by the number of model parameters per block times the number of bytes per parameter. In addition, each server holds an attention cache of the past key-value pairs for each ongoing request routed through the server and each block it processes for that request. For a total sequence length of $\lmax^I+\lmax$, each attention cache stores two tensors (one for the keys and one for the values), each containing $\dmodel\cdot  (\lmax^I+\lmax)$ parameters, where $\dmodel$ is the embedding dimension of the LLM. Thus, the \emph{size of each attention cache} is $s_c:=2\dmodel\cdot (\lmax^I+\lmax) \cdot \dtypebytes$ (bytes), where `$\dtypebytes$' is the number of bytes per parameter in the cache (usually $\dtypebytes=2$).  
This means that a server $j$ storing $m_j$ blocks and processing $k^r_j$ blocks for each request $r$ has a \emph{total memory consumption} of 
\begin{align}\label{eq:model of total memory consumption}
s_m m_j + s_c \sum_{r\in \mathcal{R}} k^r_j.
\end{align}

\emph{Remark:} In addition to storing the model parameters for the hosted blocks and the attention caches, the GPU memory also needs to hold the CUDA context (depending on the CUDA version and the device) as well as a small number of intermediate variables generated during inference (assuming in-place updating), but this space is invariant to the number of placed blocks and the number of hosted inference sessions, and can thus be modeled as a constant overhead in the memory consumption. In practice, there is also some waste of GPU memory due to memory fragmentation. Thus, the memory capacity $M_j$ is the effective memory for resource allocation and should be slightly smaller than the physical memory capacity 
to avoid the out-of-memory errors. 

\begin{figure}[t!]
\begin{minipage}{.495\linewidth}
\centerline{
\includegraphics[width=1\linewidth,height=2in]{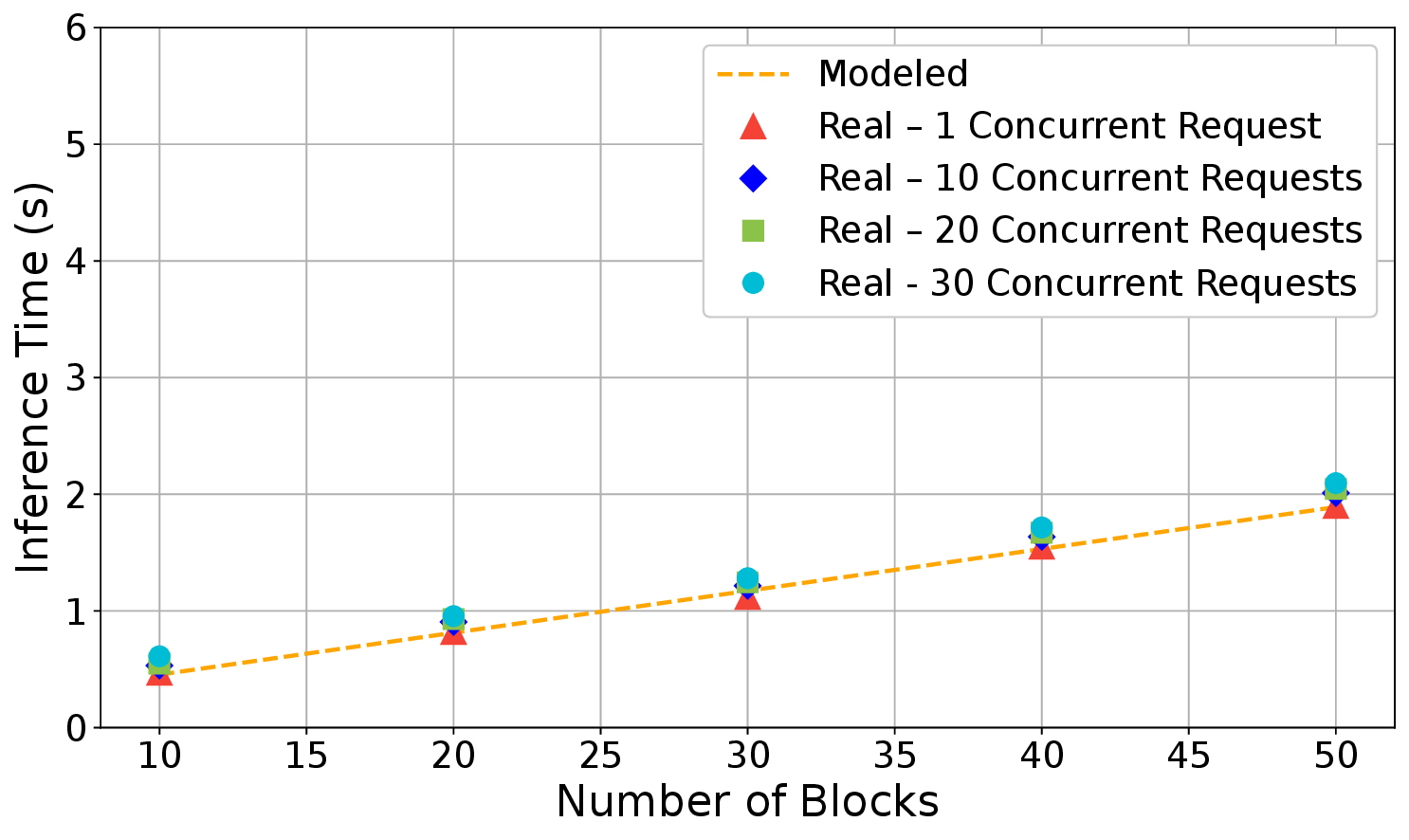}}
\centerline{\scriptsize (a) first token}
\end{minipage}
\begin{minipage}{.495\linewidth}
\centerline{
\includegraphics[width=1\linewidth,height=2in]{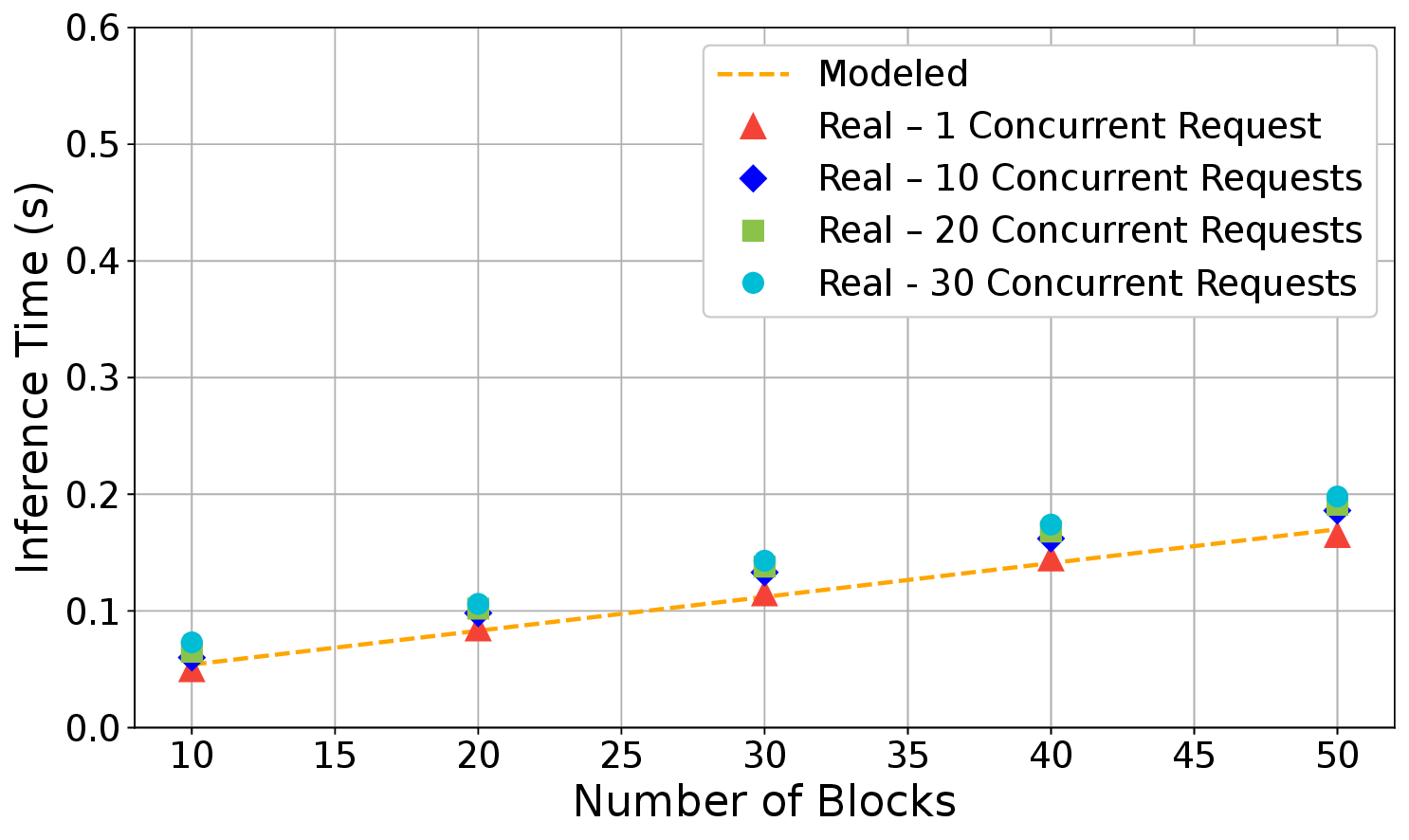}}
\centerline{\scriptsize (b) each of remaining tokens}
\end{minipage}
\vspace{-1em}
\caption{Inference time vs. \#processed blocks on A100 for: (a) first token; (b) each of remaining tokens ($\lmax^I=20,\ \lmax=128$).} \label{fig:time_blocks}
\vspace{-.05em}
\end{figure}

\begin{figure}[t!]
\begin{minipage}{.495\linewidth}
\centerline{
\includegraphics[width=1\linewidth,height=2in]{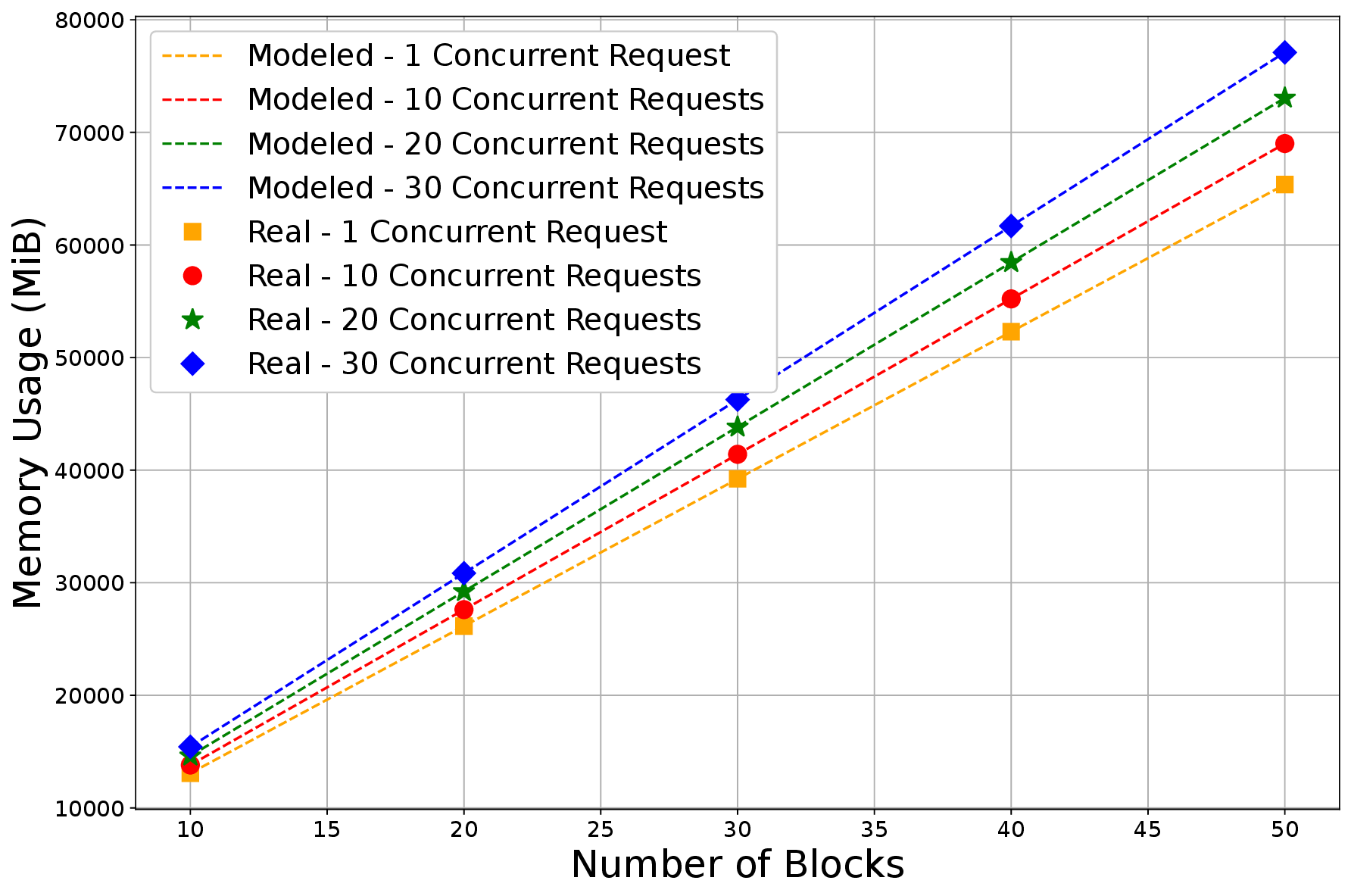}}
\centerline{\scriptsize (a) total memory usage }
\end{minipage}
\begin{minipage}{.495\linewidth}
\centerline{
\includegraphics[width=1\linewidth,height=2in]{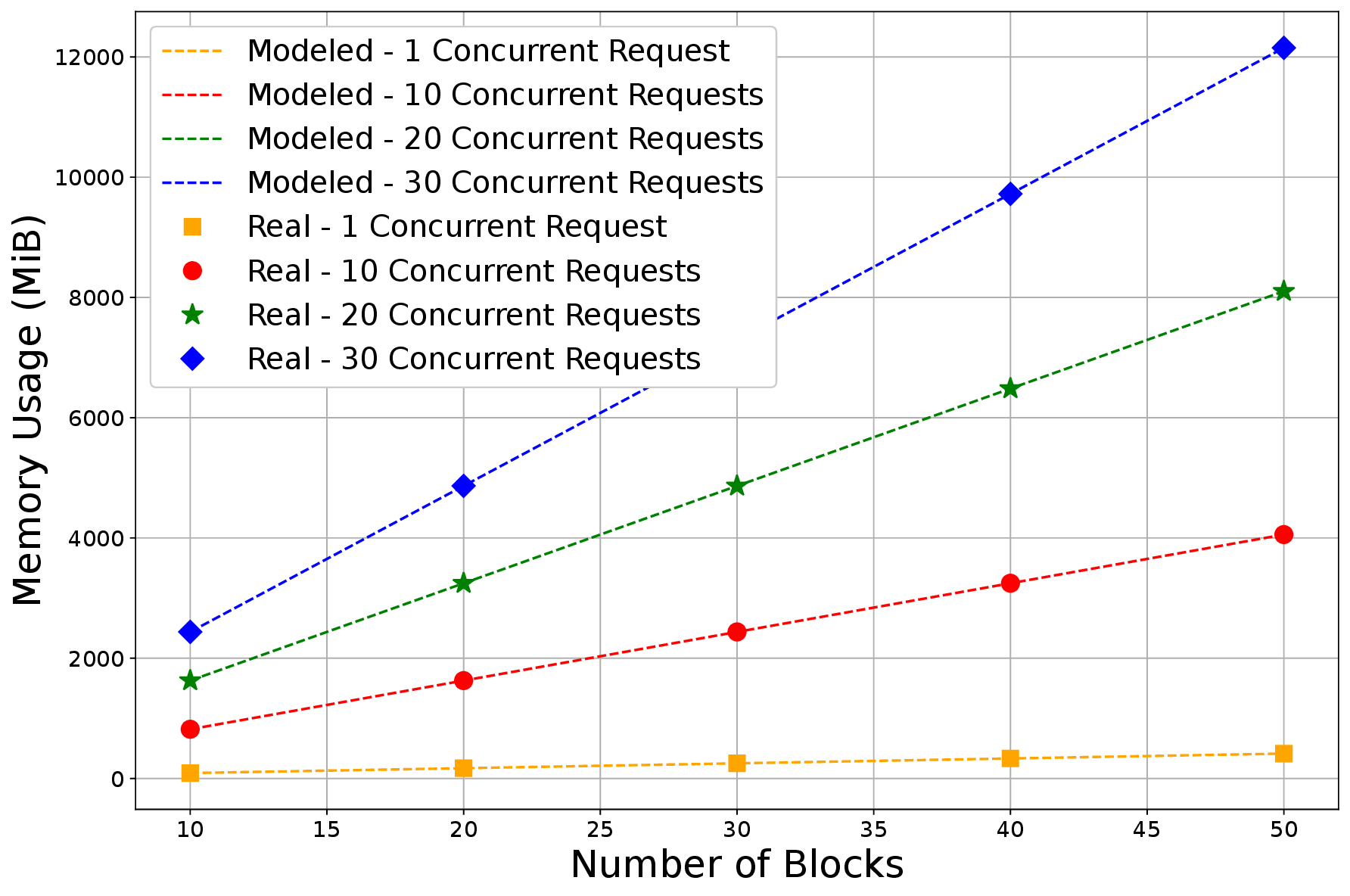}}
\centerline{\scriptsize (b) memory used by attention caches }
\end{minipage}
\vspace{-1em}
\caption{Memory consumption vs. \#blocks on A100 for: (a) total memory usage; (b) attention caches ($\lmax^I=20,\: \lmax = 128$). } \label{fig:memory_block}
\vspace{-.05em}
\end{figure}

\textbf{Experimental validation:} 
We have validated the above models through various experiments based on PETALS~\cite{Borzunov23NeurIPS} and BLOOM-176B~\cite{BigScience23BLOOM}. 
As an example, Fig.~\ref{fig:time_blocks} compares the per-token inference time incurred at a given server according to our model in \eqref{eq:request processing time} with the average time measured from 10 Monte Carlo runs of controlled experiments with various \#processed blocks and \#concurrent requests, in a basic setting of co-located server and client (where the communication time is just the time for serializing and deserializing tokens), and similar results hold under other settings with network delays. The results validate the accuracy of our proposed model 
{in capturing the linear dependency of the inference time on the number of processed blocks. The results also indicate that the inference time for one request is largely independent of the number of concurrent requests at the same server (as long as they all fit into the GPU memory), thanks to the massive parallel computing capability of GPU}. 
Fig.~\ref{fig:time_input_length} and \ref{fig:time_output_length} in \ref{appendix:Additional Model Validation} provide similar comparisons with respect to the input/output length, respectively, again validating the accuracy of our inference time model. 
%
%
Fig.~\ref{fig:memory_block}a compares the GPU memory consumption predicted by \eqref{eq:model of total memory consumption} with the actual allocated GPU memory, which  validates the accuracy of \eqref{eq:model of total memory consumption}. 
A closer examination reveals that the first term in \eqref{eq:model of total memory consumption} (memory used to store blocks) is much larger than the second term (memory used to store attention caches), and only the second term depends on the number of concurrent requests as shown in Fig.~\ref{fig:memory_block}b. 
Nevertheless, it is crucial to model the impact of concurrent requests on memory consumption, as the GPU memory is typically the bottleneck resource in LLM inference (see Remark in Section~\ref{subsec:Resource Allocation Problem}).

\subsection{Resource Allocation Problem}\label{subsec:Resource Allocation Problem}

\begin{table}[]
\small
    \centering
    \begin{tabular}{l|l}
       Notation  &  Description  \\
         \hline
$a_j, m_j$ & index of first block \& \#blocks placed on server $j$ \\
$f^r_{ij}, f^r_p$ & indicator for request $r$ to be routed on link $(i,j)$ or path $p$ \\
\hline
$L$ & total \#blocks in the LLM under consideration\\
$V_c, V_s$ & set of clients/servers \\
$(V, E)$ & logical topology for joint block placement and request routing \\
$(V^c, E^c_{\bm{a},\bm{m}})$ & logical topology for request routing for client $c$ under block placement $(\bm{a},\bm{m})$ (see Alg.~\ref{Alg:CG-BPRR}) \\
$\mathcal{R}, \mathcal{R}_c$ & set of all the requests or requests from client $c$ \\
$\lmax^I, \lmax$ & maximum \#tokens in an input/output sequence \\
$s_m, s_c$ & size per block and size per attention cache \\
$M_j, \overline{f}_j$ & total memory and maximum \#parallel sessions at server $j$ \\
$\tau^I_j(\lmax^I), \tau_j$ & per-block processing time at server $j$ during prefill or decoding phase \\
$t^I_{cj}(\lmax^I), t_{cj}$ & per-input/per-token RTT between client $c$ and server $j$ \\
$t^{c,I}_{ij}(\lmax^I), t^c_{ij}$ & per-token inference time of a first/later token  at server $j$ for a request from client $c$ routed through $(i,j)$ \\
$\widetilde{t}_j, t_{*j}$ & amortized inference time and maximum per-token RTT for server $j$ (see \eqref{eq:amortized inference time}) \\
$C_b, T_b$ & total capacity (in terms of \#sessions) and total amortized inference time for block $b$ (see Alg.~\ref{Alg:CG-BPRR}) \\
$(T^j_r(t), M^j_r(t))_{r=1}^{R_j(t)}$ & state of server $j$ at time $t$ (see Section~\ref{subsubsec:Online Request Routing}) \\
$t^W_{ij}(t)$ & waiting time for link $(i,j)$ at time $t$ (see \eqref{eq:waiting time on (i,j)})
    \end{tabular}
    \vspace{-.5em}
    \caption{Main notations (first two rows: free decision variables; rest: input parameters or dependent variables).
    }
    \label{tab:main_notations}
    \vspace{-.0em}
\end{table} 

We study the joint optimization of how to place the blocks at the servers, referred to as ``\emph{block placement}'', and how to select the chain of servers for each request, referred to as ``\emph{request routing}''. Our objective is to minimize the average time in serving each request within the resource constraints, with focus on the GPU memory constraint. 
Table~\ref{tab:main_notations} lists the main notations used in our presentation. 

\emph{Remark~1:} The objective of this work is to optimize the inference performance in terms of inference time for a given set of available GPU resources. Other performance measures, e.g., cost of renting GPUs or robustness in the face of unreliable nodes, are left for future work.

\emph{Remark~2:} In theory, there are other resource constraints besides GPU memory that can limit how the inference requests can be served. For example, each server has a limited processing capacity, 
and each client-server connection has a limited throughput. 
In practice, however, GPU memory is usually the bottleneck resource that will be saturated before other resources. 
For example, for BLOOM-176B \cite{BigScience23BLOOM}, a server with an A100 (80 GB) GPU and $100$ Mbits/s bandwidth can process over 700 tokens/second when hosting the maximum number of blocks according to \cite{Borzunov23NeurIPS} (53 blocks) and transfer over 400 tokens/second, which should be enough for over 80 concurrent  sessions without causing notable queueing delays, but the available GPU memory only allows 
21 concurrent sessions for $\lmax^I=20$ and $\lmax = 128$, and even fewer for longer sequences. 
We will thus focus on the GPU memory constraint. \looseness=-1

\section{Joint Block Placement and Request Routing (BPRR)}\label{sec:Joint Block Placement and Request Routing}

Based on the experimentally validated model (Section~\ref{subsec:System Model}),  we will first study BPRR in an offline setting to understand the properties of the problem (Section~\ref{subsec:Offline Setting}), and then use the understanding to develop a solution for the online setting (Section~\ref{subsec:Online Setting}). 

\subsection{Preliminaries}\label{subsec:Performance Modeling}

\begin{figure}[!t]
   \centerline{\includegraphics[width=0.45\linewidth]{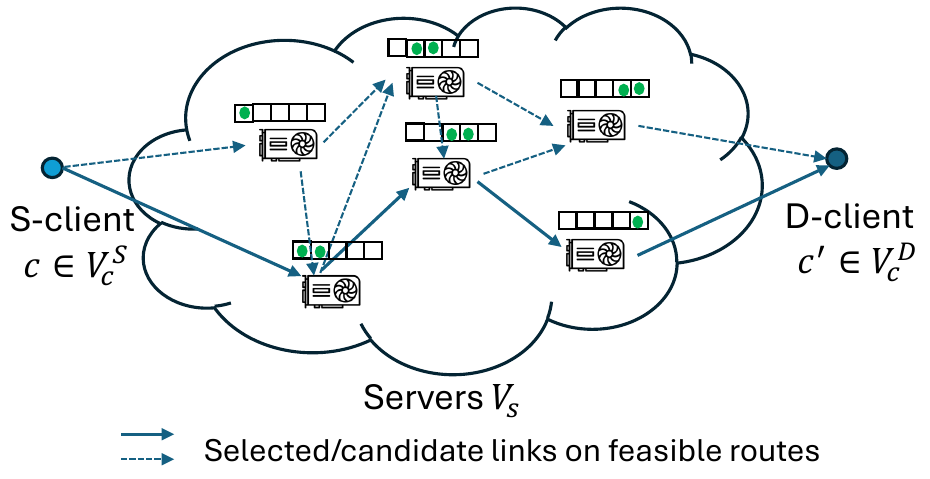}}
   \vspace{-1em}
    \caption{Logical topology $G$ for request routing and route feasibility condition.  
    }
    \label{fig:logical_topology}
    \vspace{-.05em}
\end{figure}

\textbf{Logical topology for routing:} 
To facilitate the optimization formulation, we construct a \emph{logical topology} $G=(V, E)$ as a directed graph illustrated in Fig.~\ref{fig:logical_topology}, where the node set $V := V^S_c\cup V_s\cup V^D_c$ contains the servers $V_s$, the copies of clients $V^S_c$ viewed as sources of requests (called \emph{S-clients}), and the copies of clients $V^D_c$ viewed as destinations of requests (called \emph{D-clients}). We will use $c\in V^S_c$ and $c'\in V^D_c$ to denote the S-client and the D-client corresponding to a real client $c\in V_c$. Define $t_{c c'}:= 0$ and $\tau_{c'}:=0$. The link set $E$ includes the links connecting each S-client with the servers allowed to host the first block, the links between servers allowed to be adjacent on a chain, and the links connecting the servers allowed to host the last block to each D-client. Typically, this leads to full connectivity within $V_s$ and complete bipartite connectivity between $V_s$ and $V^S_c$/$V^D_c$. The split of clients into S-clients and D-clients will facilitate the formulation for request routing.

\textbf{Route feasibility condition:} 
It is easy to see that an optimal block placement should only place \emph{consecutive blocks} at each server, as non-consecutive blocks will increase the number of communications between the client and the server that introduces unnecessary delay. This type of block placement implies that each server can be traversed at most once by each inference session, implying a simple routing path. A \emph{path $p$ in $G$ is feasible} for routing the requests from client $c$ if it leads from the corresponding S-client to the corresponding D-client and traverses a chain of servers that collectively host all the blocks of the LLM in order, as illustrated in Fig.~\ref{fig:logical_topology}. 

Since each server should host consecutive blocks, we can represent the block placement at each server $j\in V_s$ by the \emph{first block $a_j$} and the \emph{number of blocks $m_j$}, i.e., server $j$ stores the blocks in $\{a_j,\ldots,a_j+m_j-1\}$, where\footnote{Throughout this work, we use $[k]$ for a positive integer $k$ to denote the set $\{1,\ldots,k\}$.} $a_j, m_j\in [L]$ and $a_j+m_j-1\leq L$. It is possible for the same block to be placed at multiple servers traversed by an inference session. In this case, we assume that the first server hosting the block will process it according to \cite{Borzunov23petals}. That is, if a session traverses server $i$ immediately before server $j$, the blocks processed at server $j$ will be $\{\max(a_j,a_i+m_i),\ldots,a_j+m_j-1\}$. A chain of servers is feasible for request routing if and only if the block placement satisfies the following condition.  

\begin{lemma}\label{lem:chain feasibility constraint}
Assume that each node $j\in V$ stores all the blocks in $\{a_j,\ldots,a_j+m_j-1\}$, where each S-client $c\in V^S_c$ stores a dummy block $0$ (i.e., $a_c:=0$, $m_c:=1$), and each D-client $c'\in V^D_c$ stores another dummy block $L+1$ (i.e., $a_{c'}:=L+1$, $m_{c'}:=1$). Then a $c$-to-$c'$ path $p$ in $G$ is a feasible routing path for client $c$ if and only if
\begin{align}\label{eq:chain feasibility}
a_j\leq a_i+m_i\leq a_j+m_j-1,~~~\forall (i,j)\in p.
\end{align}
\end{lemma}

\emph{Remark:} Intuitively, the condition \eqref{eq:chain feasibility} states that a path is feasible for request routing if and only if after processing all the blocks at the previous servers, the next block can always be found at the next server on the path.

\textbf{Performance models:}
Lemma~\ref{lem:chain feasibility constraint} allows us to explicitly model how the inference performance depends on block placement and request routing. 
Specifically, according to Section~\ref{subsec:System Model}, Lemma~\ref{lem:chain feasibility constraint} implies that a request from client $c$ routed through a feasible path $p$ traversing link $(i,j)$ will incur a \emph{per-token inference time} of
\begin{align}\label{eq:per-token inference time}
t^c_{ij} := t_{cj} + \tau_j(a_j+m_j-a_i-m_i)
\end{align}
at server $j$ after the first token generation (including both computation and communication), where $a_j+m_j-a_i-m_i$ is the number of blocks that are processed for this request at server $j$. The first token will incur a larger {per-token inference time} of 
$t^{c,I}_{ij}(\lmax^I) := t^I_{cj}(\lmax^I) + \tau^I_j(\lmax^I)(a_j+m_j-a_i-m_i)$
at server $j$. 

Meanwhile, 
by Lemma~\ref{lem:chain feasibility constraint}, each request routed through link $(i,j)\in E$ will require $s_c(a_j+m_j-a_i-m_i)$ of dedicated memory at server $j$ for the attention caches. This together with the memory for storing $m_j$ blocks leads to a total {memory consumption} of
\begin{align}\label{eq:GPU memory}
s_m m_j + s_c \sum_{i:(i,j)\in E}f_{ij}(a_j+m_j-a_i-m_i)
\end{align}
at server $j$, where $f_{ij}$ is the number of requests routed (concurrently) over link $(i,j)$. 


\subsection{Offline Setting}\label{subsec:Offline Setting}

We first consider the offline setting with a given set of requests to understand the  properties of the problem. 

\subsubsection{Vanilla Formulation}

As explained in Section~\ref{subsec:Performance Modeling}, the block placement can be concisely encoded by $\bm{a}:=(a_j)_{j\in V_s}$ and $\bm{m}:=(m_j)_{j\in V_s}$, which jointly determine which paths are feasible for request routing according to Lemma~\ref{lem:chain feasibility constraint}. 
Let $P_c(\bm{a},\bm{m})$ denote the set of feasible routing paths for client $c$ under the block placement $(\bm{a},\bm{m})$. Due to the server-side attention caches, the chain of servers should remain the same for all the tokens of the same sequence. We ensure this by controlling routing at the granularity of requests, represented by $f^r_p\in \{0,1\}$ that indicates the selection of path $p\in P_c(\bm{a},\bm{m})$ for request $r\in \mathcal{R}_c$. Then we formulate the \emph{joint block placement and request routing problem (BPRR)} as 
\begin{subequations}\label{eq:BPRR - direct}
\begin{align}
\min_{\bm{f},\bm{a},\bm{m}} \quad & \sum_{c\in V_c}  \sum_{r\in \mathcal{R}_c} \sum_{p\in P_c(\bm{a},\bm{m})} f^r_p \sum_{(i,j)\in p} t^c_{ij} \label{direct:obj} \\
\mbox{s.t.}\quad & s_m m_j + s_c \sum_{c\in V_c}\sum_{r\in \mathcal{R}_c}\mathop{\sum_{p\in P_c(\bm{a},\bm{m}):}}_{(i,j)\in p, \exists i} \hspace{-1em} f^r_p (a_j+m_j-a_i-m_i)  \leq M_j,~~\forall j\in V_s, \label{direct:memory} \\
& \sum_{p\in P_c(\bm{a},\bm{m})} f^r_p = 1,~~\forall c\in V_c, r\in \mathcal{R}_c, \label{direct:flow conservation} \\
& a_j+m_j-1 \leq L,~~\forall j\in V_s, \label{direct:placement} \\
& f^r_p\in \{0,1\},\: a_j,m_j\in [L], \label{direct:variable constraint}
\end{align}
\end{subequations}
where the objective \eqref{direct:obj} is to minimize the total (and hence the average) per-token inference time over all the requests when ignoring the first token, constraint \eqref{direct:memory} ensures that the requests can be served within the GPU memory at each server (where $\sum_{c\in V_c}\sum_{r\in \mathcal{R}_c}\sum_{p\in P_c(\bm{a},\bm{m}): j\in p} \hspace{-0em} f^r_p$ is the number of requests routed to server $j$), 
constraint \eqref{direct:flow conservation} ensures that each request is routed to a feasible path, and constraints \eqref{direct:placement}--\eqref{direct:variable constraint} ensure the feasibility of the block placement. 

\emph{Remark:} Strictly speaking, to minimize the average time in serving each request, we should minimize
\begin{align}\label{eq:true total inference time}
\sum_{c\in V_c}  \sum_{r\in \mathcal{R}_c} \sum_{p\in P_c(\bm{a},\bm{m})} f^r_p \left( \sum_{(i,j)\in p}t^{c,I}_{ij}(\lmax^I) + (\lmax-1) \sum_{(i,j)\in p} t^c_{ij} \right),
\end{align}
which is the total time to complete all the requests. Minimizing \eqref{eq:true total inference time} is equivalent to minimizing an objective function of the form \eqref{direct:obj}, except that $t^c_{ij}$ needs to be redefined as follows: 
\begin{align}\label{eq:true avg per-token time}
\left({1\over \lmax}t^I_{cj}(\lmax^I)+{\lmax-1\over \lmax}t_{cj}\right) + \left( {1\over \lmax}\tau^I_j(\lmax^I)+{\lmax-1\over \lmax}\tau_j\right) (a_j+m_j-a_i-m_i),
\end{align}
which denotes the average per-token inference time incurred by a request from client $c$ at link $(i,j)$ over \emph{all the tokens}. 
In the case of $\lmax^I\ll\lmax$, \eqref{eq:true avg per-token time} reduces to \eqref{eq:per-token inference time}. Even in the general case, all our results remain applicable when redefining $t^c_{ij}$ as in \eqref{eq:true avg per-token time}. 
We will thus focus on solving \eqref{eq:BPRR - direct}. 
While we have used the maximum input/output length to simplify our formulation, it is easily extensible to the case of heterogeneous input/output lengths as explained in \ref{appendix:Heterogeneous Lengths}. \looseness=-1


\subsubsection{MILP Formulation}\label{subsubsec:MILP Formulation}

The vanilla formulation \eqref{eq:BPRR - direct} is not computationally tractable because of the implicit and nonlinear dependency of the routing variable $\bm{f}$ on the block placement variables $\bm{a}$ and $\bm{m}$. Moreover, given a block placement $(\bm{a},\bm{m})$, the number of feasible paths in $P_c(\bm{a},\bm{m})$ can be exponential. These limitations motivate us to seek a more tractable formulation. 
Below we will show that by introducing appropriate auxiliary variables, we can convert the nonlinear exponential-sized optimization \eqref{eq:BPRR - direct} into a MILP with polynomial numbers of variables and constraints. 

To avoid an exponential number of routing variables, we replace the path-level routing variable $f^r_p$ by a link-level routing variable $f^r_{ij}\in \{0,1\}$, which indicates if request $r$ is routed over link $(i,j)\in E$ (meaning that the request will be processed by server $i$ and server $j$ consecutively). 
This simplifies the objective function \eqref{direct:obj} into
\begin{align}\label{eq:link-level obj}
\sum_{c\in V_c} \sum_{r\in \mathcal{R}_c}\sum_{(i,j)\in E} f^r_{ij} \Big( t_{cj} + \tau_j(a_j+m_j-a_i-m_i)\Big). 
\end{align}
We can guarantee all the requests to be properly routed by imposing a flow conservation constraint 
\begin{align}\label{eq:flow conservation}
\sum_{i\in V} f^r_{ji} = \sum_{i\in V}f^r_{ij} + d^c_j,~~~\forall c\in V_c, r\in \mathcal{R}_c, j\in V, 
\end{align}
where $d^c_j$ is a constant, defined as $1$ if $j=c$ (S-client for client $c$), $-1$ if $j=c'$ (D-client for client $c$), and $0$ otherwise. This ensures that each request will be routed from its S-client to its D-client in the logical topology $G$ (Fig.~\ref{fig:logical_topology}).  

To model the dependency between request routing and block placement, we leverage Lemma~\ref{lem:chain feasibility constraint}, which states that it is feasible to route a request from node $i$ to node $j$ if and only if $a_j\leq a_i+m_i\leq a_j+m_j-1$. Thus, we can ensure route feasibility by requiring
\begin{align}
&\hspace{-1em}a_j f^r_{ij} \leq a_i+m_i, ~~~\forall r\in \mathcal{R}, (i,j)\in E, \label{eq:bilinear - 1} \\
&\hspace{-1em}(a_i+m_i)f^r_{ij} \leq a_j+m_j-1,~\forall r\in \mathcal{R}, (i,j)\in E, \label{eq:bilinear - 2}
\end{align}
where $\mathcal{R}:= \bigcup_{c\in V_c}\mathcal{R}_c$ denotes the total set of requests. This ensures that a request is routed over $(i,j)$ only if the  condition in Lemma~\ref{lem:chain feasibility constraint} is satisfied. \looseness=-1

However, using the above objective and constraints directly will lead to a nonlinear optimization due to the bilinear terms $a_j f^r_{ij}$, $a_i f^r_{ij}$, $m_j f^r_{ij}$, and $m_i f^r_{ij}$ in the objective function \eqref{eq:link-level obj} and constraints like \eqref{eq:bilinear - 1}--\eqref{eq:bilinear - 2}. Fortunately, because $f^r_{ij}$ is binary, we can convert them into linear forms by introducing auxiliary variables $\alpha^r_{ij}, \beta^r_{ij}, \gamma^r_{ij}, \delta^r_{ij}$ and the corresponding linear constraints \eqref{eq:define alpha}\mbox{--}\eqref{eq:define delta} as detailed in \ref{appendix:Linearization}. 

Combining all the above allows us to rewrite the BPRR problem in \eqref{eq:BPRR - direct} as follows:
\begin{subequations}\label{eq:BPRR - MILP}
\begin{align}
\mathop{\min_{\bm{f},\bm{a},\bm{m},}}_{\bm{\alpha},\bm{\beta},\bm{\gamma},\bm{\delta}} & \sum_{c\in V_c} \sum_{r\in \mathcal{R}_c} \sum_{(i,j)\in E} \Big( t_{cj}f^r_{ij} + \tau_j (\alpha^r_{ij}+\gamma^r_{ij} -\beta^r_{ij}-\delta^r_{ij}) \Big) \label{MILP:obj} \\
\mbox{s.t.} \quad & s_m m_j + s_c \sum_{c\in V_c}\sum_{r\in \mathcal{R}_c} \sum_{i: (i,j)\in E} (\alpha^r_{ij}+\gamma^r_{ij}-\beta^r_{ij}-\delta^r_{ij}) \leq M_j,~~\forall j\in V_s, \label{MILP:memory} \\
&\hspace{-1em} \sum_{i\in V}f^r_{ji} = \sum_{i\in V}f^r_{ij} + d^c_j,~\forall c\in V_c, r\in \mathcal{R}_c, j\in V, \label{MILP:flow conservation} \\
&\hspace{-1em} a_j+m_j-1\leq L,~~\forall j\in V_s, \label{MILP:placement} \\
&\hspace{-1em} \alpha^r_{ij}\leq a_i + m_i,~~\forall r\in \mathcal{R}, (i,j)\in E, \label{MILP:chain feasibility 1} \\
&\hspace{-1em} \beta^r_{ij}+\delta^r_{ij}\leq a_j+m_j-1,~\forall r \in \mathcal{R}, (i,j)\in E, \label{MILP:chain feasibility 2} \\
&\hspace{-1em} \eqref{eq:define alpha}\mbox{--}\eqref{eq:define delta}, \label{MILP:auxiliary}\\
&\hspace{-1em} f^r_{ij}\in \{0,1\},\: a_j,m_j\in [L],\: \alpha^r_{ij}, \beta^r_{ij}, \gamma^r_{ij}, \delta^r_{ij}\geq 0, \label{MILP:variable}
\end{align}
\end{subequations}
where \eqref{MILP:obj} is the linearized expression of the total inference time, \eqref{MILP:memory} models the GPU memory capacity, \eqref{MILP:flow conservation} ensures flow conservation, \eqref{MILP:placement} ensures block placement feasibility, \eqref{MILP:chain feasibility 1}--\eqref{MILP:chain feasibility 2} ensure route feasibility, and \eqref{MILP:auxiliary} defines the auxiliary variables used to linearize the objective function and the constraints. 
As explained in \ref{appendix:Linearization}, $\alpha^r_{ij} = a_j f^r_{ij}$, $\beta^r_{ij} = a_i f^r_{ij}$, $\gamma^r_{ij} = m_j f^r_{ij}$, and $\delta^r_{ij} = m_i f^r_{ij}$. Plugging these into \eqref{MILP:obj} and \eqref{MILP:memory} recovers the original objective \eqref{direct:obj} and memory constraint \eqref{direct:memory}, except that the path-level routing variable $f^r_p$ is replaced by the link-level routing variable $f^r_{ij}$. The flow conservation constraint \eqref{MILP:flow conservation} together with the route feasibility constraints \eqref{MILP:chain feasibility 1}--\eqref{MILP:chain feasibility 2} and the integer constraint on $f^r_{ij}$ \eqref{MILP:variable} ensures that each request is routed on a single feasible path as in \eqref{direct:flow conservation}.
The overall optimization is a MILP, where $\bm{a}, \bm{m}$ (block placement) and $\bm{f}$ (request routing) are free variables, and the others are dependent variables. 

\emph{Remark:} 
The MILP \eqref{eq:BPRR - MILP} has $O(|\mathcal{R}|\cdot |E|)$ variables, dominated by the routing variable $\bm{f}$ and the auxiliary variables $\bm{\alpha}$, $\bm{\beta}$, $\bm{\gamma}$, and $\bm{\delta}$. As the number of links $|E|$ is bounded by $O(|V_s|(|V_c|+|V_s|))$, the number of variables is in $O(|\mathcal{R}|\cdot |V_s| (|V_c|+|V_s|))$. Similarly, the number of constraints is also in $O(|\mathcal{R}|\cdot |E|) = O(|\mathcal{R}|\cdot |V_s| (|V_c|+|V_s|))$, dominated by \eqref{MILP:chain feasibility 1}--\eqref{MILP:auxiliary}. Thus, even if BPRR can be formulated as a MILP, the size of the MILP will grow linearly in the number of requests/clients and quadratically in the number of servers, making it challenging to solve.

Formally, we have shown via a reduction from the optimization version of the \emph{partition problem}~\cite{Hayes02AS} that the BPRR problem is hard to solve to optimality as stated below. 

\begin{theorem}\label{thm:NP-hardness of BPRR}
The BPRR problem as formulated in \eqref{eq:BPRR - direct} or \eqref{eq:BPRR - MILP} is NP-hard. 
\end{theorem}

\emph{Remark:} In fact, the reduction in the proof of Theorem~\ref{thm:NP-hardness of BPRR} leads to a special case of BPRR with a single client ($|V_c|=1$). This case represents a common deployment scenario where end users query the system through a proxy (e.g., a Flask web server as in \cite{Borzunov23petals}). We have proved that the BPRR problem remains NP-hard even in this special case. \looseness=-1

\subsubsection{Algorithm Design}\label{subsubsec:Algorithm Design - general case}

The NP-hardness of BPRR implies the need of efficient suboptimal algorithms. Although the MILP formulation in \eqref{eq:BPRR - MILP} allows us to apply existing heuristics for MILP, e.g., LP relaxation plus rounding, such a heuristic still has limited scalability due to the large size of the LP, 
and more importantly, it does not even guarantee a feasible solution\footnote{Specifically, the  fractional solution $(\tilde{\bm{f}},\tilde{\bm{a}},\tilde{\bm{m}})$ to the LP relaxation of \eqref{eq:BPRR - MILP} may not lead to a feasible integer solution after rounding, as the rounded block placement $(\bm{a},\bm{m})$ may not place each block on at least one server. This is because the auxiliary variables $\bm{\alpha},\bm{\beta},\bm{\gamma},\bm{\delta}$ used to linearize the problem can only enforce feasibility when $\bm{f}$ is binary.}. 
To efficiently solve large instances of the BPRR problem, we propose a three-step algorithm called \emph{Conservative Greedy BPRR (CG-BPRR)} as shown in Alg.~\ref{Alg:CG-BPRR}, by  
decomposing the BPRR problem into the three subproblems of optimizing $\bm{m}$, $\bm{a}$, and $\bm{f}$ sequentially, as explained below. 

\begin{algorithm}[tb]
\small
\SetKwInOut{Input}{input}\SetKwInOut{Output}{output}
\Input{set of clients $V_c$, set of requests $\bigcup_{c\in V_c}\mathcal{R}_c$, \#blocks $L$, size per block $s_m$, size per cache $s_c$, set of servers $V_s$, parameters for each $j\in V_s$ including GPU memory $M_j$, processing time $\tau_j$, and per-token RTTs $(t_{cj})_{c\in V_c}$}
\Output{Block placement $(\bm{a},\bm{m})$ and request routing $\bm{f}$}
\tcp{conservative assignment of \#blocks per server:}
$m_j \leftarrow \min(\lfloor M_j/(s_m+s_c |\mathcal{R}|) \rfloor,\: L)$, $\forall j\in V_s$, where $|\mathcal{R}|= \sum_{c\in V_c} |\mathcal{R}_c|$\nl\label{CG:1}
\tcp{greedy block placement:}
$C_b\leftarrow 0,\: T_b\leftarrow \widetilde{t}_0|\mathcal{R}|$, $\forall b\in [L]$\nl\label{CG:2}
\For{each server $j\in V_s$ in increasing order of $\widetilde{t}_j$ \nl\label{CG:3}}
{
\If{$\exists b\in [L]$ with $C_b<|\mathcal{R}|$ \nl\label{CG:4}}
{$a_j \leftarrow \argmax_{a\in [L-m_j+1]:\:C_b<|\mathcal{R}|,\: \exists b\in \{a,\ldots,a+m_j-1\}} \sum_{b'=a}^{a+m_j-1} T_{b'}$\nl\label{CG:5}}
\Else
{$a_j\leftarrow \argmin_{a\in [L-m_j+1]} (C_a,\ldots,C_{a+m_j-1})$\nl\label{CG:7}}
$T_b\leftarrow T_b - (\widetilde{t}_0-\widetilde{t}_j) \min\left(\max(|\mathcal{R}|-C_b, 0), \overline{f}_j\right)$, $\forall b\in \{a_j,\ldots,a_j+m_j-1\}$\nl\label{CG:8}
$C_b\leftarrow C_b+\overline{f}_j$, $\forall b\in \{a_j,\ldots,a_j+m_j-1\}$\nl\label{CG:9}
}
\tcp{shortest-path request routing:}
\For{each client $c\in V_c$ \label{CG:10}}
{
$G^c_{\bm{a},\bm{m}}\leftarrow$ the feasible routing topology for client $c$ under block placement $(\bm{a}, \bm{m})$, with a node/link set $(V^c, E^c_{\bm{a},\bm{m}})$ and a cost of $t^c_{ij}$ for each $(i,j)\in E^c_{\bm{a},\bm{m}}$\nl\label{CG:11}
$p_c\leftarrow$ shortest path from the S-client to the D-client in $G^c_{\bm{a},\bm{m}}$\nl\label{CG:12}
$f^r_{ij}\leftarrow \mathbb{1}((i,j)\in p_c)$\footnotemark, $\forall r\in \mathcal{R}_c,\: (i,j)\in E$\nl\label{CG:13}
}
\caption{Conservative Greedy BPRR (CG-BPRR)}
\vspace{-.0em}
\label{Alg:CG-BPRR}
\end{algorithm}
\normalsize
\footnotetext{We use $\mathbb{1}(\cdot)$ to denote the indicator function.}

\textbf{Step~1:} First, we set the number of blocks per server conservatively (line~\ref{CG:1}) to make sure that each server will have enough remaining GPU memory to hold the attention caches even if all the requests are routed through it. This is because according to the memory consumption model in \eqref{eq:GPU memory}, the memory consumption at server $j$ is upper-bounded by $s_m m_j + s_c |\mathcal{R}| m_j$ (achieved if all the requests are routed to it and all the hosted blocks are processed), and thus storing 
$\min(\lfloor M_j/(s_m+s_c |\mathcal{R}|) \rfloor,\: L)$
blocks on server $j$ will guarantee the maximum memory consumption to be feasible. 

\textbf{Step~2:} Next, we greedily place a set of continuous blocks $\{a_j,\ldots,a_j+m_j-1\}$ at each server $j$ in the descending order of ``server speeds'' so that each placement improves the performance for the worst-performing blocks (lines~\ref{CG:2}--\ref{CG:9}). Specifically, we measure the speed of server $j$ by the \emph{amortized inference time} defined as
\begin{align}\label{eq:amortized inference time}
\widetilde{t}_j := {1\over m_j} \max_{c\in V_c} (t_{cj}+\tau_j m_j) = \tau_j + {t_{*j}\over m_j},
\end{align}
where $t_{*j}:= \max_{c\in V_c} t_{cj}$ is the \emph{maximum per-token RTT} between any client and server $j$. Since $\max_{c\in V_c} (t_{cj}+\tau_j m_j)$ is an upper bound on the per-token inference time incurred at server $j$, $\widetilde{t}_j$ denotes the amortized maximum inference time per block at server $j$, amortized over the blocks hosted by this server. The amortization allows us to compare the speeds of servers with different memory capacities. 

Under such amortization, we consider a \emph{relaxed request routing problem}, where each request is routed among servers on a block-by-block basis and incurs a per-block inference time $\widetilde{t}_j$ at server $j$. It is clear that the relaxed request routing must use up the capacity of a faster server (with a smaller amortized inference time) before going to a slower server for each block. Meanwhile, under $m_j$ computed in Step~1, each server $j$ can guarantee to process up to  
\begin{align}
\overline{f}_j := \left\lfloor {M_j-s_m m_j\over s_c m_j} \right\rfloor \geq |\mathcal{R}|
\end{align}
requests concurrently without running out of memory\footnote{The actual number of requests server $j$ runs concurrently may be larger when not all the $m_j$ blocks are processed at this server for some requests.\looseness=-1}, which defines a \emph{capacity} of the server. It means that under the relaxed request routing, only the fastest server hosting block $b$ will be used to process the block for all the requests. 

Thus, to make the best use of fast servers, we place blocks on servers in the increasing order of $\widetilde{t}_j$, where each server $j$ receives the set of continuous blocks that ``need service the most''. To measure this, we use a variable $C_b$ to track the \emph{total capacity of servers hosting block $b$}, and another variable $T_b$ to track the \emph{total amortized inference time all the requests spend on block $b$} under the optimal relaxed routing. That is, $C_b$ is the sum of $\overline{f}_j$'s over all the servers satisfying $a_j\leq b \leq a_j+m_j-1$ (line~\ref{CG:9}), and $T_b = \widetilde{t}_{j}|\mathcal{R}|$ if server $j$ is the fastest server hosting block $b$. 
To ensure that $T_b$ is well-defined before block $b$ is placed on any server, we introduce a \emph{dummy server $0$} (assuming $0\not\in V_s$) with a large capacity $\overline{f}_0:= |\mathcal{R}|$ and a large amortized inference time $\widetilde{t}_0>\widetilde{t}_j$ ($\forall j\in V_s$), which initially hosts all the blocks. Introducing this dummy server allows us to meaningfully initialize $T_b$ before any real block placement (line~\ref{CG:2}), and its large inference time ensures that $T_b$ (as updated in line~\ref{CG:8}) will be equal to the actual total amortized inference time for block $b$ when $C_b\geq |\mathcal{R}|$. Our idea is to measure the ``need of service'' by the total amortized inference time, i.e., the next fastest server will receive the set of continuous blocks with the maximum $\sum_{b=a_j}^{a_j+m_j-1} T_b$. 
\emph{One subtle point} in applying this idea is that since the relaxed request routing always uses the fastest server for each block, $T_b$ will stop changing after block $b$ is placed, which can cause subsequent servers to be assigned the same blocks over and over again. To avoid such waste of servers, we will only select the blocks according to $\sum_{b=a_j}^{a_j+m_j-1} T_b$ if this set contains at least one unserved block (line~\ref{CG:5}); otherwise, we will select the set of continuous blocks with the minimum capacities (line~\ref{CG:7}), where ``$\argmin$'' selects the first value of $a$ in the lexicographical order of sorted $(C_a,\ldots,C_{a+m_j-1})$. 

The rationale of Step~2 is that it minimizes the average inference time under the relaxed request routing. 

\begin{lemma}\label{lem:minimize inference time under relaxed routing}
The block placement in lines~\ref{CG:2}--\ref{CG:9} of Alg.~\ref{Alg:CG-BPRR} minimizes the average per-token inference time over all the requests under the relaxed request routing, assuming that ``$\argmax$'' in line~\ref{CG:5} breaks ties in favor of the smallest index. 
\end{lemma}

\emph{Remark:} As shown in the proof of Lemma~\ref{lem:minimize inference time under relaxed routing}, Step~2 of CG-BPRR generates an intuitive block placement that uses the ``fastest servers'' to cover all the blocks sequentially. That is, if $j\in V_s$ denotes the $j$-th fastest server in terms of $\widetilde{t}_j$, then $a_k = \sum_{j=1}^{k-1}m_j+1$ for all $k=1,\ldots,K-1$ and $a_K = L-m_K+1$, where $K:=\min\{k:\: \sum_{j=1}^k m_j\geq L\}$.

\textbf{Step~3:} Finally, we compute the request routing under the block placement $(\bm{a}, \bm{m})$ given by the previous steps, assuming the block placement to be feasible (i.e., every block is hosted by at least one server). 

Given a feasible block placement $(\bm{a},\bm{m})$, the (conditionally) optimal request routing that minimizes the average inference time is given by a subproblem of \eqref{eq:BPRR - MILP} as follows
\begin{subequations}\label{eq:RR - given BP}
\begin{align}
\min_{\bm{f}}\quad & \sum_{c\in V_c} \sum_{r\in \mathcal{R}_c}\sum_{(i,j)\in E} t^c_{ij} f^r_{ij} \label{old RR:obj}\\
\mbox{s.t. } & s_c  \sum_{c\in V_c} \sum_{r\in \mathcal{R}_c} \sum_{i:(i,j)\in E} f^r_{ij} (a_j + m_j - a_i - m_i) \leq M_j - s_m m_j,~~\forall j\in V_s, \label{old RR:memory} \\
&\eqref{MILP:flow conservation}, \eqref{eq:bilinear - 1}\mbox{--}\eqref{eq:bilinear - 2}, \\
& f^r_{ij}\in \{0,1\},~~\forall r\in \mathcal{R}, (i,j)\in E,
\end{align}
\end{subequations}
where we have plugged in $t^c_{ij}$ as defined in \eqref{eq:per-token inference time} as a constant. Generally, this is an integer linear programming (ILP) problem with $O(|\mathcal{R}|\cdot |E|)$ variables and $O(|\mathcal{R}|\cdot |E|)$ constraints, which is no easier to solve than \eqref{eq:BPRR - MILP}. However, under the block placement computed by CG-BPRR, we can convert the problem into a simple shortest-path routing problem as stated below. 

\begin{lemma}\label{lem:optimality of shortest path routing - CG-BPRR}
Suppose that the block placement computed by lines~\ref{CG:1}--\ref{CG:9} of Alg.~\ref{Alg:CG-BPRR} is feasible. Then under this block placement, the shortest-path request routing in lines~\ref{CG:10}--\ref{CG:13} is optimal for \eqref{eq:RR - given BP}. 
\end{lemma}

\emph{Remark:} As explained in Section~\ref{subsec:Online Setting}, when applying CG-BPRR, the number of requests $|\mathcal{R}|$ is actually a design parameter that can be tuned to ensure the feasibility of the block placement.

\subsubsection{Performance Analysis}

We now analyze the overall performance of CG-BPRR (Alg.~\ref{Alg:CG-BPRR}) in terms of both the complexity and the average inference time. 

\emph{Complexity:} The complexity can be dominated by either greedy block placement (lines~\ref{CG:2}--\ref{CG:9}) or shortest-path request routing (lines~\ref{CG:10}--\ref{CG:13}). Each \textbf{for} loop in lines~\ref{CG:4}--\ref{CG:9} takes time $O(L\overline{m})$ if line~\ref{CG:5} is executed (where $\overline{m}:= \max_{j\in V_s}m_j$) or $O(L\overline{m}\log{\overline{m}})$ if line~\ref{CG:7} is executed, both in $O(L^2\log{L})$. Each \textbf{for} loop in lines~\ref{CG:11}--\ref{CG:13} takes time $O(|\mathcal{R}_c|\cdot |E|)$, dominated by line~\ref{CG:13}. The overall complexity of Alg.~\ref{Alg:CG-BPRR} is thus $O(|V_s| L^2\log{L} + |\mathcal{R}|\cdot |E|) = O\left(|V_s|\left(L^2\log{L} + |\mathcal{R}|(|V_c|+|V_s|) \right)\right)$ (recall that $|E|=O(|V_s|(|V_c|+|V_s|))$), which is polynomial in the problem size.

\begin{figure}[!t]
   \centerline{\includegraphics[width=0.75\linewidth]{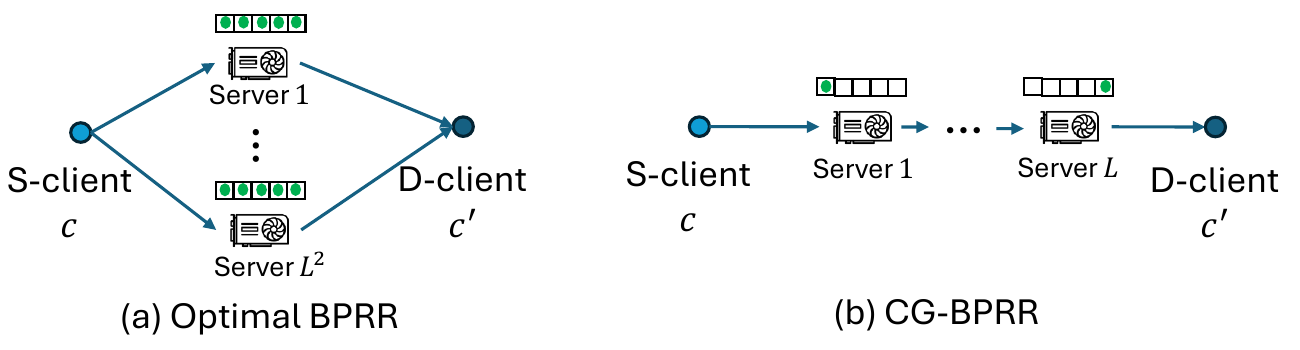}}
   \vspace{-1em}
    \caption{Example for suboptimality of CG-BPRR (Alg.~\ref{Alg:CG-BPRR}). }
    \label{fig:CG_suboptimality}
    \vspace{-.05em}
\end{figure}

\emph{Suboptimality:} First of all, due to its conservative assignment of blocks, CG-BPRR may not give a feasible solution when a feasible solution exists. For example, we may have a large number of servers ($|V_s|\geq L |\mathcal{R}|$) but also a large number of requests, such that $s_m+s_c|\mathcal{R}|>M_j$ and $s_m+s_c\leq M_j$ ($\forall j\in V_s$). Then CG-BPRR will be unable to place any block, but it is still feasible to satisfy all the requests by placing one block per server and routing each request to a disjoint set of $L$ servers. 

Moreover, CG-BPRR can be suboptimal in comparison to the optimal solution even when it is feasible. Consider an example with $|V_c|=1$ client and $|V_s|=L^2$ servers, where each $j\in V_s$ has a per-token RTT of $t_{cj}\equiv t$ to the client, a per-block processing time of $\tau_j\equiv \tau$, and a GPU memory of $M_j\equiv(L+1)s_m$. Suppose that $s_m = L s_c$, and the number of requests is $|\mathcal{R}|=L s_m/s_c = L^2$. Then CG-BPRR will place only one block per server and route each request through a chain of $L$ servers as illustrated in Fig.~\ref{fig:CG_suboptimality}b, incurring an average per-token inference time of $T^g:=L(t+\tau)$. However, the optimal solution is to place all the blocks on each server and route only one request  to each server as illustrated in Fig.~\ref{fig:CG_suboptimality}a, incurring an average per-token inference time of $T^o := t+\tau L < T^g$. 

\emph{Performance guarantee:}
Nevertheless, CG-BPRR provides a guaranteed average inference time as long as its block placement is feasible. \looseness=-1

\begin{theorem}\label{thm:CG-BPRR}
Without loss of generality, assume that the servers have been sorted into an increasing order of amortized inference times (i.e., $\widetilde{t}_1\leq \ldots \leq \widetilde{t}_{|V_s|}$). Let $K:=\min\{k:\: \sum_{j=1}^k m_j\geq L\}$ be the number of iterations of lines~\ref{CG:3}--\ref{CG:9} until all the blocks are placed on at least one server. If CG-BPRR gives a feasible solution (i.e., $K<\infty$), then its average per-token inference time $T^g$ is bounded by
    \begin{align}\label{eq:CG-BPRR bound}
T^g \leq \sum_{j=1}^K \widetilde{t}_j m_j - \tau_K \left(\sum_{j=1}^K m_j -L\right),
    \end{align}
where $\widetilde{t}_j$ is defined as in \eqref{eq:amortized inference time} and $m_j$ is computed as in line~\ref{CG:1} of Alg.~\ref{Alg:CG-BPRR}.  
\end{theorem}

\emph{Remark:} The upper bound in \eqref{eq:CG-BPRR bound} has an intuitive meaning that it is the per-token inference time under the \emph{worst-case demand}, where all the requests come from a client $c$ that simultaneously achieves $t_{cj} = t_{*j}$ for all $j\in V_s$ (i.e., farthest away from all the servers), and are routed through the chain of servers $1,\ldots,K$. This worst-case guarantee will facilitate the adaptation of CG-BPRR for the online setting as explained below.  \looseness=0

In addition to the upper bound in \eqref{eq:CG-BPRR bound}, we also provide a lower bound on the minimum average per-token inference time in \ref{appendix:CG-BPRR Analysis}, which together with the upper bound yields a bounded approximation ratio for CG-BPRR.

\subsection{Online Setting}\label{subsec:Online Setting}

In practice, requests usually arrive dynamically, and block placement and request routing are usually performed at different time scales: block placement should be at a \emph{large time scale} to avoid excessive overhead due to frequently loading/unloading blocks, and request routing should be at a \emph{small time scale} to provide timely service to dynamically arriving requests. Nevertheless, the key ideas in CG-BPRR can be applied in the online setting through a two-time-scale solution as follows.

\subsubsection{Block Placement via Robust Optimization}

Due to its large overhead, one block placement solution should be able to serve a variety of request scenarios with reasonable performance. We address this requirement by treating the block placement problem as a \emph{robust optimization problem}, with the goal of minimizing the \emph{worst-case} average per-token inference time for a target range of scenarios. To this end, we note that the block placement steps in CG-BPRR (i.e., lines~\ref{CG:1}--\ref{CG:9}), referred to as \emph{Conservative Greedy Block Placement (CG-BP)}, are already solving a robust optimization. Specifically, given a maximum number of concurrent requests $|\mathcal{R}|$ that the system aims to support, CG-BP places blocks according to the worst case that all the requests are from a client that is farthest away from all the servers (i.e., with a per-token RTT of $t_{*j},\: \forall j\in V_s$). Thus, the performance bound in Theorem~\ref{thm:CG-BPRR} still applies as stated below. 

\begin{corollary}\label{cor:CG-BPRR}
If CG-BP (i.e., lines~\ref{CG:1}--\ref{CG:9} of Alg.~\ref{Alg:CG-BPRR}) gives a feasible block placement for $|\mathcal{R}|$ requests, then during online inference, the average per-token inference time under the optimal request routing will always be bounded by \eqref{eq:CG-BPRR bound} as long as the number of concurrent requests\footnote{In the online setting, ``concurrent requests'' refer to requests that are simultaneously active (i.e., submitted but unfinished) but do not necessarily arrive simultaneously.} is no more than $|\mathcal{R}|$.
\end{corollary}

\emph{Remark:} 
The parameter $|\mathcal{R}|$ here is a design parameter that denotes the number of requests the system \emph{guarantees} to serve concurrently and can be tuned to ensure the feasibility of CG-BP. 
Specifically, given the set of servers and their memory capacities, it is easy to see that the block placement given by lines~\ref{CG:1}--\ref{CG:9} of Alg.~\ref{Alg:CG-BPRR} is feasible if and only if \looseness=-1
\begin{align}\label{eq:CG-BP feasibility,iff}
 \sum_{j\in V_s} \min(\lfloor M_j/(s_m+s_c |\mathcal{R}|) \rfloor,\: L) \geq L.
\end{align}
To satisfy the condition \eqref{eq:CG-BP feasibility,iff}, it suffices for $|\mathcal{R}|$ to satisfy
\begin{align}\label{eq:|R| upper bound}
|\mathcal{R}| \leq \left\lfloor {\sum_{j\in V_s}M_j-s_m(L+|V_s|)\over s_c(L+|V_s|)} \right\rfloor,
\end{align}
which provides an upper bound on the number of requests that the system can guarantee to serve concurrently under CG-BP. 
Meanwhile, increasing $|\mathcal{R}|$ has a negative effect of decreasing $m_j$ ($\forall j\in V_s$), which can increase the path length and hence the inference time. Thus, the parameter $|\mathcal{R}|$ effectively controls the ``throughput-delay tradeoff'' of the system. When the number of concurrent requests exceeds $|\mathcal{R}|$, the new request may have to wait for some of the existing requests to finish (as detailed in Section~\ref{subsubsec:Online Request Routing}). Therefore, $|\mathcal{R}|$ could be tuned to optimize the tradeoff between the waiting time and the inference time after waiting. 
A configuration method that empirically works well is as follows: (i) estimate the mean and the standard deviation (std) of the number of new requests during the processing of a given request, and (ii) set $|\mathcal{R}|$ to the minimum of the mean plus the std and the upper bound in \eqref{eq:|R| upper bound}. 
Such estimation can be further adjusted, on a time scale suitable for block placement, to accommodate time-varying demands.

\subsubsection{Request Routing via Individually Optimal Scheduling}\label{subsubsec:Online Request Routing}

Given the block placement by CG-BP, online request routing strives to schedule each incoming request to a feasible path to complete the request as soon as possible. If the number of concurrent requests is within $|\mathcal{R}|$ (the number of concurrent requests CG-BP plans for), then there will be no memory contention, and the new request can be routed optimally as in lines~\ref{CG:10}--\ref{CG:13} of Alg.~\ref{Alg:CG-BPRR}. If the number of concurrent requests exceeds $|\mathcal{R}|$, however, then the new request may have to wait for some existing requests to finish, in which case it is necessary to know the system state in handling the existing requests. 

To this end, we track the state of each server as follows. Let $(T^j_r(t), M^j_r(t))_{r=1}^{R_j(t)}$ denote the \emph{state of server $j$ at time $t$}, where $R_j(t)$ is the number of existing (and unfinished) requests routed through server $j$, $T^j_r(t)$ is the remaining time\footnote{We can estimate the remaining time for an existing request $r$ using its completion time by \eqref{Online RR:obj} when scheduling $r$ minus the elapsed time. 
} for a request $r\in \mathcal{R}_j(t)$, and $M^j_r(t)$ is the number of attention caches hosted on server $j$ for request $r$ (i.e., \#blocks $j$ processes for $r$). Without loss of generality, suppose that the existing requests have been sorted into increasing order of $T^j_r(t)$. Then for a new request arriving at time $t$, the \emph{waiting time} for link $(i,j)\in E$ to be available for routing the request (i.e., the time until server $j$ has enough cache space for a new request routed through $(i,j)$) is given by
\begin{align}\label{eq:waiting time on (i,j)}
t^W_{ij}(t) := \min \left\{T^j_k(t):\: \left\lfloor{M_j-s_m m_j\over s_c} \right\rfloor -\sum_{r=k+1}^{R_j(t)} M^j_r(t) \geq a_j+m_j-a_i-m_i \right\},
\end{align}
where $T^j_0(t) := 0$. That is, $t+t^W_{ij}(t)$ is the earliest time for server $j$ to have enough memory for the new request if it is routed from server $i$. The earliest time a path $p$ is available for the request is then given by $t+\max_{(i,j)\in p} t^W_{ij}(t)$. 

Given the current system state at time $t$, we formulate the \emph{individually optimal} (i.e., myopic) scheduling of a new request arriving from client $c$ at time $t$ as 
\begin{subequations}\label{eq:online RR}
\begin{align}
\min_{\bm{f}, t^W}\quad & t^W + \lmax \sum_{(i,j)\in E^c_{\bm{a},\bm{m}}} t^c_{ij}f_{ij} \label{Online RR:obj} \\
\mbox{s.t.} \quad & t^W_{ij}(t) f_{ij} \leq t^W,~~\forall (i,j)\in E^c_{\bm{a},\bm{m}}, \label{Online RR:t^W} \\
& \sum_{i\in \mathcal{N}^+(j; \bm{a},\bm{m})} f_{ji}-\sum_{i\in \mathcal{N}^-(j;\bm{a},\bm{m})}f_{ij} = d^c_j,~~\forall j\in V^c, \label{Online RR:flow conservation} \\
& f_{ij}\in \{0,1\},~~~\forall (i,j)\in E^c_{\bm{a},\bm{m}},
\end{align}
\end{subequations}
where the main variable is $f_{ij}\in \{0,1\}$ that indicates whether to route the request through link $(i,j)$. The objective \eqref{Online RR:obj} is to minimize \emph{the total waiting plus inference time} for this request\footnote{In PETALS implementation~\cite{PetalsGitHub}, there is another delay due to retrying inference at each server according to binary exponential backoff with a maximum backoff time (by default 60 s), causing the actual time from the request arrival to the start of inference to be nonlinear and mildly larger than $t^W$. We have ignored the retry delay to linearize the formulation \eqref{eq:online RR} but considered such delay in our simulation (see Section~\ref{subsec:Evaluation Setup}). 
}, \eqref{Online RR:t^W} ensures that all the invoked servers have enough memory when the session starts, and \eqref{Online RR:flow conservation} is the flow conservation constraint, with $\mathcal{N}^+(j; \bm{a},\bm{m})$ and $\mathcal{N}^-(j;\bm{a},\bm{m})$ denoting the outgoing/incoming neighbors of $j$ in the feasible routing topology $G^c_{\bm{a},\bm{m}}$ for client $c$ under the given block placement $(\bm{a},\bm{m})$, and $d^c_j$ denoting the constant defined in \eqref{eq:flow conservation}. 
While the second term of \eqref{Online RR:obj} has approximated the total inference time by \#output tokens multiplied by the time between tokens, it can be refined to model the actual total inference time by redefining $t^c_{ij}$ as in \eqref{eq:true avg per-token time}. 

While \eqref{eq:online RR} is a MILP that is hard to solve directly, we can simplify it by relaxing $t^W$ into its upper bound $\sum_{(i,j)\in E^c_{\bm{a},\bm{m}}} t^W_{ij}(t) f_{ij}$, which reduces \eqref{eq:online RR} into a shortest-path routing problem with \emph{waiting-penalized link cost} $t^W_{ij}(t) + \lmax t^c_{ij}$ for each $(i,j)\in E^c_{\bm{a},\bm{m}}$. We refer to this solution as \emph{Waiting-penalized Shortest-path Request Routing (WS-RR)}.  

\begin{corollary}\label{cor:Online RR}
Let $p_c(t)$ denote the path selected by WS-RR for client $c$ at time $t$. Then the  cost of this path $\sum_{(i,j)\in p_c(t)} \left(t^W_{ij}(t) + \lmax t^c_{ij}\right)$ is an upper bound on the completion time of a request $r^*$ arriving from client $c$ at time $t$. Moreover, if the number of concurrent requests at time $t$ is within $|\mathcal{R}|$ (the number of requests used by CG-BP in computing the block placement), then $p_c(t)$ is optimal under the given block placement. 
\end{corollary}

\subsubsection{Summary of Online BPRR}

 Alg.~\ref{Alg:Online BPRR} in \ref{appendix:Online BPRR} summarizes our proposed two-time-scale solution for the online setting that combines CG-BP and WS-RR. Together, Corollaries~\ref{cor:CG-BPRR} and \ref{cor:Online RR} imply that this combined solution 
can guarantee a request completion time that is bounded by 
\begin{align}
\left\{ \begin{array}{ll}
\lmax \left(\sum_{j=1}^K \widetilde{t}_j m_j - \tau_K \left(\sum_{j=1}^K m_j -L\right)\right) & \mbox{if \#concurrent requests }\leq |\mathcal{R}|, \\
\sum_{(i,j)\in p_c(t)} \left(t^W_{ij}(t) + \lmax t^c_{ij}\right) & \mbox{o.w.,}
\end{array}\right.
\end{align}
where the first bound is independent of the system state while the second bound is dependent. 
We note that while individually optimal scheduling may not be globally optimal, in our solution it is just used to provide a feasible solution in the exceptional case when \#concurrent requests $>|\mathcal{R}|$ and thus suffices under properly set $|\mathcal{R}|$. 

\section{Performance Evaluation}\label{sec:Performance Evaluation}

We evaluate the performance of the proposed algorithms against the state of the art through both controlled experiments based on the PETALS system~\cite{Borzunov23NeurIPS} and data-driven simulations based on a simulator we have developed that has been cross-validated with the experiment results. 

\subsection{Evaluation Setup}\label{subsec:Evaluation Setup}

\noindent\emph{\bf Evaluation environment:} We employ two cross-validated evaluation environments: \\
    \noindent 1) \emph{PETALS-based distributed system:} This environment runs real LLM inference workloads on a modified version of PETALS~\cite{Borzunov23NeurIPS} that can take any block placement and request routing decisions as inputs. {We deploy two copies of this system: (i) a \emph{smaller deployment} on a server with 3 A100 (80 GB) GPUs, where we leverage the multi-instance GPU technology~\cite{NvidiaMIG} to partition one of the A100 GPUs into 7 smaller virtual GPUs (referred to as MIGs) that together with the remaining A100s provide 9 servers with one GPU each, and (ii) a \emph{larger deployment} on a server with 8 A100 (80 GB) GPUs, where we partition three of the A100s into 21 MIGs that together with the remaining A100s provide 26 servers with one GPU each.} We use the CPUs to emulate the clients. We use the namespace and the traffic control features of Linux~\cite{schubert2019network} to simulate network latency and bandwidth between the nodes as specified in ``System configuration'' below. \\    
    \noindent 2) \emph{MATLAB-based simulator:} To overcome the limited scale of the experiments due to limited GPU resources, we develop a customized simulator in MATLAB that implements the block placement and request routing logic of both the original algorithm in PETALS~\cite{Borzunov23NeurIPS} and the proposed algorithm. The simulator has been engineered to replicate the decision of the real system under the same system state, and validated to approximate the overall experiment results over multiple requests (see Section~\ref{subsec:Experiment Results and Simulator Validation}). We have open-sourced our simulator code\footnote{Our simulator code is available at: \href{https://github.com/TingyangSunJeff/LLM_inference_simulator/tree/main}{https://github.com/TingyangSunJeff/LLM\_inference\_simulator/tree/main}.
    } to facilitate future research on distributed LLM inference for researchers with limited GPU access..  

\begin{table}[]
\small
    \centering
    \begin{tabular}{c|c|c|c}
         &  Cluster0 (CPU) & Cluster1 (2 A100s) & Cluster2 (7 MIGs)  \\
         \hline
Cluster0 (CPU) & 5 ms, 1 Gbit/s   & 100 ms, 100 Mbit/s  & 100 ms, 100 Mbit/s   \\       
Cluster1 (2 A100s) & 100 ms, 100 Mbit/s & 5 ms, 1 Gbit/s & 100 ms, 100 Mbit/s   \\    
Cluster2 (7 MIGs) & 100 ms, 100 Mbit/s & 100 ms, 100 Mbit/s  & 5 ms, 1 Gbit/s 
    \end{tabular}
     \vspace{-0.5em}
    \caption{{System configuration for experimentation (entries denote RTT and bandwidth).}
    }
    \label{tab: network_link_property}
    \vspace{-.0em}
\end{table} 

\noindent\emph{\bf System configuration:} 
We employ BLOOM-176B~\cite{BigScience23BLOOM} as the LLM, which is one of the largest open-source LLMs, and configure our evaluation environment to mimic a clustered or scattered deployment scenario as follows.    

 For the \emph{clustered scenario}, we configure the smaller-scale deployment into three clusters with heterogeneous hardware capabilities: 
    \begin{itemize}
        \item Cluster0: a cluster containing clients remote to all the servers; 
        \item Cluster1: a cluster containing two high-performance servers represented by the A100 GPUs as well as clients local to these servers;
        \item Cluster2: a cluster containing seven low-performance servers represented by the MIGs as well as clients local to these servers.
    \end{itemize}
We add the servers to the system in a random order, and follow the prior work~\cite{Borzunov23NeurIPS} in configuring the networking parameters for intra/inter-cluster communications as presented in Table~\ref{tab: network_link_property}. \looseness=-1

\begin{table}[]
\small
    \centering
    \begin{tabular}{c|c|c|c}
         &   AboveNet  & BellCanada & GTS-CE\\
         \hline
 \#nodes &   23 & 48  & 149\\
 \#links &   62  & 130 & 386 \\
 link capacities (Gbps) & 1  & 1 & 1 \\
 link delays (ms) & $[0.100, 13.800]$   & $[0.078, 6.160]$ & $[0.005,1.081]$ 
    \end{tabular}
    \vspace{-.5em}
    \caption{Topologies used in simulation.     
    }
    \label{tab:topo_info}
    \vspace{-.0em}
\end{table}

For the \emph{scattered scenario}, we simulate three topologies of different sizes from the Internet Topology Zoo~\cite{knight2011internet} with link capacities and delays from \cite{gay2017repetita}, as presented in Table~\ref{tab:topo_info}. We compute the RTTs between nodes according to the cumulative delays along the delay-based shortest paths. In each simulated network, we randomly select $C$ nodes as the locations of servers, $\eta$ fraction of which (randomly selected) are equipped with high-performance servers represented by A100s and the rest are equipped with MIGs. In the experiments, we use AboveNet for the smaller-scale deployment and BellCanada/GTS-CE for the larger-scale deployment, where $C$ and $\eta$ are set based on \#GPUs of each type in each deployment; in the simulations, we vary $C$ and $\eta$ to evaluate their impacts. 

We generate $N_R$ requests {according to a Poisson process with rate $\lambda$} from one of the clusters (in clustered scenario) or a randomly selected node not hosting any server (in scattered scenario), which represents a proxy that queries the system on behalf of end users (e.g., a Flask web server as in \cite{Borzunov23petals}). 
Unless stated otherwise, we set the design parameter $|\mathcal{R}|$ as discussed after Corollary~\ref{cor:CG-BPRR}.
Our experiment results are averaged over 5 Monte Carlo runs and our simulation results are averaged over $20$ Monte Carlo runs. 

\noindent\emph{\bf Benchmark:}
We compare the proposed two-time-scale algorithm (Alg.~\ref{Alg:Online BPRR}) with the original block placement and request routing algorithm in PETALS~\cite{Borzunov23NeurIPS}, where block placement is performed sequentially by letting each newly-added server choose a consecutive range of the most under-served blocks as measured by a heuristic throughput metric, and request routing is performed by letting each client build a graph with heuristic edge weights based on network latency and processing time\footnote{Unlike $G^c_{\bm{a},\bm{m}}(t)$ in Alg.~\ref{Alg:Online BPRR} that uses our validated performance models to compute edge weights, the routing algorithm in PETALS uses heuristic weights, which causes the sum weight along a path to differ from the actual inference time; see \cite{PetalsGitHub} for details.}, and then use a Dijkstra-like algorithm to find the shortest path that traverses the required blocks in order. 
We set the target number of concurrent requests to the expected number of arrivals during the first inference session plus one standard deviation.  

\noindent\emph{\bf Metrics:} Our primary performance metric is the average inference time per token for \emph{all the tokens}, excluding the local processing times at the client (as they are not affected by the resource allocation at servers). 
In addition, we also measure the average inference times for \emph{the first token} and \emph{each of the remaining tokens}, respectively, for additional insights, as well as the average running time of each algorithm. 

\subsection{Experiment Results and Simulator Validation}\label{subsec:Experiment Results and Simulator Validation}

\begin{table}[]
\centering
\renewcommand{\arraystretch}{0.95} 
\setlength{\tabcolsep}{4pt} 
\small 

\vspace{3pt}

\resizebox{\textwidth}{!}{
\begin{tabular}{
>{\centering\arraybackslash}m{2cm} 
>{\centering\arraybackslash}m{2cm} 
>{\centering\arraybackslash}m{2cm} 
>{\centering\arraybackslash}m{2cm} 
>{\centering\arraybackslash}m{2cm} 
>{\centering\arraybackslash}m{2cm} 
}
\toprule
\multirow{2}{*}{\textbf{Client Location}} 
& \multirow{2}{*}{\textbf{Algorithm}} & \multicolumn{2}{c}{\textbf{0.1 requests/s}} & \multicolumn{2}{c}{\textbf{0.5 requests/s}} \\
\cmidrule(lr){3-4} \cmidrule(lr){5-6}
& & \textbf{$\lmax=64$} & \textbf{$\lmax=128$} & \textbf{$\lmax=64$} & \textbf{$\lmax=128$} \\
\midrule
\multirow{2}{*}{Cluster0} & PETALS & $6.23\ (5.33)$ & $4.76\ (4.74)$ & $6.28\ (5.33)$ & $5.14\ (4.74)$ \\
& Proposed                         & $1.92\ (1.59)$ & $1.43\ (0.92)$ & $2.00\ (1.59)$ & $1.34\ (0.92)$ \\
\midrule
\multirow{2}{*}{Cluster1} & PETALS & $5.44\ (5.17)$ & $4.60\ (4.58)$ & $5.56\ (5.17)$ & $4.79\ (4.58)$ \\
& Proposed                         & $1.78\ (1.65)$ & $1.04\ (0.83)$ & $1.88\ (1.65)$ & $1.11\ (0.83)$ \\
\midrule
\multirow{2}{*}{Cluster2} & PETALS & $5.30\ (4.85)$ & $4.85\ (4.07)$ & $5.34\ (4.85)$ & $5.25\ (4.07)$ \\
& Proposed                         & $1.79\ (1.59)$ & $1.31\ (0.92)$ & $1.94\ (1.59)$ & $1.37\ (0.92)$ \\
\bottomrule
\end{tabular}
}
\vspace{-0.5em}
\caption{{Average per-token inference time (s) under the  configuration in Table~\ref{tab: network_link_property} (\textit{$\lmax^I=20$}; 100 requests; MATLAB results shown in parentheses). }}
\label{tab: time_per_token_clustered}
\vspace{-.0em}
\end{table}

\subsubsection{Inference Time}

\textbf{Clustered scenario:} Tables~\ref{tab: time_per_token_clustered}, \ref{tab: TTFT_clustered}, and \ref{tab: time_per_remaining_token_clustered} present the average inference times for all the tokens, the first token in each sequence, and the remaining tokens, respectively, 
under various client locations (Cluster 0, 1, or 2), BPRR algorithms (PETALS~\cite{Borzunov23NeurIPS} or the proposed), request rates, and sequence lengths. The results show that: (i) the proposed algorithm (Alg.~\ref{Alg:Online BPRR}) achieves significantly ($60$--$70\%+$) smaller average inference times compared to the original algorithm in PETALS in all the tested cases, and (ii) our MATLAB simulator produces results that are roughly consistent with the actual experiments. A closer look shows that the performance improvement mainly comes from the time for the first token (Table~\ref{tab: TTFT_clustered}), for which our algorithm achieves an order-of-magnitude reduction. Although the improvement in the inference time for the remaining tokens is smaller (Table~\ref{tab: time_per_remaining_token_clustered}), the overall improvement is dominated by the time reduction for the first token. 
Comparing the results across different settings, we further see that: (i) increasing the output length can reduce the average per-token time due to amortizing the first token's inference time over more tokens, and 
(ii) the client location can affect the performance, e.g., a client in Cluster1 has relatively smaller inference times than a client in Cluster0 under the same algorithm due to its proximity to high-performance servers, but the difference is small as the communication time is only a small fraction of the total inference time (e.g., sending one embedding for BLOOM-176B across two clusters only takes about $0.12$ s). \looseness=-1

\begin{table}[]
\centering
\renewcommand{\arraystretch}{0.95} 
\setlength{\tabcolsep}{4pt}      
\small                            

\vspace{3pt}

\resizebox{\textwidth}{!}{%
\begin{tabular}{
  >{\centering\arraybackslash}m{2cm} 
  >{\centering\arraybackslash}m{2cm} 
  >{\centering\arraybackslash}m{2cm} 
  >{\centering\arraybackslash}m{2cm} 
  >{\centering\arraybackslash}m{2cm} 
  >{\centering\arraybackslash}m{2cm} 
} 
\toprule
\multirow{2}{*}{\textbf{Topology}} 
  & \multirow{2}{*}{\textbf{Algorithm}} 
    & \multicolumn{2}{c}{\textbf{0.1 requests/s}} 
    & \multicolumn{2}{c}{\textbf{0.5 requests/s}} \\ 
\cmidrule(lr){3-6}
  & & \textbf{$\lmax=64$} & \textbf{$\lmax=128$} 
    & \textbf{$\lmax=64$} & \textbf{$\lmax=128$} \\
\midrule
\multirow{2}{*}{AboveNet} 
  & PETALS    & $4.98\ (4.75)$ & $4.03\ (3.88)$  & $5.26\ (5.11)$ & $4.58\ (4.10)$ \\
  & Proposed  & $1.86\ (1.63)$ & $1.44\ (1.36)$ & $1.97\ (1.83)$ & $1.35\ (1.05)$ \\
\midrule
\multirow{2}{*}{BellCanada} 
  & PETALS    & $6.31\ (6.03)$ & $3.82\ (3.49)$ & $6.74\ (6.19)$  & $4.16\ (3.41)$  \\
  & Proposed  & $1.33\ (1.41)$ & $1.26\ (0.92)$  & $1.49\ (1.41)$  & $1.11\ (0.92)$  \\
\midrule
\multirow{2}{*}{GTS-CE} 
  & PETALS    & $7.05\ (6.12)$  & $4.69\ (3.47)$ & $6.89\ (5.97)$  & $4.89\ (3.37)$ \\
  & Proposed  & $1.38\ (1.41)$ & $0.95\ (0.91)$ & $1.35\ (1.40)$ & $1.07\ (0.91)$  \\
\bottomrule
\end{tabular}%
}
\vspace{-0.5em}
\caption{{Average per‐token inference time (s) under the topologies in Table~\ref{tab:topo_info} ($\lmax^I = 20$; 100 requests; MATLAB results shown in parentheses).}}
\label{tab:time_per_token_scattered}
\vspace{-.0em}
\end{table}

\textbf{Scattered scenario:} Tables~\ref{tab:time_per_token_scattered}, \ref{tab:TTFT_scattered}, and \ref{tab:time_per_remaining_token_scattered} present the corresponding results under the network topologies in Table~\ref{tab:topo_info}. The results show qualitatively similar observations as before, e.g., the proposed algorithm significantly reduces the inference times in a diverse set of cases with different numbers of servers, different topologies, and different loads (controlled by request rate and sequence length). This validates the generalizability of our observations. We also see bigger improvements for the larger networks (e.g., $60$--$70\%$ improvement for AboveNet and around $80\%$ improvement for GTS-CE), suggesting promising performance improvement for large deployments. 

\emph{Remark:} We have examined the specific block placement and request routing decisions made by the two algorithms to understand the causes of the observed performance difference. The primary cause turns out to be the difference in how the GPU memory is allocated between model blocks and attention caches: PETALS uses a fixed allocation of attention cache space without considering concurrent sessions~\cite{PetalsGitHub}, which causes it to frequently run out of memory and incur waiting times for incoming requests; in contrast, our solution is designed to handle a certain number of concurrent sessions, which allows it to avoid waiting when properly configured, leading to the significant reduction in the time for the first token. These different memory allocations manifest in the number of blocks placed at each server, where PETALS places 53 blocks on A100 and 4 blocks on MIG, whereas our algorithm (CG-BP) only places 41 blocks on A100 and 3 blocks on MIG. 
We also note that although our simulator implements the same decision making logic as the real system, it cannot perfectly predict the time for executing the decisions due to complex runtime factors (e.g., memory fragmentation and low-level scheduling) not captured by our system model, which is the cause of the mildly different results. Nevertheless, the simulated performance is overall predictive of the actual performance observed in our experiments.

\subsubsection{Algorithm Running Time}

\begin{table}[]
\centering
\renewcommand{\arraystretch}{0.95} 
\setlength{\tabcolsep}{4pt} 
\scriptsize  
\vspace{-0.5em}

\resizebox{0.8\textwidth}{!}{%
\begin{tabular}{
 >{\centering\arraybackslash}m{2.5cm}
 >{\centering\arraybackslash}m{3cm}
 >{\centering\arraybackslash}m{3cm}
}
\toprule
\textbf{Scenario} & \textbf{PETALS} & \textbf{Proposed} \\
\midrule
{Clustered}  & $0.0186 \pm 0.0013$ & $0.0216 \pm 0.0004$ \\
{AboveNet}   & $0.0190 \pm 0.0081$ & $0.0333 \pm 0.0128$ \\
{BellCanada} & $0.0291 \pm 0.0011$ & $0.0287 \pm 0.0018$ \\
{GTS-CE}      & $0.0350 \pm 0.0020$ & $0.0320 \pm 0.0012$ \\
\bottomrule
\end{tabular}
} 
\vspace{-0.5em}
\caption{Algorithm running time (s) (algorithm running time is independent of client location, request rate, and sequence length).}
\label{tab:running_time}
\vspace{-0.0em}
\end{table}

Table~\ref{tab:running_time} shows the average running time of each algorithm evaluated in Tables~\ref{tab: time_per_token_clustered}--\ref{tab:time_per_token_scattered}. To ensure a fair comparison and separate the decision making time from the time of implementing the decisions, we evaluate the running times of both our algorithm and the original algorithm in PETALS~\cite{Borzunov23NeurIPS} based on their MATLAB implementation. Both solutions are fast enough so that the decision making time is negligible compared to the actual inference time. 

\subsection{Results of Experimentally-validated Simulations}

\begin{figure}[t!]
\centering

\begin{subfigure}[t]{0.31\textwidth}
  \centering
  \includegraphics[width=\textwidth]{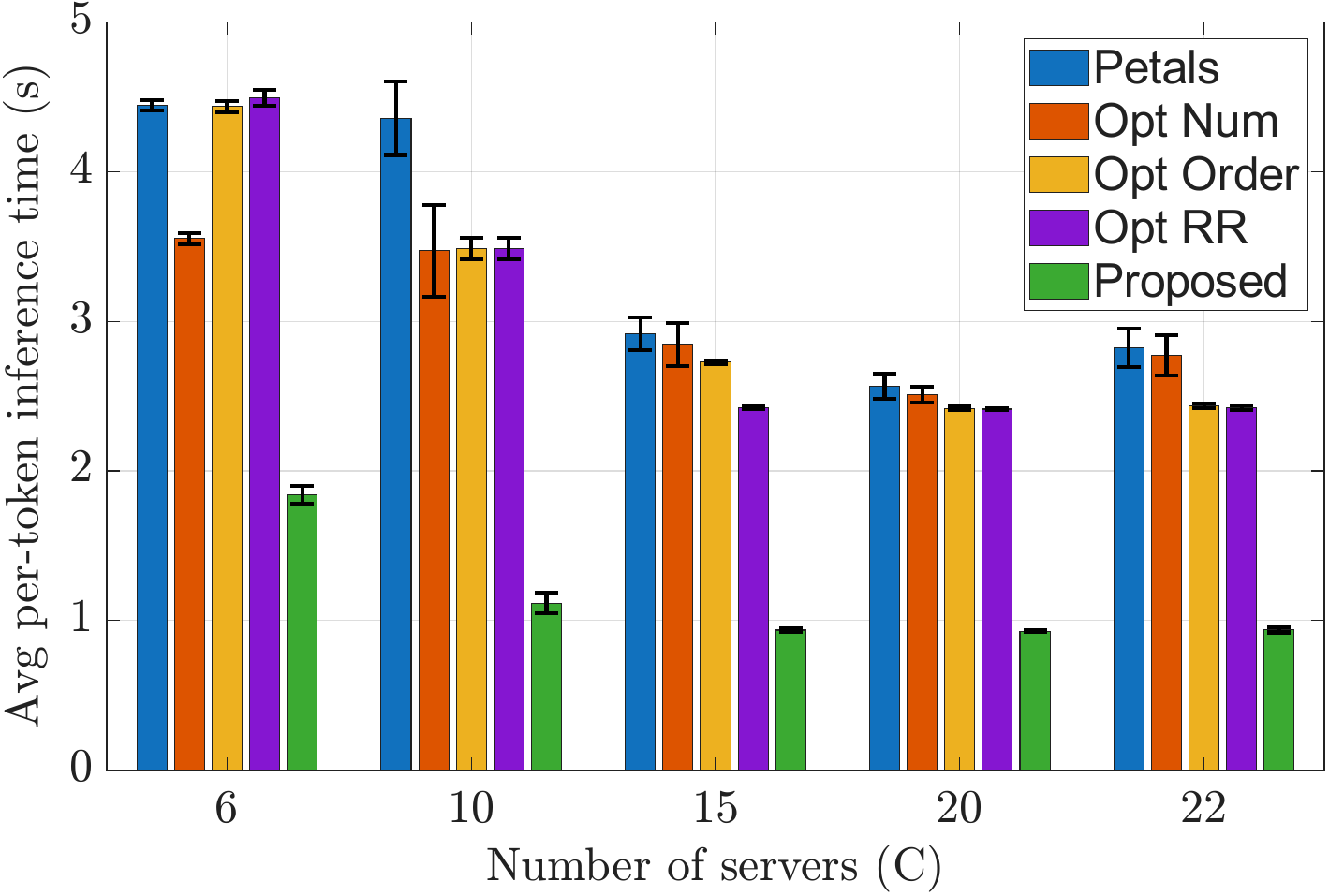}
  \caption{AboveNet}
\end{subfigure}
\begin{subfigure}[t]{0.31\textwidth}
  \centering
  \includegraphics[width=\textwidth]{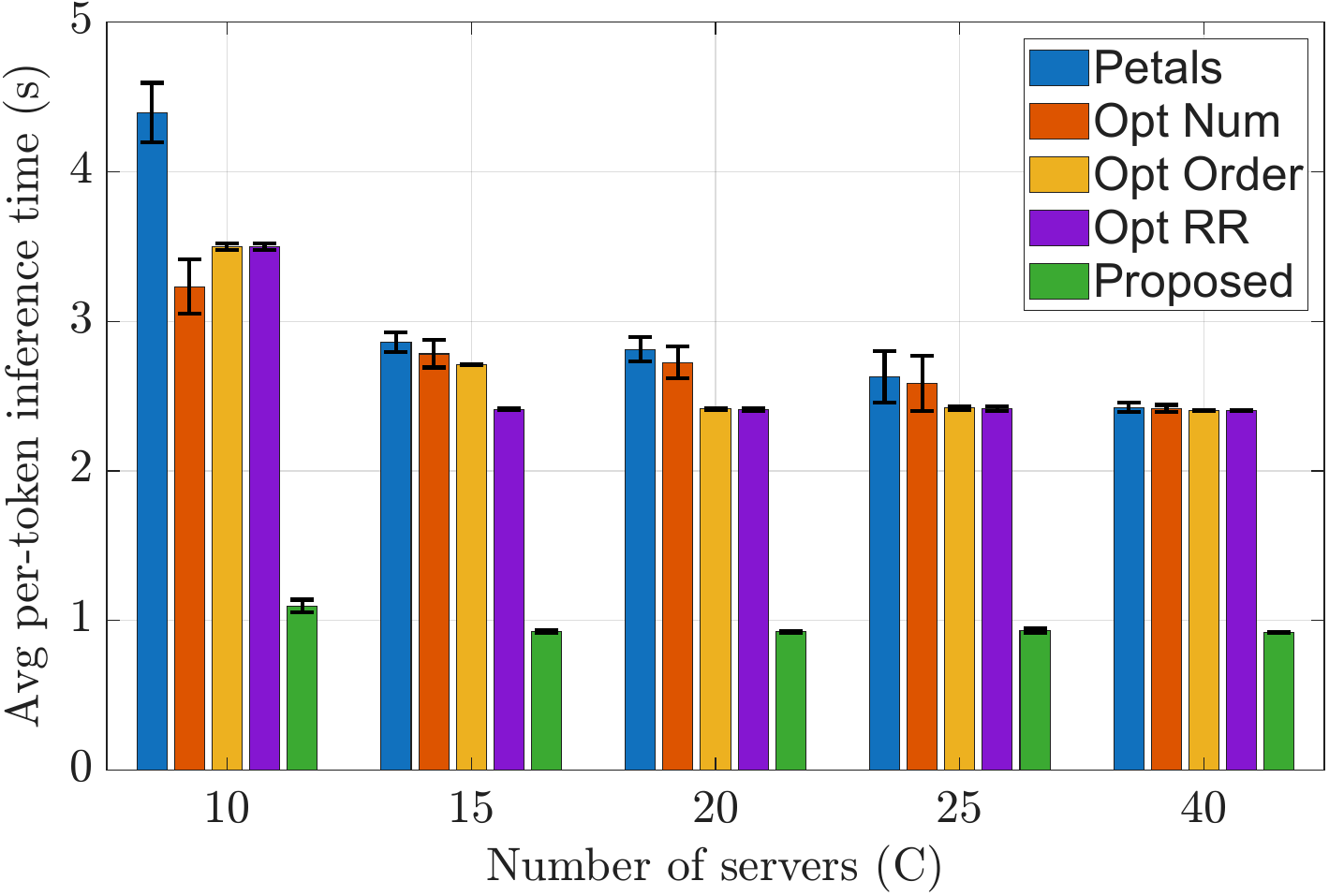}
  \caption{BellCanada}
\end{subfigure}
\begin{subfigure}[t]{0.31\textwidth}
  \centering
  \includegraphics[width=\textwidth]{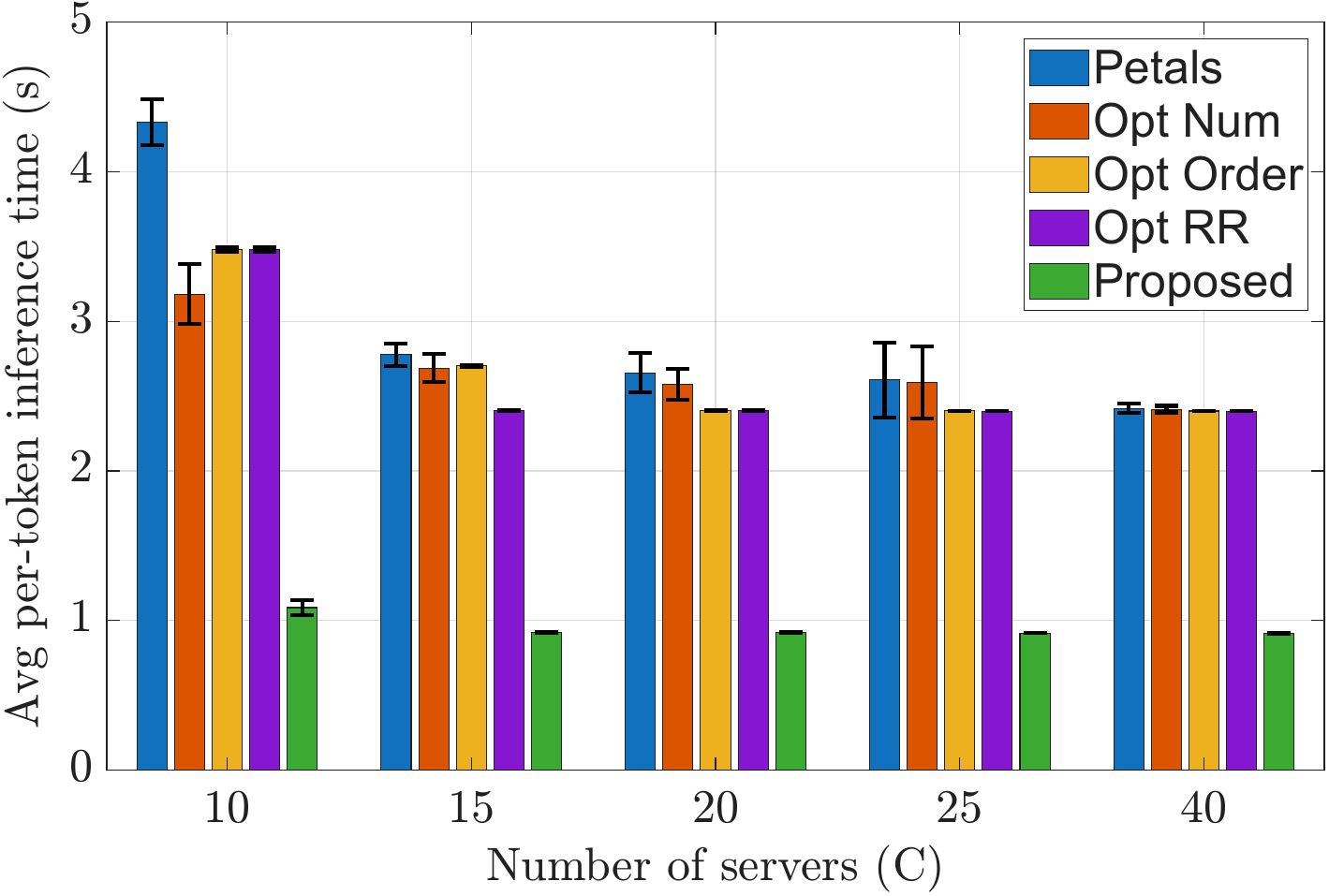}
  \caption{GTS-CE}
\end{subfigure}

\vspace{-1em}
\caption{Inference time per token when varying \#servers $C$ ($\eta = 0.2$, $\lambda=0.5$, $N_R = 100$, $\lmax^I=20$, $\lmax = 128$).}
\label{fig:inference_time_vary_C}
\end{figure}

\begin{figure}[t!]
\centering

\begin{subfigure}[t]{0.31\textwidth}
  \centering
  \includegraphics[width=\textwidth]{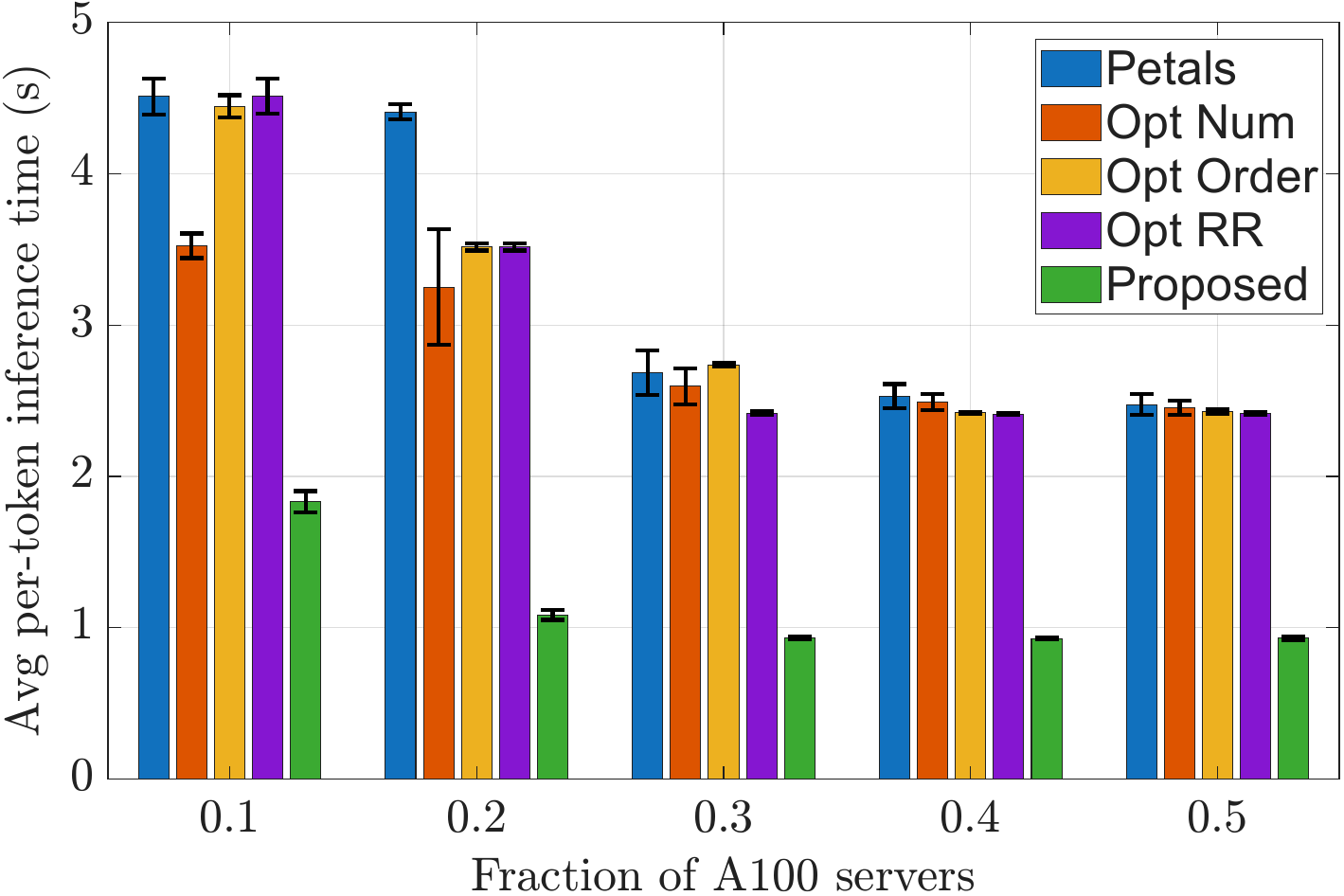}
  \caption{AboveNet}
\end{subfigure}
\begin{subfigure}[t]{0.31\textwidth}
  \centering
  \includegraphics[width=\textwidth]{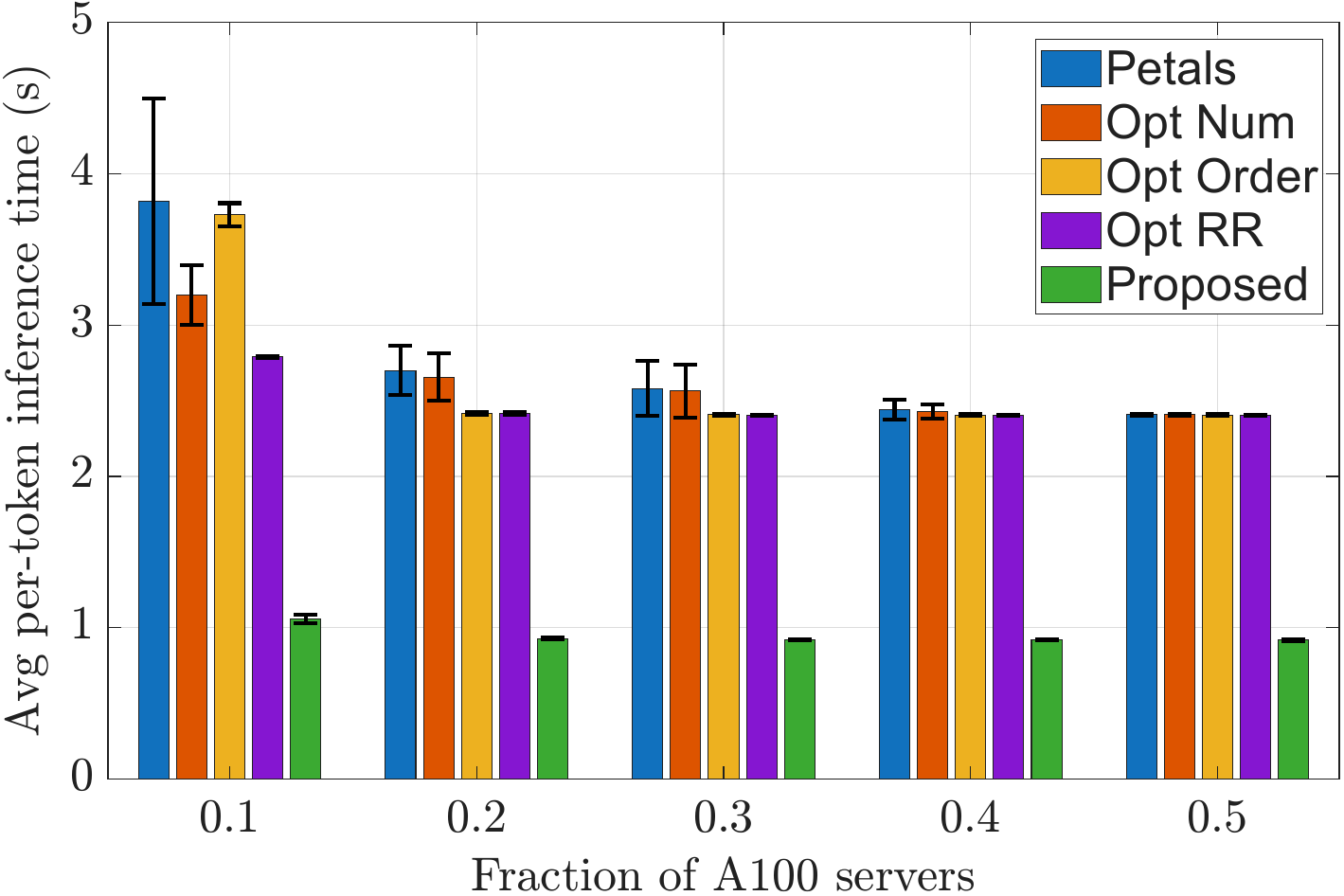}
  \caption{BellCanada}
\end{subfigure}
\begin{subfigure}[t]{0.31\textwidth}
  \centering
  \includegraphics[width=\textwidth]{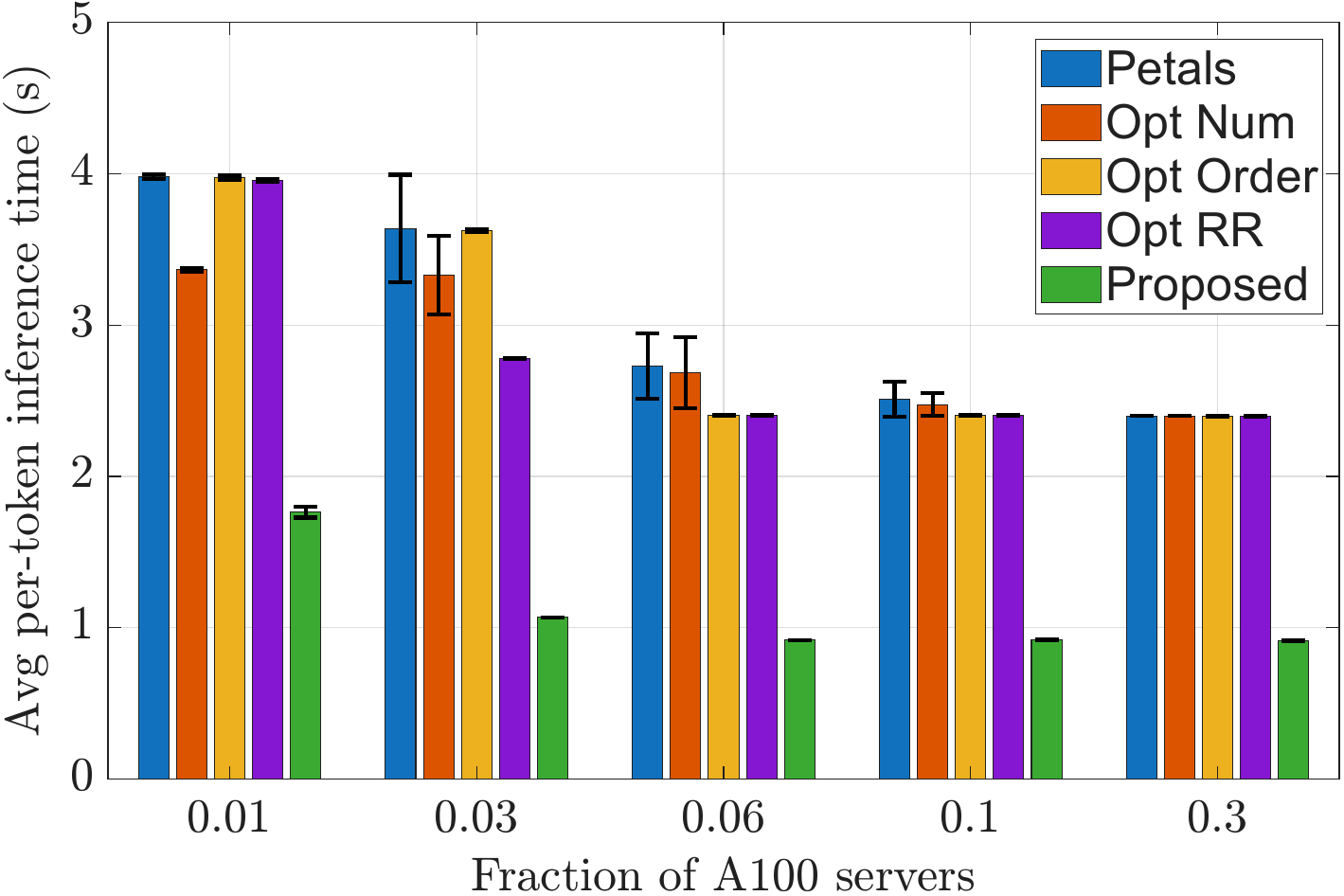}
  \caption{GTS-CE}
\end{subfigure}

\vspace{-1em}
\caption{Inference time per token when varying the fraction of high-performance servers $\eta$ 
($C = 0.4 \cdot \text{\#nodes}$, $\lambda = 0.5$, $N_R = 100$, $\lmax^I = 20$, $\lmax = 128$).}
\label{fig:inference_time_vary_eta}
\end{figure}

\begin{figure}[t!]
\centering

\begin{subfigure}[t]{0.31\textwidth}
  \centering
  \includegraphics[width=\textwidth]{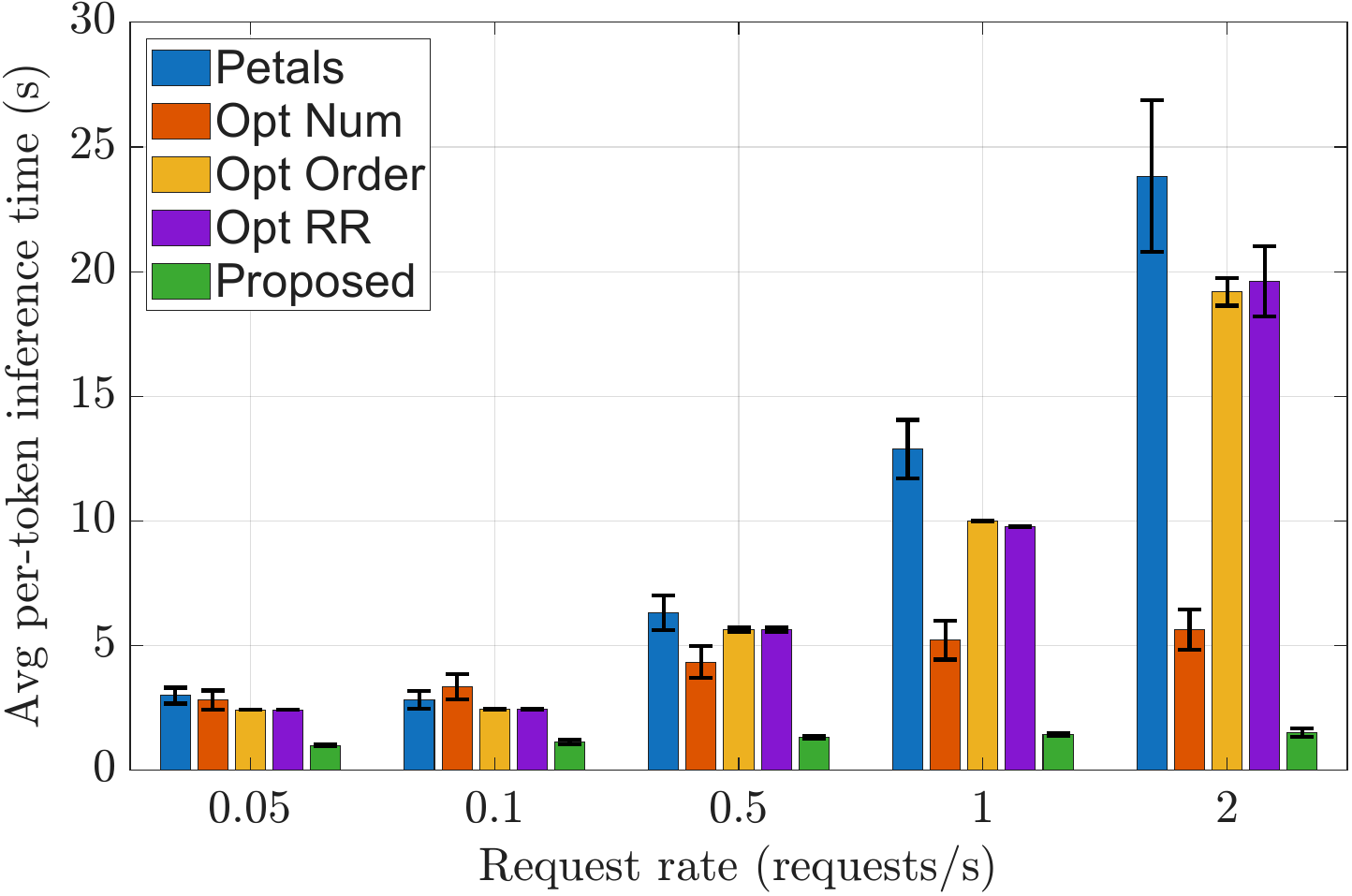}
  \caption{AboveNet}
\end{subfigure}
\begin{subfigure}[t]{0.31\textwidth}
  \centering
  \includegraphics[width=\textwidth]{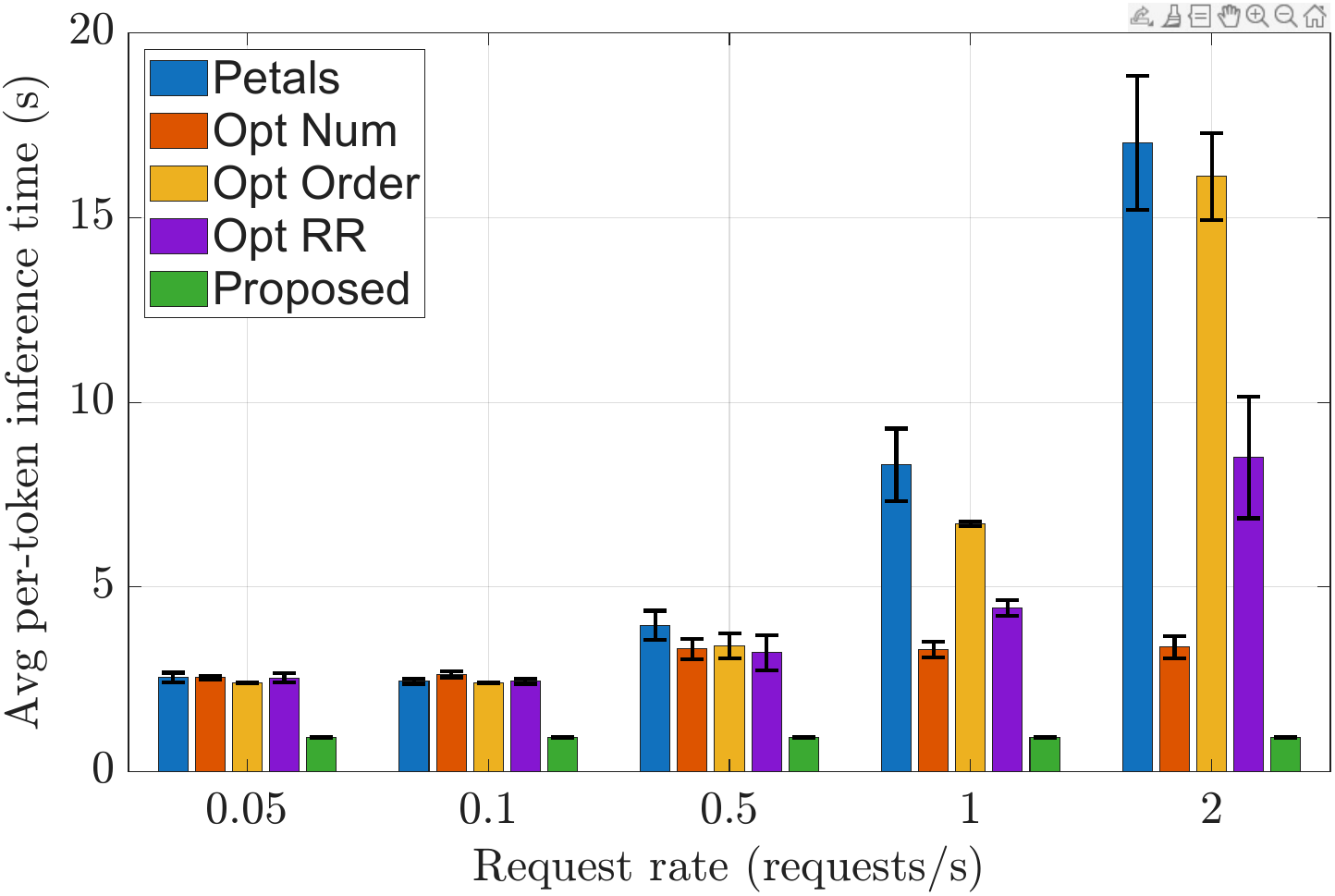}
  \caption{BellCanada}
\end{subfigure}
\begin{subfigure}[t]{0.31\textwidth}
  \centering
  \includegraphics[width=\textwidth]{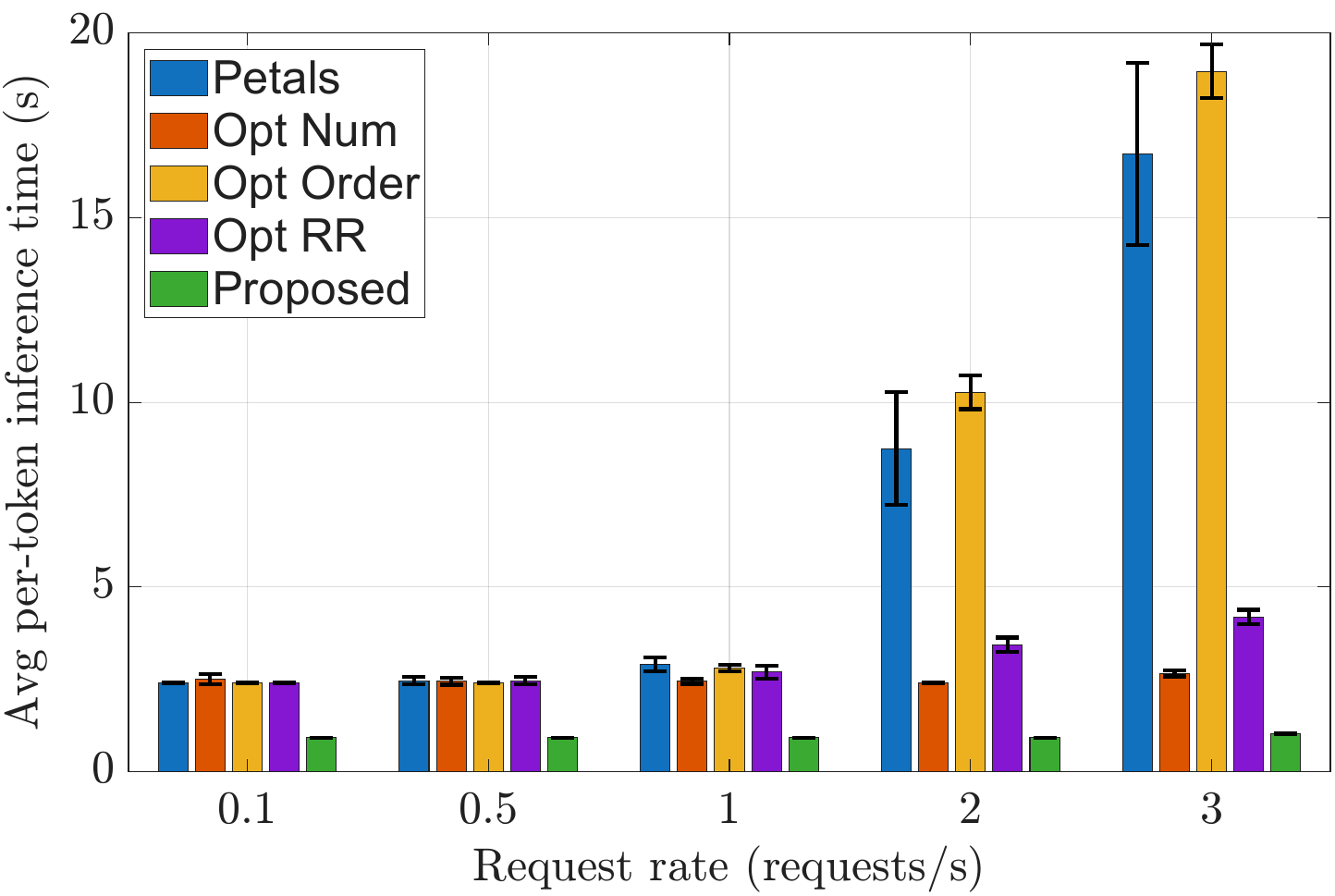}
  \caption{GTS-CE}
\end{subfigure}

\vspace{-1em}
\caption{Inference time per token when varying request rate $\lambda$ 
($C = 0.4 \cdot \text{\#nodes}$, $N_R = 200\lambda$, $\eta = 0.2$, $\lmax^I = 20$, $\lmax = 128$).}
\label{fig:inference_time_vary_lambda}
\end{figure}

\begin{figure}[t!]
\centering

\begin{subfigure}[t]{0.31\textwidth}
  \centering
  \includegraphics[width=\textwidth]{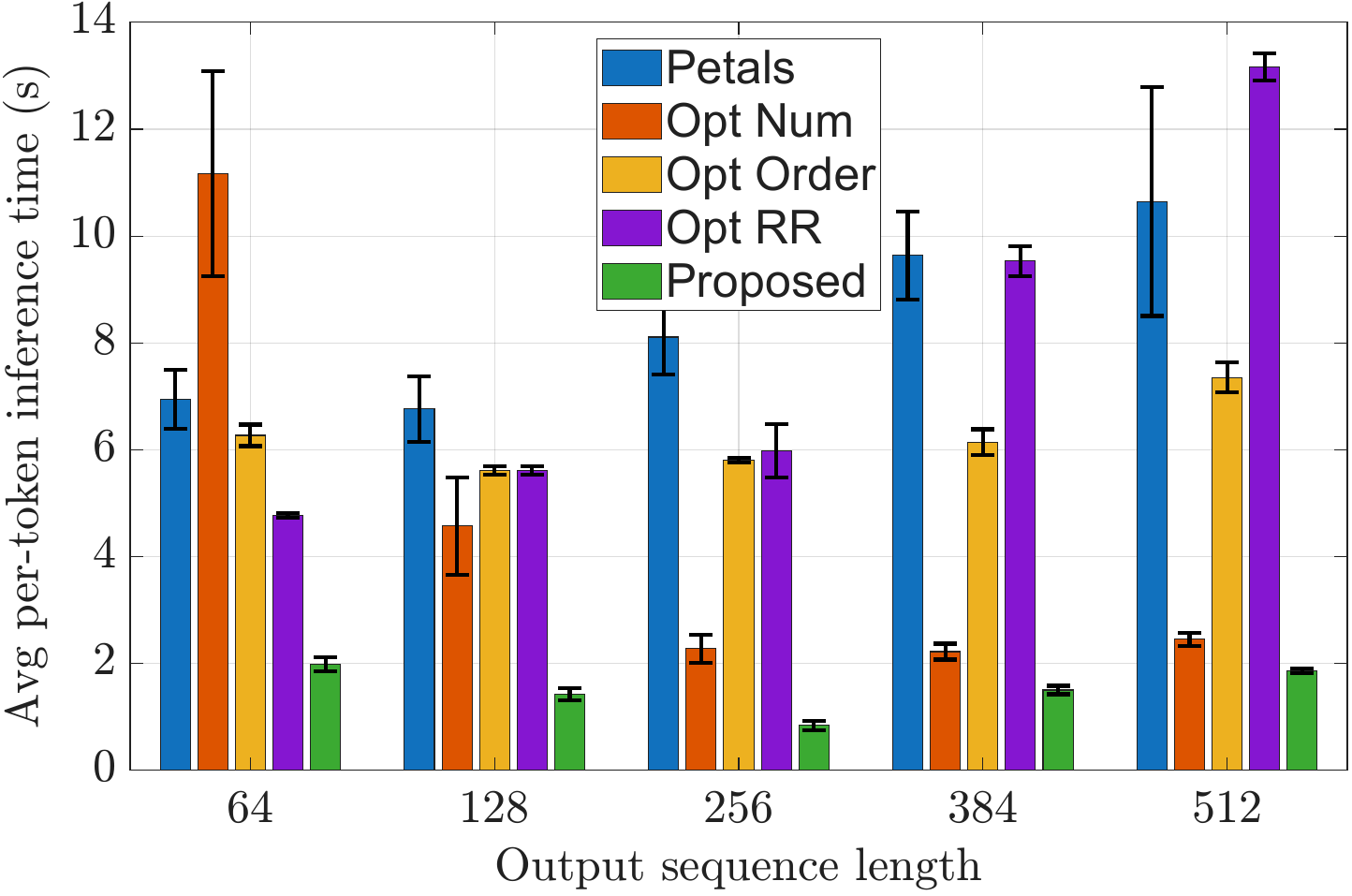}
  \caption{AboveNet}
\end{subfigure}
\begin{subfigure}[t]{0.31\textwidth}
  \centering
  \includegraphics[width=\textwidth]{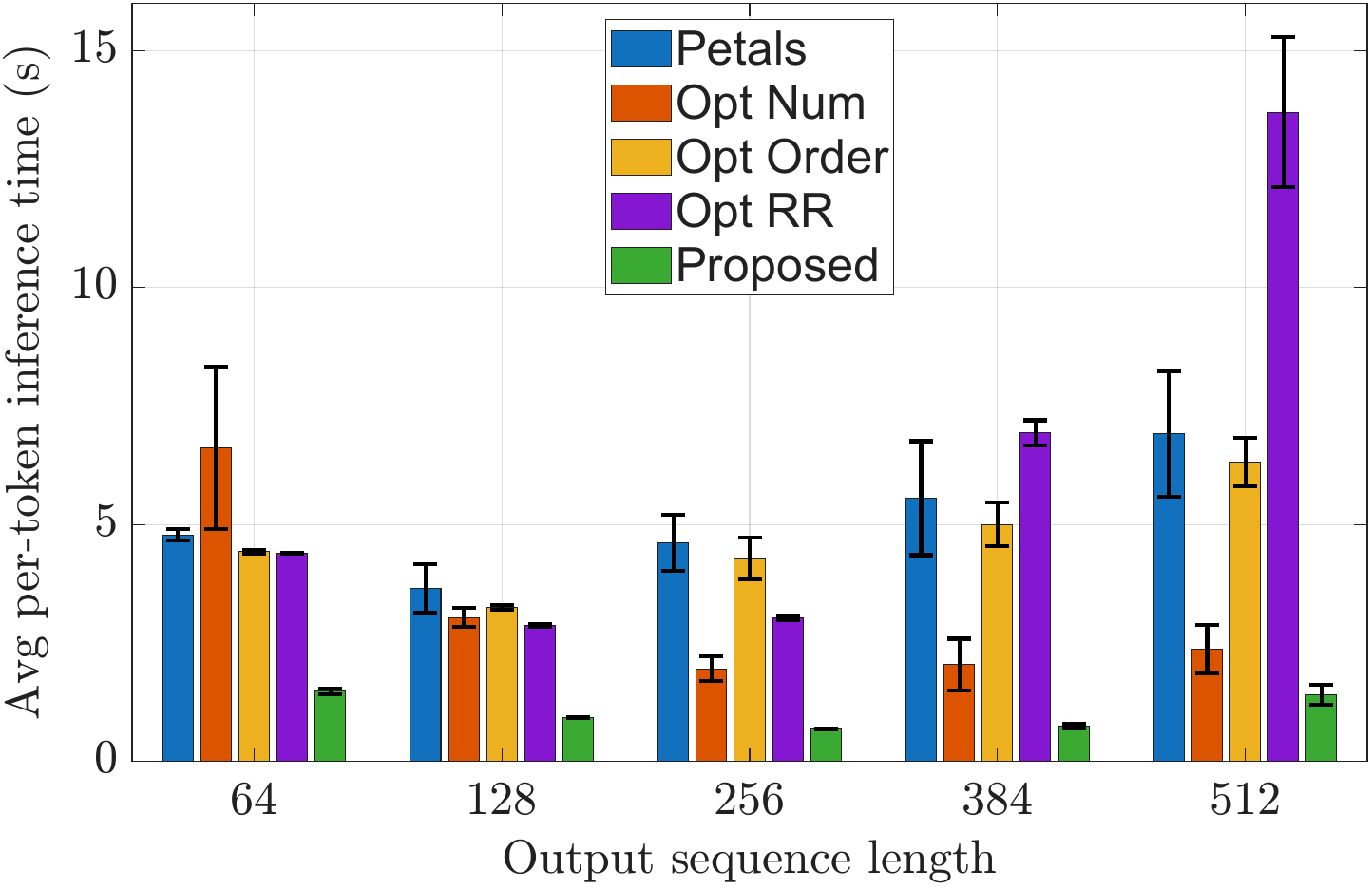}
  \caption{BellCanada}
\end{subfigure}
\begin{subfigure}[t]{0.31\textwidth}
  \centering
  \includegraphics[width=\textwidth]{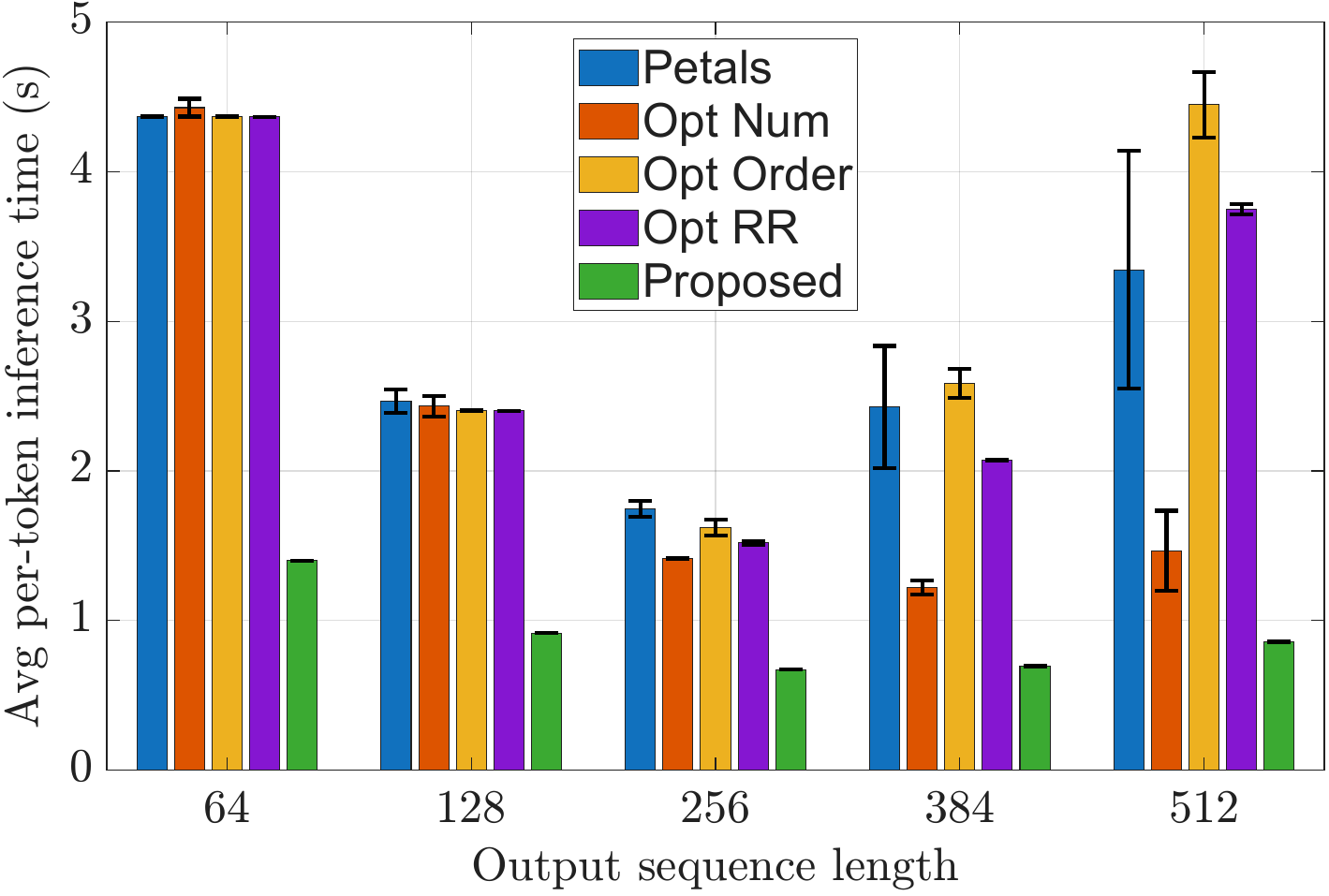}
  \caption{GTS-CE}
\end{subfigure}

\vspace{-1em}
\caption{Inference time per token when varying the sequence length $\lmax$ 
($C = 0.4 \cdot \text{\#nodes}$, $\eta = 0.2$, $\lambda = 0.5$, $N_R = 100$, $\lmax^I = 20$).}
\label{fig:inference_time_vary_lmax}
\end{figure}



We now use the validated MATLAB simulator to evaluate a wider range of scenarios based on the topologies in Table~\ref{tab:topo_info}. 
Fig.~\ref{fig:inference_time_vary_C}--\ref{fig:inference_time_vary_eta} show the average inference time per token (over all the tokens) as we vary the available resources in terms of either \#servers or the fraction of high-performance servers. 
Fig.~\ref{fig:inference_time_vary_lambda}--\ref{fig:inference_time_vary_lmax} show the corresponding results as we vary the demands in terms of the request rate or the sequence length.   
As additional benchmarks, we simulate {three} variations of PETALS' algorithm in addition to the original version in \cite{Borzunov23NeurIPS}: (1) applying PETALS' block placement algorithm to servers in an optimized order computed by line~\ref{Online CG:3} of our Alg.~\ref{Alg:Online BPRR} (`Optimized Order'), (2) applying PETALS' block placement algorithm but placing the same number of blocks per server as our algorithm (`Optimized Number'), 
and (3) optimizing the request routing according to the MILP \eqref{eq:online RR} under the block placement given by PETALS (`Optimized RR'). Besides additional benchmarks, these variations also serve the purpose of an ablation study as they each contain one aspect of our proposed solution.

\emph{First of all}, all the simulations confirm that the proposed algorithm can significantly accelerate the inference in comparison to the original algorithm in PETALS. 

\emph{Moreover}, these results suggest that: (i) optimizing the allocation of GPU memory between model blocks and attention caches as in `Optimized Number' can improve the inference time in most cases, particularly under relatively high demands; (ii) optimizing the order of block placement across servers as in `Optimized Order' may help in some cases, but does not always improve the inference time; (iii) similarly, optimizing request routing as in `Optimized RR' may help in some cases, but can worsen the performance for long sequences (Fig.~\ref{fig:inference_time_vary_lmax}). In contrast, our proposed solution that combines all these ideas is able to improve the performance in all the tested cases. This observation highlights the importance of \emph{systematically formulating and solving  the resource allocation problem in distributed LLM inference systems}. Meanwhile, the fact that `Optimized Order' and `Optimized RR' can sometimes underperform the original algorithm in PETALS suggests the suboptimality of the greedy block placement and myopic request scheduling strategy, which leaves potential room for improvement in future work. 

\emph{Furthermore}, comparing across the simulated cases indicates a trend that \emph{the performance improvement achieved by our proposed solution is larger in more resource-constrained scenarios} (e.g., fewer servers/high-performance servers or higher request rate), in which case good resource allocation is more important. We note that in practice, the demands usually grow with the system size. Thus, we further simulate a case of proportionally increasing \#servers and request rate in Fig.~\ref{fig:inference_time_vary_Clamda}, which shows a clear trend of widening performance gap between our solution and PETALS, indicating the potential for our solution to achieve even  greater {performance improvements in larger deployments}. 
We also conduct a sensitivity analysis with respect to the load parameter $|\mathcal{R}|$ as shown in Fig.~\ref{fig:inference_time_vary_actual_lambda}, which shows that while the configuration of this parameter affects the performance of our solution, its impact is mild and our solution remains superior to the benchmarks over a wide range of actual loads.


Meanwhile, our algorithm has slightly higher running times than the original algorithm in PETALS in some cases (see Fig.~\ref{fig:running_time_vary_C}--\ref{fig:running_time_vary_Clamda} in \ref{appendix:Additional Simulation Results}), but the difference is negligible compared to the actual inference time. 





\section{Conclusion}\label{sec:Conclusion}

In this work, we systematically studied the problem of performance optimization for geographically-distributed LLM inference, using PETALS as a concrete example. Using experimentally validated performance models, we formulated the optimal offline block placement and request routing problem as a MILP, proved its NP-hardness, and developed a polynomial-complexity algorithm with guaranteed performance. We then adapted our algorithm into a two-time-scale solution for the online setting with a guaranteed completion time under bounded load. Our experiments and cross-validated simulations in diverse settings not only confirmed the capability of the proposed algorithm in significantly reducing the inference times, but also provided insights on the key factors driving such improvement. 
In addition to the developed algorithm, this work also produced a light-weighted CPU-only simulator capable of predicting the performance of distributed LLM inference on profiled GPU servers, which will be open-sourced to facilitate future research on the performance of LLM for researchers with limited GPU access.

\bibliographystyle{unsrt}
\bibliography{references}

\appendix
\section{Supporting Proofs}\label{appendix:Proofs}

\begin{proof}[Proof of Lemma~\ref{lem:chain feasibility constraint}]
Denote the sequence of servers on $p$ by $i_1,\ldots,i_n$. 
For the first hop $(c,i_1)$, $a_{i_1}\leq a_c+m_c=1\leq a_{i_1}+m_{i_1}-1$ implies that block $1$ is hosted by server $i_1$, i.e., the first $a_{i_1}+m_{i_1}-1$ blocks can be found in order by traversing $p$ up to $i_1$. 
Similarly, if the first $a_{i_k}+m_{i_k}-1$ blocks ($k\in [n-1]$) can be found in order up to $i_k$, then $a_{i_{k+1}}\leq a_{i_k}+m_{i_k}\leq a_{i_{k+1}} + m_{i_{k+1}}-1$ implies that the first $a_{i_{k+1}}+m_{i_{k+1}}-1$ blocks can be found in order up to $i_{k+1}$. Thus, the first $a_{i_n}+m_{i_n}-1$ blocks can be found in order up to $i_n$. Moreover, since $L+1=a_{c'}\leq a_{i_n}+m_{i_n} \leq a_{c'}+m_{c'}-1 = L+1$, $a_{i_n}+m_{i_n}-1 = L$. Thus, all the blocks can be found in order along path $p$, i.e., $p$ is feasible. 
\end{proof}

\begin{proof}[Proof of Theorem~\ref{thm:NP-hardness of BPRR}]
We prove this theorem by a reduction from the optimization version of the \emph{partition problem}~\cite{Hayes02AS}. Given a set of positive integers $W$ with $\Delta:= (\sum_{w\in W} w)/2$, the optimization version of the partition problem aims at finding a partition of $W$ into $W_1$ and $W_2$ such that $|\sum_{w\in W_1} w - \sum_{w\in W_2}w|$ is minimized. Without loss of generality, assume all the numbers in $W$ to be even and upper-bounded by\footnote{The first assumption is because we can scale all the numbers by $2$ without changing the solution. The second assumption is because if $\exists w^*\in W$ such that $w^*> \Delta$, then the optimal solution is simply $W_1=\{w^*\}$ and $W_2=W\setminus \{w^*\}$, and thus it suffices to consider instances where $w\leq \Delta$ for all $w\in W$.} $\Delta$. We construct an instance of the BPRR problem as follows: set the total number of blocks to $L=2$; construct a single client $V_c:=\{c\}$ with $|\mathcal{R}_c|=\Delta$; construct a server $j$ for each $w_j\in W$, and set $s_m$, $s_c$, and $M_j$ such that $s_m+s_c w_j = M_j < 2 s_m$; set $t_j:= t_{cj} + \tau_j = 1$ for each $w_j\in W$; construct another server $0$ (assuming $w_0\not\in W$) with $s_m+s_c \Delta = M_0 < 2 s_m$ and $t_0:= t_{c0}+\tau_0 = 2$. This construction is possible as long as $s_m/s_c> \Delta$ (which implies $s_m/s_c > w_j$, $\forall w_j\in W$). For any partition $(W_1, W_2)$ with $\sum_{w\in W_1}w \geq \sum_{w\in W_2}w$, let $\epsilon := {1\over 2} (\sum_{w\in W_1}w - \sum_{w\in W_2}w)$, i.e., $\sum_{w\in W_1}w = \Delta+\epsilon$ and $\sum_{w\in W_2}w = \Delta-\epsilon$. The assumption of $W$ containing only even numbers ensures that $\Delta$ and $\epsilon$ are both nonnegative integers. 
By construction, every server can only store one block, server $j$ for each $w_j\in W$ can serve $w_j$ parallel sessions with a per-token inference time of $t_j=1$, and server $0$ can serve $\Delta$ parallel sessions with a per-token inference time of $t_0=2$. It is easy to see that the optimal objective value of the constructed instance of BPRR conditioned on placing block $1$ at the servers in $S_1:=\{j: w_j\in W_1\}$ and block $2$ at the servers in $S_2:=\{j: w_j\in W_2\}\cup \{0\}$ is $2\Delta + \epsilon$, achieved by routing the maximum \#requests to the faster servers before using server $0$. Thus, minimizing the subset sum difference for the partition problem is equivalent to minimizing the objective value of the corresponding BPRR problem. The proof completes by noting that the optimization version of the partition problem is NP-hard~\cite{Hayes02AS}. 
\end{proof}

\begin{proof}[Proof of Lemma~\ref{lem:minimize inference time under relaxed routing}]
\begin{figure}[!t]
   \centerline{\includegraphics[width=0.4\linewidth]{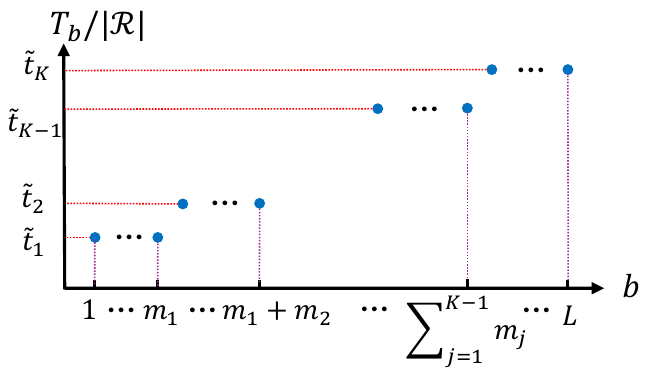}}
   \vspace{-1em}
    \caption{Illustration of $T_b$ in Alg.~\ref{Alg:CG-BPRR} after placing all the blocks. }
    \label{fig:CG_proof}
    \vspace{-.05em}
\end{figure}
As servers are considered in the increasing order of $\widetilde{t}_j$ (line~\ref{CG:3}), the value of $T_b$ for each block $b$ only changes once from $\widetilde{t}_0 |\mathcal{R}|$ to $\widetilde{t}_{j_b} |\mathcal{R}|$, when block $b$ is placed for the first time on a server $j_b\in V_s$. Therefore, with ``$\argmax$'' breaking ties in favor of smaller index, line~\ref{CG:5} will select continuous blocks for each server before all the blocks are placed, i.e., server $1$ will host blocks $1,\ldots,m_1$, server $2$ will host blocks $m_1+1,\ldots,m_1+m_2$, etc. This continues until server $K:= \min\{k: \sum_{j=1}^k m_j \geq L\}$, for which the ``$\argmax$'' in line~\ref{CG:5} will be achieved at the last possible index $a_K = L-m_K+1$. The resulting $T_b$'s form a monotone increasing piece-wise constant series as shown in Fig.~\ref{fig:CG_proof}. 
Under the conservative setting of $m_j$ in line~\ref{CG:1} and this block placement, the relaxed request routing will simply route all the $|\mathcal{R}|$ requests through the first $K$ servers, achieving an average per-token inference time of 
\begin{align}\label{eq:proof CG-BP, min per-token time}
\sum_{j=1}^K m_j \widetilde{t}_j - \widetilde{t}_K(\sum_{j=1}^K m_j - L),
\end{align}
where the second term is to adjust for the $\sum_{j=1}^K m_j - L$ blocks hosted on both server $K-1$ and server $K$, as they will only be processed at server $K-1$. 
Meanwhile, since these are the fastest servers (in terms of $\widetilde{t}_j$) that can host the entire model, the minimum per-token inference time under the relaxed request routing is lower-bounded by \eqref{eq:proof CG-BP, min per-token time}. 
This completes the proof. 
\end{proof}

\begin{proof}[Proof of Lemma~\ref{lem:optimality of shortest path routing - CG-BPRR}]
Due to the conservative calculation of the number of blocks per server (line~\ref{CG:1}), the GPU memory capacity will always be satisfied, i.e., \eqref{old RR:memory} can be ignored. This decouples \eqref{eq:RR - given BP} into $|V_c|$ independent routing problems for each client. 

For each client $c$, the route feasibility constraints \eqref{eq:bilinear - 1}\mbox{--}\eqref{eq:bilinear - 2} can be enforced by limiting its request routing to a subgraph $G^c_{\bm{a},\bm{m}}=(V^c, E^c_{\bm{a},\bm{m}})$ of $G$, where $V^c$ only contains the servers $V_s$ and the S-client/D-client of $c$, and $E^c_{\bm{a},\bm{m}}$ only contains feasible routing links, i.e., $(i,j)\in E^c_{\bm{a},\bm{m}}$ if and only if \eqref{eq:chain feasibility} is satisfied. Each link $(i,j)\in E^c_{\bm{a},\bm{m}}$ has a routing cost of $t^c_{ij}$. It is easy to see that the subproblem of \eqref{eq:RR - given BP} for client $c$ is equivalent to the problem of finding the shortest (i.e., least-cost) path from the S-client to the D-client in $G^c_{\bm{a},\bm{m}}$ (line~\ref{CG:12}). In absence of the capacity constraint \eqref{old RR:memory}, the optimal paths for all the requests from the same client are the same (line~\ref{CG:13}). This completes the proof. 
\end{proof}

\begin{proof}[Proof of Theorem~\ref{thm:CG-BPRR}]
Recall that $t_{*j} := \max_{c\in V_c} t_{cj}$. Since by Lemma~\ref{lem:optimality of shortest path routing - CG-BPRR} the shortest-path request routing is optimal under the block placement by CG-BPRR, the average per-token inference time $T^g$ achieved by CG-BPRR is upper-bounded by the average per-token inference time under any request routing that is feasible under the block placement given by CG-BPRR. Specifically, consider a solution that routes all the requests through the chain of servers $1,\ldots,K$. By the proof of Lemma~\ref{lem:minimize inference time under relaxed routing}, server $1$ will host the first $m_1$ blocks, server $2$ will host the next $m_2$ blocks, etc., and server $K$ will host the last $m_K$ blocks. Thus, servers $1,\ldots,K$ collectively host all the blocks. Moreover, due to the conservative computation of $m_j$ in line~\ref{CG:1} of Alg.~\ref{Alg:CG-BPRR}, each server has enough capacity to serve all the requests concurrently. Thus, the above routing is feasible. 
Under this solution, server $j$ ($j=1,\ldots,K-1$) processes $m_j$ blocks, yielding a per-token inference time of 
\begin{align}
t_{cj}+\tau_j m_j\leq t_{*j}+\tau_j m_j = m_j \widetilde{t}_j,~~~\forall c\in V_c.
\end{align}
Meanwhile, server $K$ only processes the last $L-\sum_{j=1}^{K-1}m_j$ blocks, yielding a per-token inference time of
\begin{align}
t_{cK} + \tau_K (L- \hspace{-.25em}\sum_{j=1}^{K-1}m_j) \leq t_{*K} + \tau_K (L- \hspace{-.25em}\sum_{j=1}^{K-1}m_j) = t_{*K}\left({\sum_{j=1}^K m_j-L\over m_K} \right)+(L- \hspace{-.25em}\sum_{j=1}^{K-1}m_j ) \widetilde{t}_K,~\forall c\in V_c. 
\end{align}
The average per-token inference time $\overline{T}^g$ under the above routing solution is thus upper-bounded by
\begin{align}
\overline{T}^g &\leq \sum_{j=1}^{K-1} m_j \widetilde{t}_j + t_{*K}\left({\sum_{j=1}^K m_j-L\over m_K} \right) + (L- \sum_{j=1}^{K-1}m_j ) \widetilde{t}_K  \nonumber\\
&= \sum_{j=1}^K \widetilde{t}_j m_j - \left(\sum_{j=1}^K m_j -L\right)\left(\widetilde{t}_K-{t_{*K}\over m_K} \right) \nonumber\\
&=  \sum_{j=1}^K \widetilde{t}_j m_j - \tau_K\left(\sum_{j=1}^K m_j -L\right), \label{eq:CG-proof-1}
\end{align}
where \eqref{eq:CG-proof-1} is because $\widetilde{t}_K-{t_{*K}/ m_K} = \tau_K$ by the definition in \eqref{eq:amortized inference time}. 
The proof is completed by noting that $T^g \leq \overline{T}^g$. 
\end{proof}

\begin{proof}[Proof of Corollary~\ref{cor:CG-BPRR}]
The result is directly applied by Theorem~\ref{thm:CG-BPRR}. Specifically, if the number of concurrent requests is no more than $|\mathcal{R}|$, then the block placement by CG-BP ensures that there is no memory contention between requests, i.e., each request can be routed independently of the others. Since the average per-token inference time under the feasible but possibly suboptimal request routing through servers $1,\ldots,K$ is already bounded by \eqref{eq:CG-BPRR bound} as shown in the proof of Theorem~\ref{thm:CG-BPRR}, the average per-token inference time under the optimal request routing must be bounded by \eqref{eq:CG-BPRR bound}. 
\end{proof}

\begin{proof}[Proof of Corollary~\ref{cor:Online RR}]
The first claim follows from the fact that the completion time of $r^*$ (i.e., the time from its arrival to its completion) if scheduled to path $p_c(t)$ is
\begin{align}\label{eq:completion time bound}
\max_{(i,j)\in p_c(t)}t^A_{ij} + \lmax \sum_{(i,j)\in p_c(t)}t^c_{ij} \leq \sum_{(i,j)\in p_c(t)} \left(t^A_{ij}(t) + \lmax t^c_{ij}\right). 
\end{align}

The second claim follows from the fact that there is no waiting time if the number of concurrent requests is within $|\mathcal{R}|$, i.e., $t^A_{ij}(t)\equiv 0$ $\forall (i,j)$. In this case, the bound \eqref{eq:completion time bound} is tight, i.e., the cost of a path equals the completion time on this path. By definition, $p_c(t)$ minimizes the path cost, and thus it achieves the minimum completion time. 
\end{proof}

\section{Other Supporting Materials}\label{appendix:Other Materials}

\subsection{Additional Results on Performance Model Validation}\label{appendix:Additional Model Validation}

\begin{figure}[t!]
\begin{minipage}{.495\linewidth}
\centerline{
\includegraphics[width=1\linewidth]{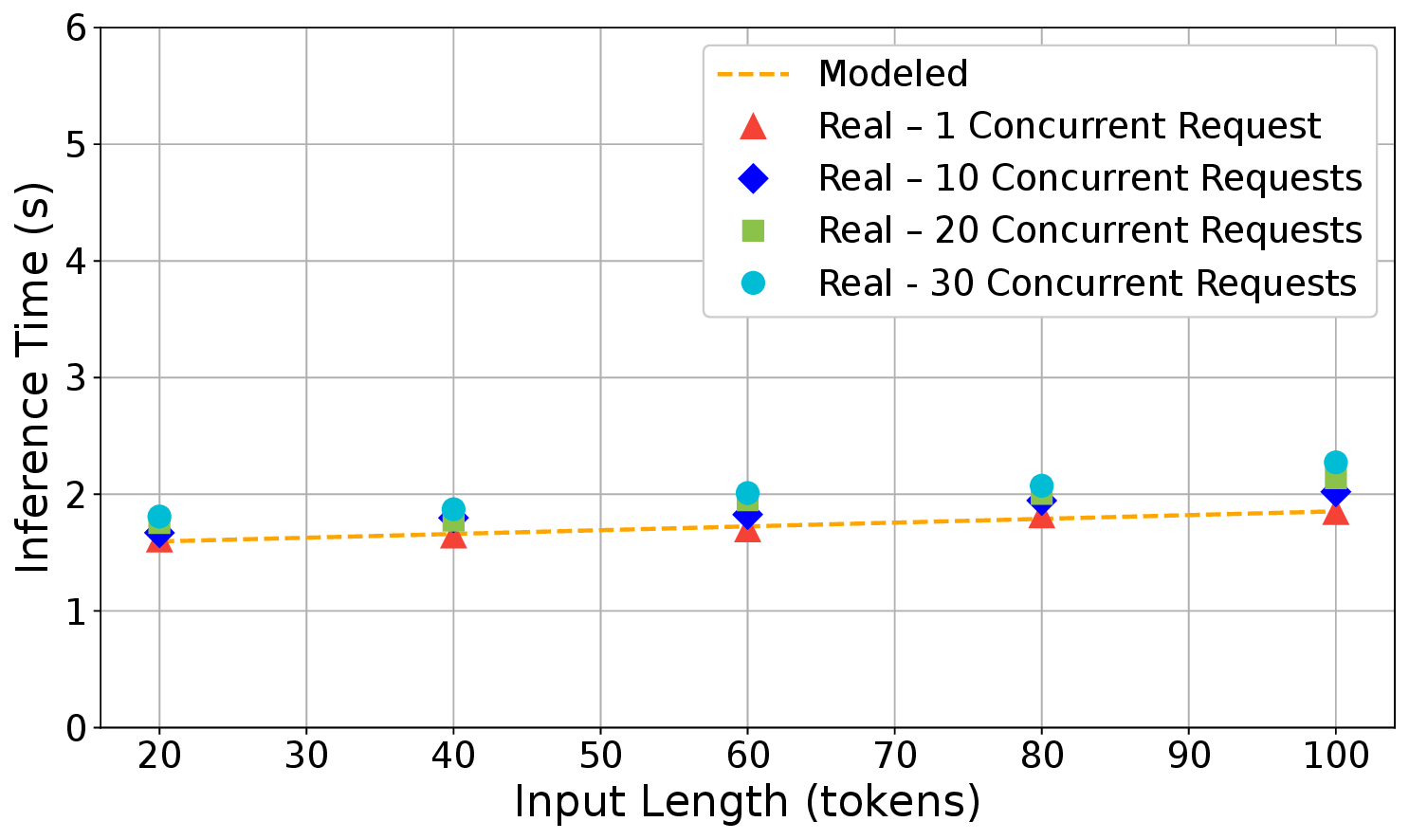}}
\centerline{\scriptsize (a) first token}
\end{minipage}
\begin{minipage}{.495\linewidth}
\centerline{
\includegraphics[width=1\linewidth]{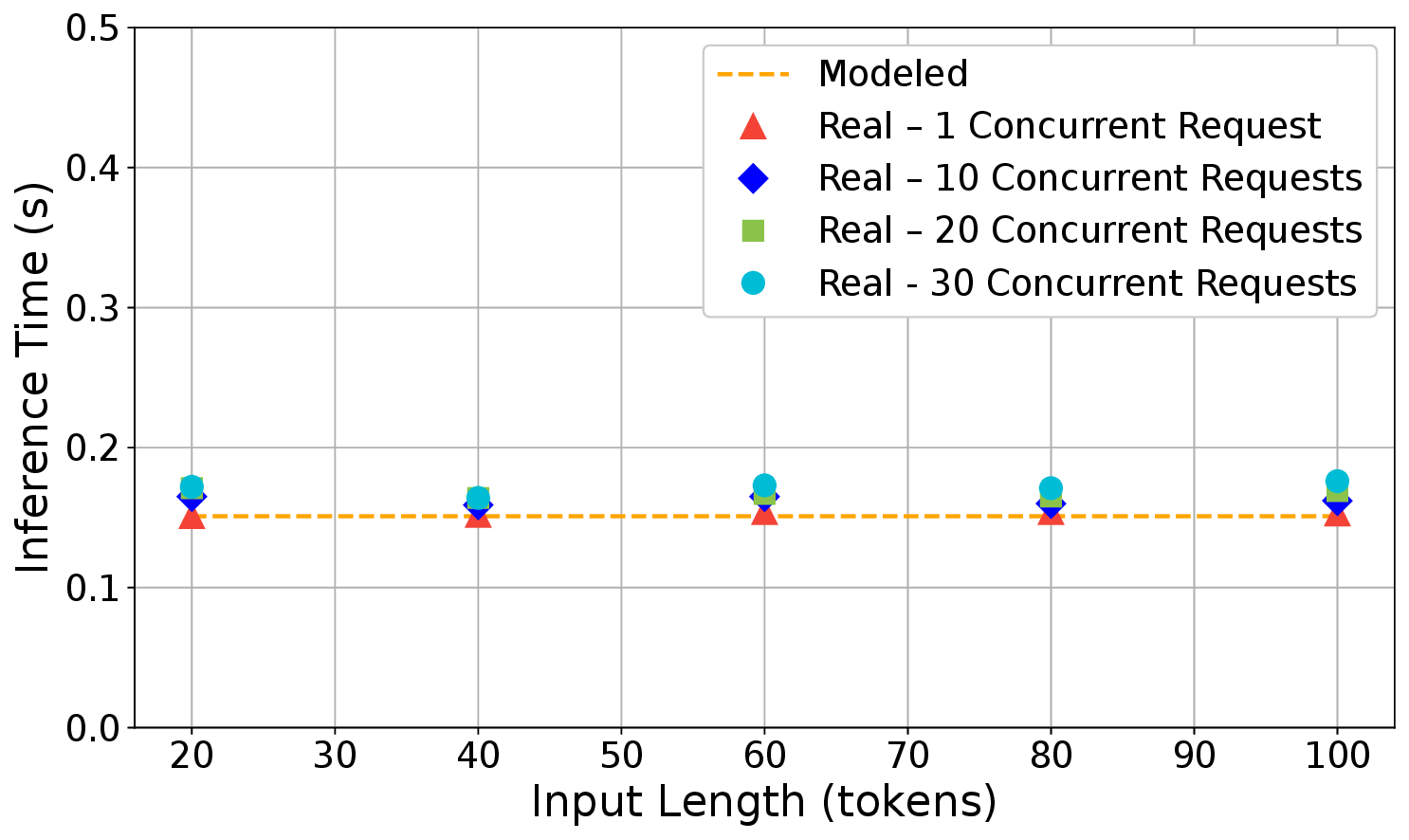}}
\centerline{\scriptsize (b) each of remaining tokens}
\end{minipage}
\vspace{-.5em}
\caption{Inference time vs. input length on A100: (a) for first token generation (b) for rest of token generation (40 blocks, $\lmax=128$). }\label{fig:time_input_length}
\vspace{-.05em}
\end{figure}

\begin{figure}[t!]
\begin{minipage}{.495\linewidth}
\centerline{
\includegraphics[width=1\linewidth]{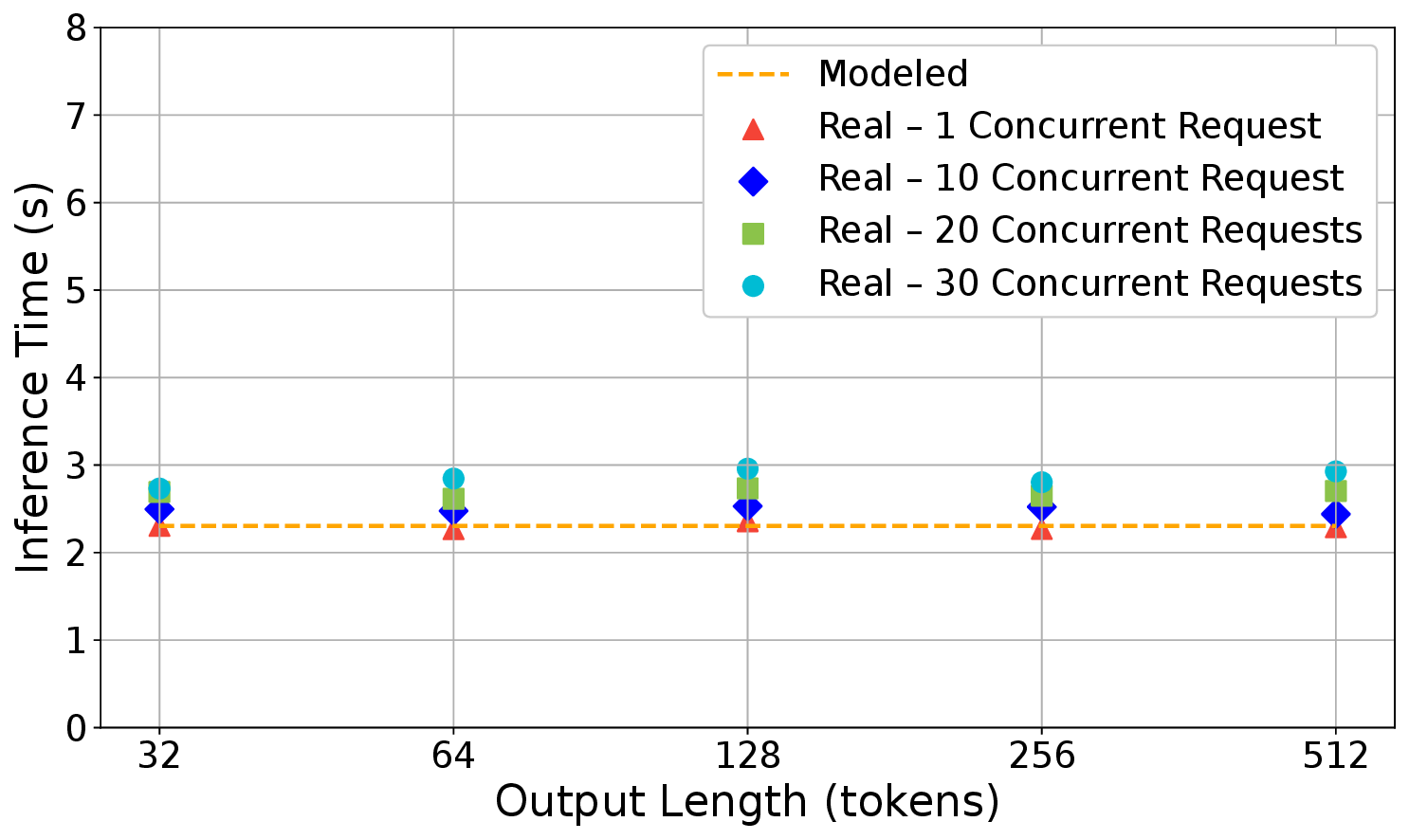}}
\centerline{\scriptsize (a) first token}
\end{minipage}
\begin{minipage}{.495\linewidth}
\centerline{
\includegraphics[width=1\linewidth]{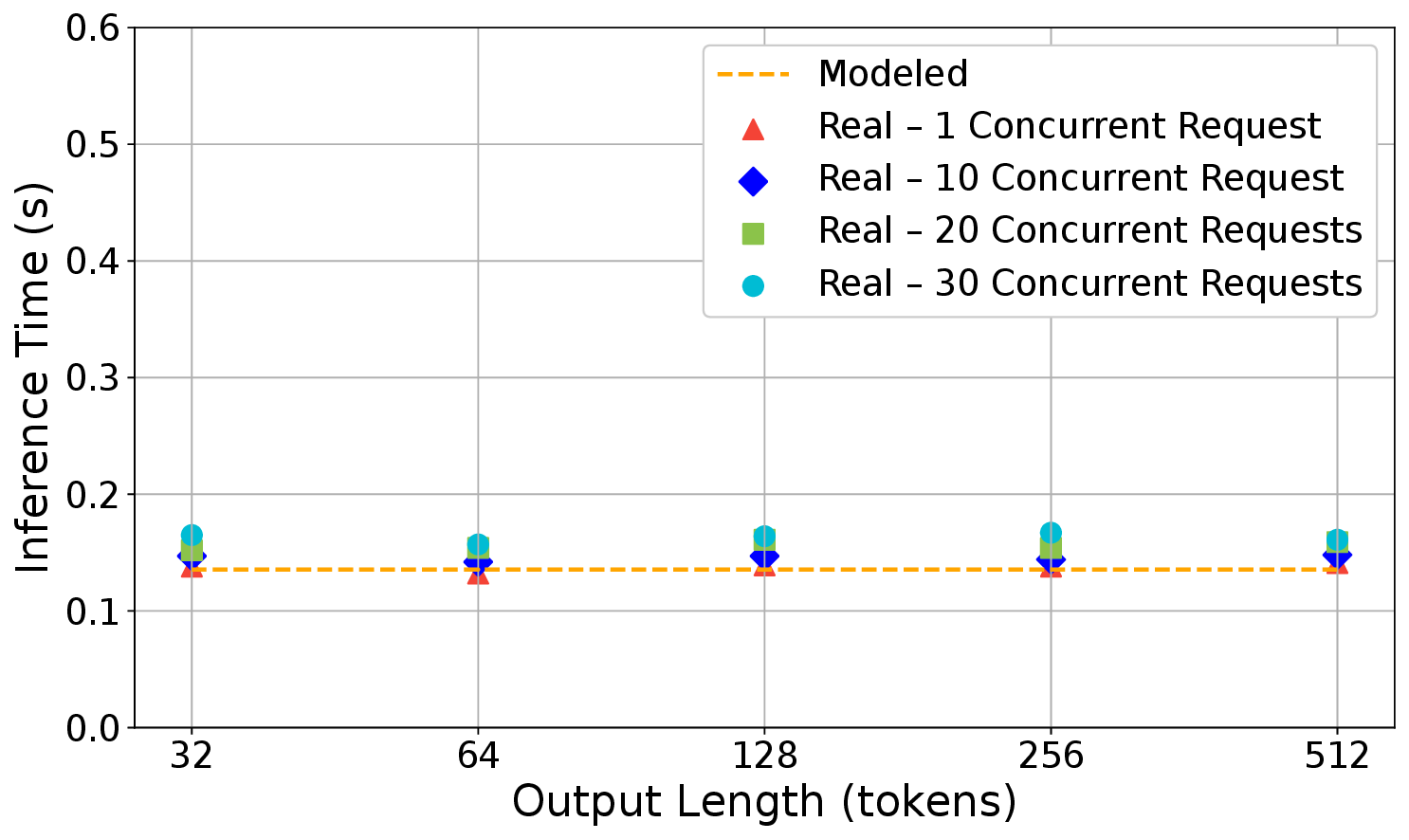}}
\centerline{\scriptsize (b) each of remaining tokens}
\end{minipage}
\vspace{-.5em}
\caption{Inference time vs. output length on A100: (a) for first token generation (b) for rest of token generation (40 blocks, $\lmax^I=20$). }\label{fig:time_output_length}
\vspace{-.05em}
\end{figure}

In addition to the number of processed blocks as in Fig.~\ref{fig:time_blocks}, we have also evaluated the inference time incurred at a given server under varying input/output length as shown in Fig.~\ref{fig:time_input_length}--\ref{fig:time_output_length}. 
These results have validated our inference time model by confirming that the inference time for the first token only depends on the input length $\lmax^I$ but not the output length $\lmax$, and the inference time for each of the subsequent tokens does not depend on either the input or the output length. 
Besides A100, we have also validated our performance models on an MIG, which shows similar results. \looseness=-1



\subsection{Extension to Heterogeneous Input/Output Lengths}
\label{appendix:Heterogeneous Lengths}

In the heterogeneous case where each request $r\in \mathcal{R}$ has an input length $l^I_r$ and an output length $l_r$, the attention cache size becomes heterogeneous $s^r_c:=2\dmodel\cdot (l^I_r+l_r) \cdot \dtypebytes$ (bytes). Our formulation \eqref{eq:BPRR - direct} can be easily extended by replacing the objective function \eqref{direct:obj} with 
\begin{align}
\sum_{c\in V_c}  \sum_{r\in \mathcal{R}_c} \sum_{p\in P_c(\bm{a},\bm{m})} f^r_p l_r \sum_{(i,j)\in p} t^r_{ij}, 
\end{align}
where $t^r_{ij}$ denotes the per-token inference time incurred by request $r$ at link $(i,j)$ averaged over all the tokens generated for this request, and the constraint \eqref{direct:memory} with 
\begin{align}
s_m m_j + \sum_{c\in V_c}\sum_{r\in \mathcal{R}_c}s^r_c \mathop{\sum_{p\in P_c(\bm{a},\bm{m}):}}_{(i,j)\in p, \exists i} \hspace{-1em} f^r_p (a_j+m_j-a_i-m_i)  \leq M_j,~~\forall j\in V_s. 
\end{align}
Similar to \eqref{eq:true avg per-token time}, $t^r_{ij}$ is defined as:
\begin{align}
\left({1\over l_r}t^I_{cj}(l^I_r)+{l_r-1\over l_r}t_{cj}\right) + \left( {1\over l_r}\tau^I_j(l^I_r)+{l_r-1\over l_r}\tau_j\right) (a_j+m_j-a_i-m_i).\label{eq:heterogeneous - original t}
\end{align}
When $l^I_r\ll l_r$ for all $r\in \mathcal{R}$, \eqref{eq:heterogeneous - original t} can be simplified to
\eqref{eq:per-token inference time}, where $c$ is the client sending request $r$. 

Since this change only affects some constant scaling factors, we can still convert the resulting optimization into a MILP as in Section~\ref{subsubsec:MILP Formulation}. 
With $s_c:= \max_{r} s^r_c$, CG-BPRR can still be used to obtain a feasible solution in the offline setting. 
We note that since in practice the requests arrive dynamically, the precise input/output length for each request is unknown ahead of time. In this sense, our proposed solution based on \eqref{eq:BPRR - direct} allocates resources according to the maximum input and output lengths to ensure feasibility while trying to optimize the worst-case performance, where the maximum input and output lengths are system parameters announced to the clients.

\subsection{Linearization of Bilinear Terms}\label{appendix:Linearization}

To linearize the bilinear terms $a_j f^r_{ij}$, $a_i f^r_{ij}$, $m_j f^r_{ij}$, and $m_i f^r_{ij}$, we introduce auxiliary variables $\alpha^r_{ij}, \beta^r_{ij}, \gamma^r_{ij}, \delta^r_{ij} \geq 0$ as well as the following linear constraints. Specifically, using the constraints
\begin{subequations}\label{eq:define alpha}
\begin{align}
&-(L+1) f^r_{ij} + \alpha^r_{ij} \leq 0,\\
&-a_j + \alpha^r_{ij} \leq 0,\\
&a_j + (L+1) f^r_{ij} - \alpha^r_{ij} \leq L+1,
\end{align}
\end{subequations}
we can ensure that $\alpha^r_{ij} = a_j f^r_{ij}$ for both $f^r_{ij}=0$ and $f^r_{ij}=1$ (the `$L+1$' is to cover the case of $a_j=L+1$ for $j\in V_c^D$). 
Similarly, we can ensure that $\beta^r_{ij} = a_i f^r_{ij}$ by the constraints
\begin{subequations}\label{eq:define beta}
\begin{align}
&-L f^r_{ij}+\beta^r_{ij}\leq 0, \\
&-a_i + \beta^r_{ij}\leq 0, \\
&a_i + L f^r_{ij} - \beta^r_{ij}\leq L, 
\end{align}
\end{subequations}
$\gamma^r_{ij} = m_j f^r_{ij}$ by the constraints
\begin{subequations}\label{eq:define gamma}
\begin{align}
&-L f^r_{ij}+\gamma^r_{ij}\leq 0, \\
&-m_j + \gamma^r_{ij}\leq 0, \\
&m_j + L f^r_{ij} - \gamma^r_{ij}\leq L, 
\end{align}
\end{subequations}
and $\delta^r_{ij} = m_i f^r_{ij}$ by the constraints
\begin{subequations}\label{eq:define delta}
\begin{align}
&-L f^r_{ij}+\delta^r_{ij}\leq 0, \\
&-m_i + \delta^r_{ij}\leq 0, \\
&m_i + L f^r_{ij} - \delta^r_{ij}\leq L. 
\end{align}
\end{subequations}

\begin{algorithm}[tb]
\small
\SetKwInOut{Input}{input}\SetKwInOut{Output}{output}
\Input{set of clients $V_c$, target \#requests $|\mathcal{R}|$, \#blocks $L$, size per block $s_m$, size per cache $s_c$, set of servers $V_s$, parameters for each $j\in V_s$ including GPU memory $M_j$, processing time $\tau_j$, and per-token RTTs $(t_{cj})_{c\in V_c}$}
\Output{Block placement $(\bm{a},\bm{m})$ and request routing $\bm{f}$}
\tcp{Conservative Greedy Block Placement (CG-BP):}
$m_j \leftarrow \min(\lfloor M_j/(s_m+s_c |\mathcal{R}|) \rfloor,\: L)$, $\forall j\in V_s$\; \label{Online CG:1}
$C_b\leftarrow 0,\: T_b\leftarrow \widetilde{t}_0|\mathcal{R}|$, $\forall b\in [L]$\; \label{Online CG:2}
\For{each server $j\in V_s$ in increasing order of $\widetilde{t}_j$ \nl\label{Online CG:3}}
{
\If{$\exists b\in [L]$ with $C_b<|\mathcal{R}|$ \label{Online CG:4}}
{$a_j \leftarrow \argmax_{a\in [L-m_j+1]:\:C_b<|\mathcal{R}|,\: \exists b\in \{a,\ldots,a+m_j-1\}} \sum_{b'=a}^{a+m_j-1} T_{b'}$\; \label{Online CG:5}}
\Else
{$a_j\leftarrow \argmin_{a\in [L-m_j+1]} (C_a,\ldots,C_{a+m_j-1})$\; \label{Online CG:7}}
$T_b\leftarrow T_b - (\widetilde{t}_0-\widetilde{t}_j) \min\left(\max(|\mathcal{R}|-C_b, 0), \overline{f}_j\right)$, $\forall b\in \{a_j,\ldots,a_j+m_j-1\}$\; \label{Online CG:8}
$C_b\leftarrow C_b+\overline{f}_j$, $\forall b\in \{a_j,\ldots,a_j+m_j-1\}$\; \label{Online CG:9}
}
\tcp{Waiting-penalized Shortest-path Request Routing (WS-RR):}
\For{each new request $r$ arriving from client $c$ at time $t$ \label{Online CG:10}}
{
$G^c_{\bm{a},\bm{m}}(t)\leftarrow$ the feasible routing topology for client $c$ under block placement $(\bm{a}, \bm{m})$, with a node/link set $(V^c, E^c_{\bm{a},\bm{m}})$ and a waiting-penalized cost of $t^W_{ij}(t) + \lmax t^c_{ij}$ for each $(i,j)\in E^c_{\bm{a},\bm{m}}$\; \label{Online CG:11}
$p_c(t)\leftarrow$ shortest path from the S-client to the D-client in $G^c_{\bm{a},\bm{m}}(t)$\; \label{Online CG:12}
$f^r_{ij}\leftarrow \mathbb{1}((i,j)\in p_c(t))$, $\forall (i,j)\in E$\; \label{Online CG:13}
}
\caption{Two-time-scale BPRR for Online Setting}
\vspace{-.0em}
\label{Alg:Online BPRR}
\end{algorithm}
\normalsize

\subsection{Approximation Ratio of CG-BPRR}\label{appendix:CG-BPRR Analysis}

In addition to the upper bound in Theorem~\ref{thm:CG-BPRR}, we can also lower-bound the inference time as follows. Let $\overline{m}_j:=\min(\lfloor M_j/(s_m+s_c)\rfloor,L)$ denote the maximum number of blocks that can be placed at server $j$ (while still able to serve at least one request), and $\widetilde{t}_{j,c}:= \tau_j + t_{cj}/\overline{m}_j$ denote the \emph{minimum amortized inference time} that server $j$ can provide for generating a token for client $c$. Let $j^{(c)}_k$ denote a server index such that  $\widetilde{t}_{j^{(c)}_1,c}\leq\cdots\leq \widetilde{t}_{j^{(c)}_{|V_s|},c}$. 


\begin{lemma}\label{lem:BPRR lower bound}
The minimum per-token inference time for client $c$ is lower-bounded by
\begin{align}\label{eq:T^o_c}
T^o_c := \sum_{k=1}^{K_c-1}\widetilde{t}_{j^{(c)}_k,c}\overline{m}_{j^{(c)}_k} + \widetilde{t}_{j^{(c)}_{K_c},c}\left(L-\sum_{k=1}^{K_c-1}\overline{m}_{j^{(c)}_k} \right),
\end{align}
where $K_c:= \min\{K:\: \sum_{k=1}^K \overline{m}_{j^{(c)}_k} \geq L\}$. Thus, the minimum average per-token inference time $T^o$ for the requests $\{\mathcal{R}_c\}_{c\in V_c}$ is lower-bounded by
\( 
T^o \geq {1\over |\mathcal{R}|}\sum_{c\in V_c}|\mathcal{R}_c| T^o_c.
\)
\end{lemma}

\begin{proof}[Proof of Lemma~\ref{lem:BPRR lower bound}]
It suffices to prove that the per-token inference time for client $c$ is lower-bounded by \eqref{eq:T^o_c}. First, we note that the minimum per-token inference time is lower if we relax the request routing to block-by-block routing with a per-block inference time $\widetilde{t}_{j,c}$ at server $j$, because in reality a request cannot be served more than $\overline{m}_j$ blocks at server $j$ and has to incur the entire client-server RTT $t_{cj}$ even if being processed by only a subset of the blocks at server $j$. Then, under such relaxed request routing, the minimum inference time for a token is achieved by routing to the fastest servers that collectively hold all the blocks, which are servers $j^{(c)}_1,\ldots,j^{(c)}_{K_c}$. That is, processing the first $\sum_{k=1}^{K_c-1}\overline{m}_{j^{(c)}_k}$ blocks at servers $j^{(c)}_1,\ldots,j^{(c)}_{K_c-1}$, and the remaining $L-\sum_{k=1}^{K_c-1}\overline{m}_{j^{(c)}_k}$ blocks at server $j^{(c)}_{K_c}$. This corresponds to the per-token inference time given in \eqref{eq:T^o_c}, which then lower-bounds the actual per-token inference time achievable for client $c$ under any feasible BPRR solution. 
\end{proof}

Combining Theorem~\ref{thm:CG-BPRR} and Lemma~\ref{lem:BPRR lower bound} yields an upper bound on $T^g/T^o$, which is the approximation ratio for CG-BPRR.

\subsection{Two-time-scale Algorithm for Online BPRR}\label{appendix:Online BPRR}

For completeness, we summarize the proposed BPRR algorithm for the online setting in Alg.~\ref{Alg:Online BPRR}, which is adapted from the CG-BPRR algorithm for the offline setting (Alg.~\ref{Alg:CG-BPRR}). While Alg.~\ref{Alg:Online BPRR} only runs CG-BP once at the beginning, it can be easily extended to adapt the block placement by invoking CG-BP again when the observed number of concurrent requests deviates significantly from the previously targeted value.

\section{Additional Evaluation Results}\label{appendix:Additional Evaluation}

\subsection{Additional Experiment Results}\label{appendix:Additional Experiment Results}

\begin{table}[]
\centering
\renewcommand{\arraystretch}{1} 
\setlength{\tabcolsep}{4pt} 
\scriptsize 

\vspace{3pt}

\resizebox{\textwidth}{!}{
\begin{tabular}{
>{\centering\arraybackslash}m{2cm} 
>{\centering\arraybackslash}m{2cm} 
>{\centering\arraybackslash}m{2cm} 
>{\centering\arraybackslash}m{2cm} 
>{\centering\arraybackslash}m{2cm} 
>{\centering\arraybackslash}m{2cm} 
}
\toprule
\multirow{2}{*}{\textbf{Client Location}} 
& \multirow{2}{*}{\textbf{Algorithm}} & \multicolumn{2}{c}{\textbf{0.1 requests/s}} & \multicolumn{2}{c}{\textbf{0.5 requests/s}} \\
\cmidrule(lr){3-4} \cmidrule(lr){5-6}
& & \textbf{$\lmax=64$} & \textbf{$\lmax=128$} & \textbf{$\lmax=64$} & \textbf{$\lmax=128$} \\
\midrule
\multirow{2}{*}{Cluster0} & PETALS & $275.51\ (252.61)$ & $455.61\ (427.72)$ & $315.53\ (252.61)$ & $508.52\ (427.72)$ \\
& Proposed                         & $60.43\ (73.51)$   & $62.98\ (60.91)$   & $66.94\ (73.51)$   & $69.62\ (60.91)$ \\
\midrule
\multirow{2}{*}{Cluster1} & PETALS & $236.34\ (252.51)$ & $463.39\ (424.94)$ & $298.27\ (252.51)$ & $489.85\ (424.06)$ \\
& Proposed                         & $55.47\ (85.50)$   & $59.67\ (72.06)$   & $57.37\ (85.50)$   & $65.96\ (72.06)$ \\
\midrule
\multirow{2}{*}{Cluster2} & PETALS & $245.70\ (251.95)$ & $470.33\ (404.42)$ & $259.69\ (251.95)$ & $480.43\ (404.42)$ \\
& Proposed                         & $50.96\ (73.51)$   & $55.72\ (60.91)$   & $55.36\ (73.51)$   & $59.86\ (60.91)$ \\
\bottomrule
\end{tabular}
}
\vspace{-0.5em}
\caption{{Average inference time for the first token (s) under the  configuration in Table~\ref{tab: network_link_property} (\textit{$\lmax^I=20$}; 100 requests; MATLAB results shown in parentheses).}}
\label{tab: TTFT_clustered}
\vspace{-.0em}
\end{table}

\begin{table}[]
\centering
\renewcommand{\arraystretch}{1} 
\setlength{\tabcolsep}{4pt} 
\scriptsize 

\vspace{3pt}

\resizebox{\textwidth}{!}{%
\begin{tabular}{
 >{\centering\arraybackslash}m{2cm} 
 >{\centering\arraybackslash}m{2cm} 
 >{\centering\arraybackslash}m{2cm} 
 >{\centering\arraybackslash}m{2cm} 
 >{\centering\arraybackslash}m{2cm} 
 >{\centering\arraybackslash}m{2cm} 
}
\toprule
\multirow{2}{*}{\textbf{Client Location}} 
& \multirow{2}{*}{\textbf{Algorithm}} & \multicolumn{2}{c}{\textbf{0.1 requests/s}} & \multicolumn{2}{c}{\textbf{0.5 requests/s}} \\
\cmidrule(lr){3-4} \cmidrule(lr){5-6}
& & \textbf{$\lmax=64$} & \textbf{$\lmax=128$} & \textbf{$\lmax=64$} & \textbf{$\lmax=128$} \\
\midrule
\multirow{2}{*}{Cluster0} & PETALS & 1.96 (1.40) & 1.21 (1.41) & 1.37 (1.40) & 1.18 (1.41) \\
                        & Proposed & 0.99 (0.45) & 0.94 (0.45) & 0.96 (0.45) & 0.81 (0.45) \\
\midrule
\multirow{2}{*}{Cluster1} & PETALS   & 1.78 (1.25) & 0.99 (1.27) & 0.91 (1.25) & 0.97 (1.27) \\
& Proposed                           & 0.93 (0.32) & 0.58 (0.26) & 0.98 (0.32) & 0.60 (0.26) \\
\midrule
\multirow{2}{*}{Cluster2} & PETALS   & 1.49 (0.93) & 1.18 (0.91) & 1.29 (0.93) & 1.51 (0.91) \\
& Proposed                           & 1.01 (0.45) & 0.88 (0.45) & 1.09 (0.45) & 0.91 (0.45) \\
\bottomrule
\end{tabular}
} 
\vspace{-0.5em}
\caption{{Average inference time for the remaining token (s) under the configuration in Table~\ref{tab: network_link_property} (\textit{$\lmax^I=20$}; 100 requests; MATLAB results shown in parentheses).}}
\label{tab: time_per_remaining_token_clustered}
\vspace{-.0em}
\end{table}

\begin{table}[]
\centering
\renewcommand{\arraystretch}{1} 
\setlength{\tabcolsep}{4pt} 
\scriptsize 

\vspace{3pt}

\resizebox{\textwidth}{!}{%
\begin{tabular}{
  >{\centering\arraybackslash}m{2cm} 
  >{\centering\arraybackslash}m{2cm} 
  >{\centering\arraybackslash}m{2cm} 
  >{\centering\arraybackslash}m{2cm} 
  >{\centering\arraybackslash}m{2cm} 
  >{\centering\arraybackslash}m{2cm} 
} 
\toprule
\multirow{2}{*}{\textbf{Topology}} 
  & \multirow{2}{*}{\textbf{Algorithm}} 
    & \multicolumn{2}{c}{\textbf{0.1 requests/s}} 
    & \multicolumn{2}{c}{\textbf{0.5 requests/s}} \\ 
\cmidrule(lr){3-4} \cmidrule(lr){5-6} 
  & & \textbf{$\lmax=64$} & \textbf{$\lmax=128$} 
    & \textbf{$\lmax=64$} & \textbf{$\lmax=128$} \\
\midrule
\multirow{2}{*}{AboveNet} 
  & PETALS   & $270.96\ (254.74)$ & $408.59\ (316.21)$ & $279.38\ (264.81)$ & $486.68\ (412.72)$ \\ 
  & Proposed & $76.48\ (73.69)$ & $89.15\ (104.19)$ & $81.65\ (81.12)$ & $91.41\ (75.78)$ \\
\midrule
\multirow{2}{*}{BellCanada} 
  & PETALS   & $360.94\ (353.12)$ & $382.01\ (354.06)$ & $378.10\ (353.46)$ & $417.17\ (353.72)$ \\ 
  & Proposed & $53.42\ (62.75)$ & $59.73\ (62.68)$ & $63.78\ (62.67)$ & $66.84\ (62.70)$ \\
\midrule
\multirow{2}{*}{GTS-CE} 
  & PETALS   & $393.93\ (353.48)$ & $478.71\ (354.05)$ & $389.35\ (353.46)$ & $505.90\ (353.79)$ \\ 
  & Proposed & $55.30\ (62.84)$ & $58.59\ (62.66)$ & $58.17\ (62.65)$ & $61.22\ (62.65)$ \\
\bottomrule
\end{tabular}%
}
\vspace{-0.5em}
\caption{{Average inference time for the first token (s) under the topologies in Table~\ref{tab:topo_info} ($\lmax^I = 20$; 100 requests; MATLAB results shown in parentheses).} }  
\label{tab:TTFT_scattered}
\vspace{-.0em}
\end{table}

\begin{table}[]
\centering
\renewcommand{\arraystretch}{1} 
\setlength{\tabcolsep}{4pt} 
\scriptsize 

\vspace{3pt}

\resizebox{\textwidth}{!}{%
\begin{tabular}{
  >{\centering\arraybackslash}m{2cm} 
  >{\centering\arraybackslash}m{2cm} 
  >{\centering\arraybackslash}m{2cm} 
  >{\centering\arraybackslash}m{2cm} 
  >{\centering\arraybackslash}m{2cm} 
  >{\centering\arraybackslash}m{2cm} 
} 
\toprule
\multirow{2}{*}{\textbf{Topology}} 
  & \multirow{2}{*}{\textbf{Algorithm}} 
    & \multicolumn{2}{c}{\textbf{0.1 requests/s}} 
    & \multicolumn{2}{c}{\textbf{0.5 requests/s}} \\ 
\cmidrule(lr){3-4} \cmidrule(lr){5-6} 
  & & \textbf{$\lmax=64$} & \textbf{$\lmax=128$} 
    & \textbf{$\lmax=64$} & \textbf{$\lmax=128$} \\
\midrule
\multirow{2}{*}{AboveNet} 
  & PETALS   & $0.77\ (0.79)$ & $0.84\ (0.92)$ & $0.91\ (0.98)$ & $0.78\ (0.88)$ \\
  & Proposed & 0.67 (0.52) & 0.75 (0.55) & 0.70 (0.57) & 0.64 (0.46) \\
\midrule
\multirow{2}{*}{BellCanada} 
  & PETALS   & 0.69 (0.53) & 0.84 (0.73) & 0.80 (0.68) & 0.91 (0.66) \\
  & Proposed & 0.50 (0.44) & 0.79 (0.44) & 0.51 (0.44) & 0.59 (0.44)\\
\midrule
\multirow{2}{*}{GTS-CE} 
  & PETALS   & 0.92 (0.61) & 0.96 (0.71) & 0.83 (0.44) & 0.95 (0.67) \\
  & Proposed & 0.52 (0.43) & 0.50 (0.42) & 0.49 (0.42) & 0.59 (0.43) \\
\bottomrule
\end{tabular}%
}
\vspace{-0.5em}
\caption{{Average inference time for the remaining tokens (s) under the topologies in Table~\ref{tab:topo_info} ($\lmax^I = 20$; 100 requests; MATLAB results shown in parentheses).}}
\label{tab:time_per_remaining_token_scattered}
\vspace{-.0em}
\end{table}

Tables~\ref{tab: TTFT_clustered}--\ref{tab: time_per_remaining_token_clustered} present the breakdown of the total average inference time in Table~\ref{tab: time_per_token_clustered} into the inference time for the first token and the inference time for each of the remaining tokens. Similarly, Tables~\ref{tab:TTFT_scattered}--\ref{tab:time_per_remaining_token_scattered} provide the breakdown of the total average inference time in Table~\ref{tab:time_per_token_scattered}.

\subsection{Additional Simulation Results}\label{appendix:Additional Simulation Results}

\begin{figure}[t!]
    \centering
    \begin{subfigure}[t]{0.31\textwidth}
        \centering
        \includegraphics[width=\linewidth]{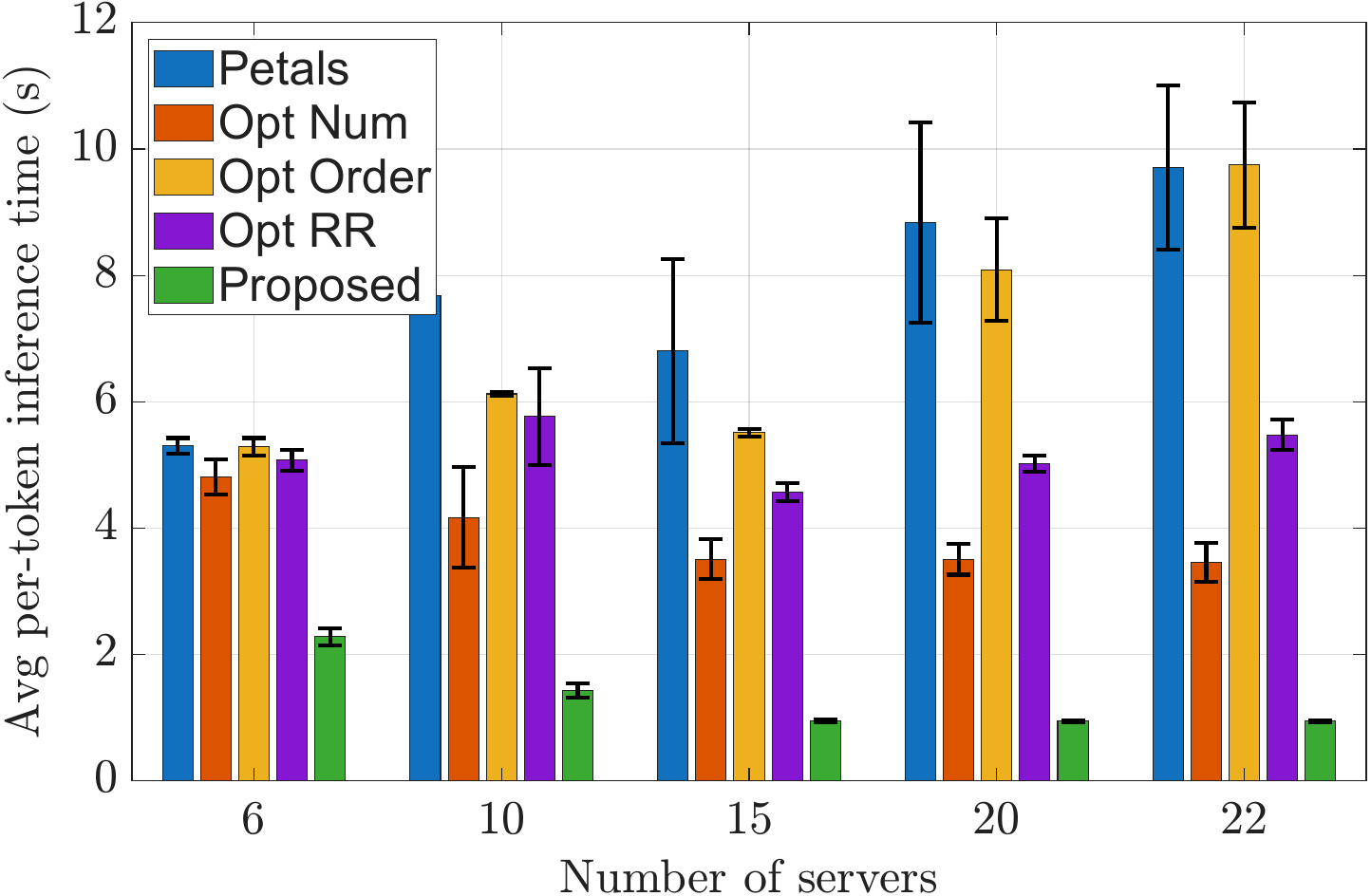}
        \caption{AboveNet}
    \end{subfigure}\hfill
    \begin{subfigure}[t]{0.31\textwidth}
        \centering
        \includegraphics[width=\linewidth]{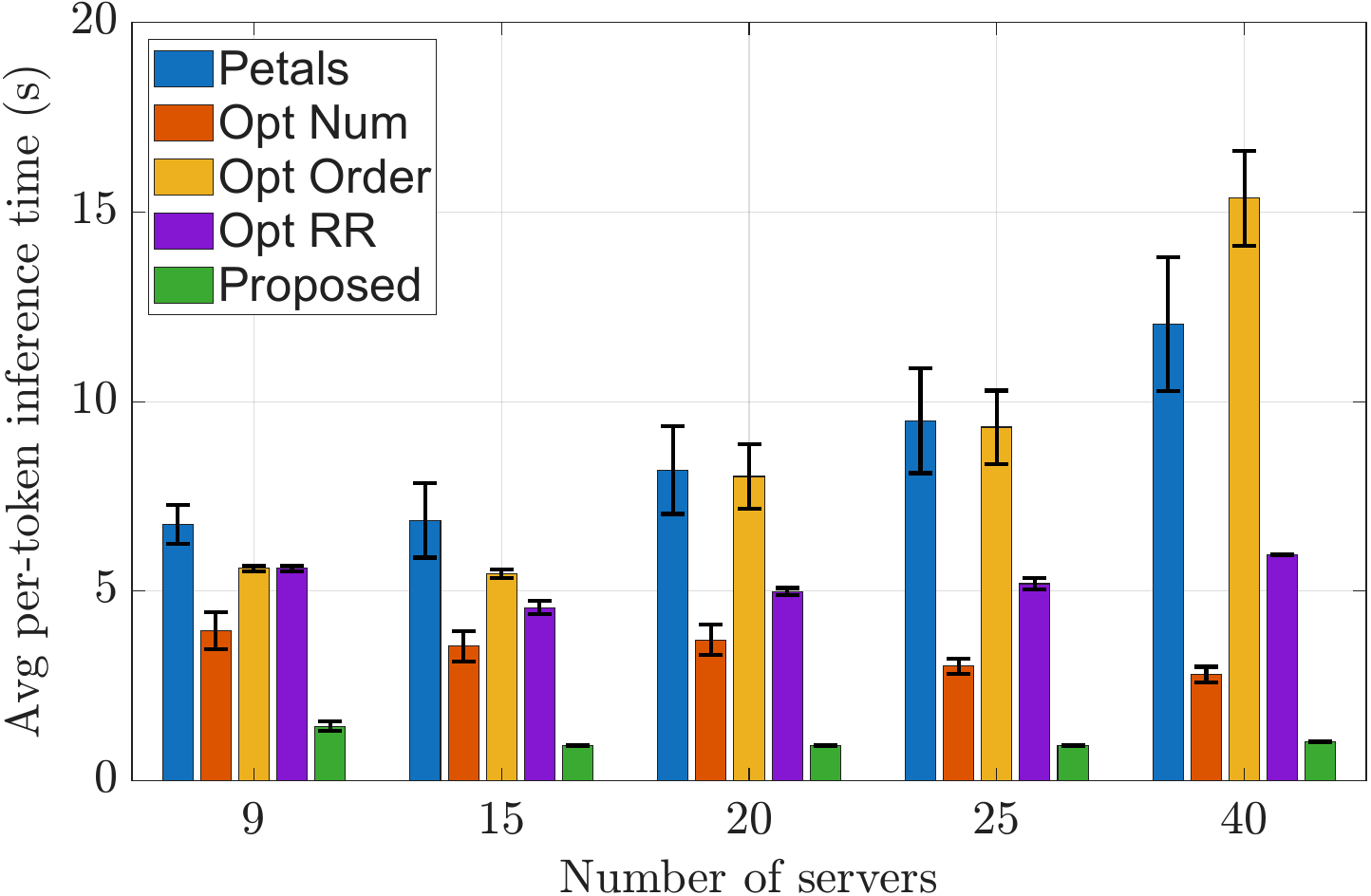}
        \caption{BellCanada}
    \end{subfigure}\hfill
    \begin{subfigure}[t]{0.31\textwidth}
        \centering
        \includegraphics[width=\linewidth]{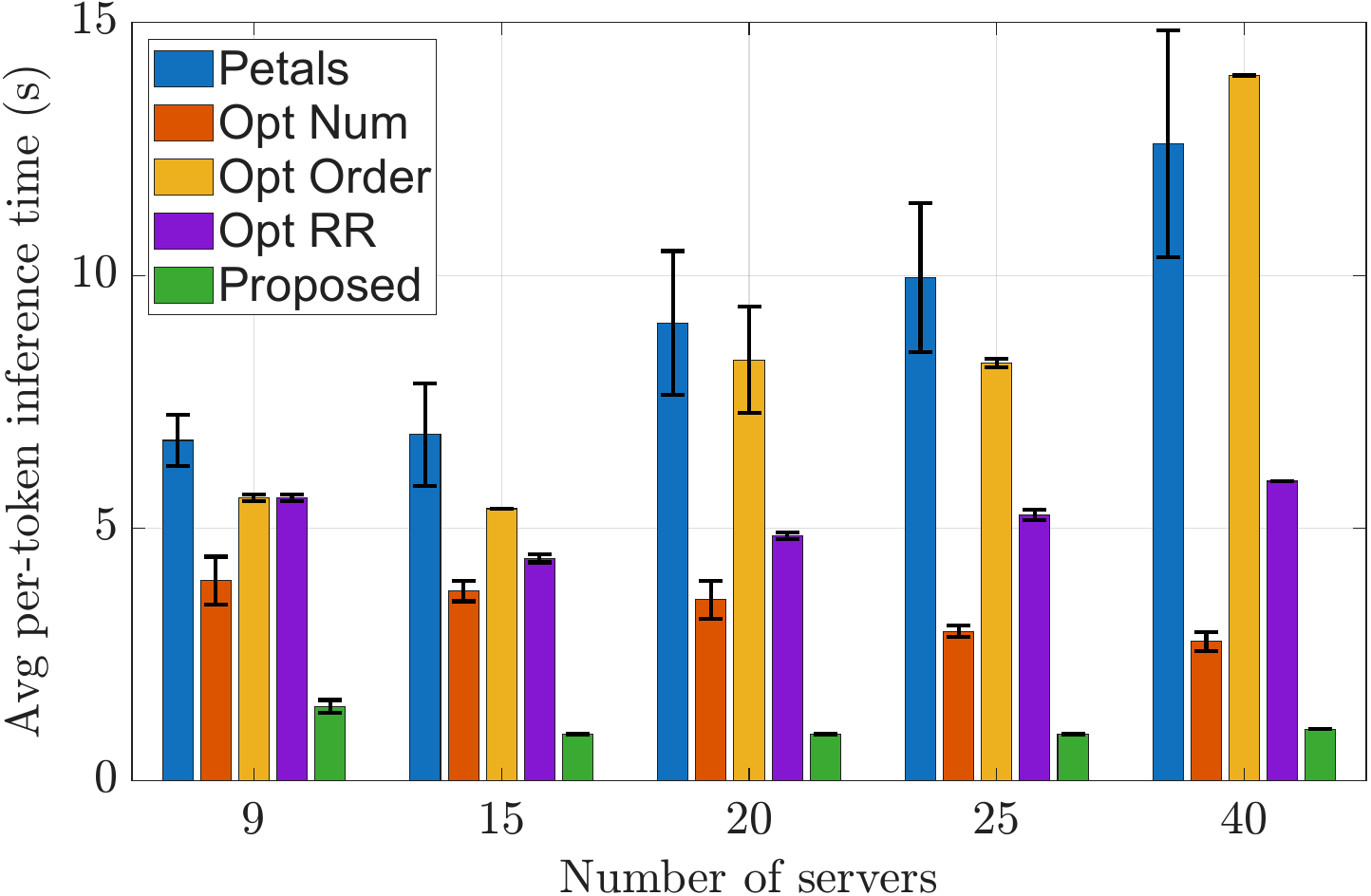}
        \caption{GTS-CE}
    \end{subfigure}

    \vspace{-0.6em}
    \caption{Inference time per token when varying the number of servers $C$ and rate $\lambda = {(0.1/9)}\cdot C$ ($N_R=100$, $\eta=0.2$, $\lmax^I=20$, $\lmax = 128$).}
    \label{fig:inference_time_vary_Clamda}
    \vspace{-0.4em}
\end{figure}

To evaluate the performance of our solution as the system scales, we proportionally increase the number of servers and the request rate as in Fig.~\ref{fig:inference_time_vary_Clamda}, and the result shows a trend of widening performance gap between our solution and PETALS.

\begin{figure}[t!]
    \centering
    \begin{subfigure}[t]{0.31\textwidth}
        \centering
        \includegraphics[width=\linewidth]{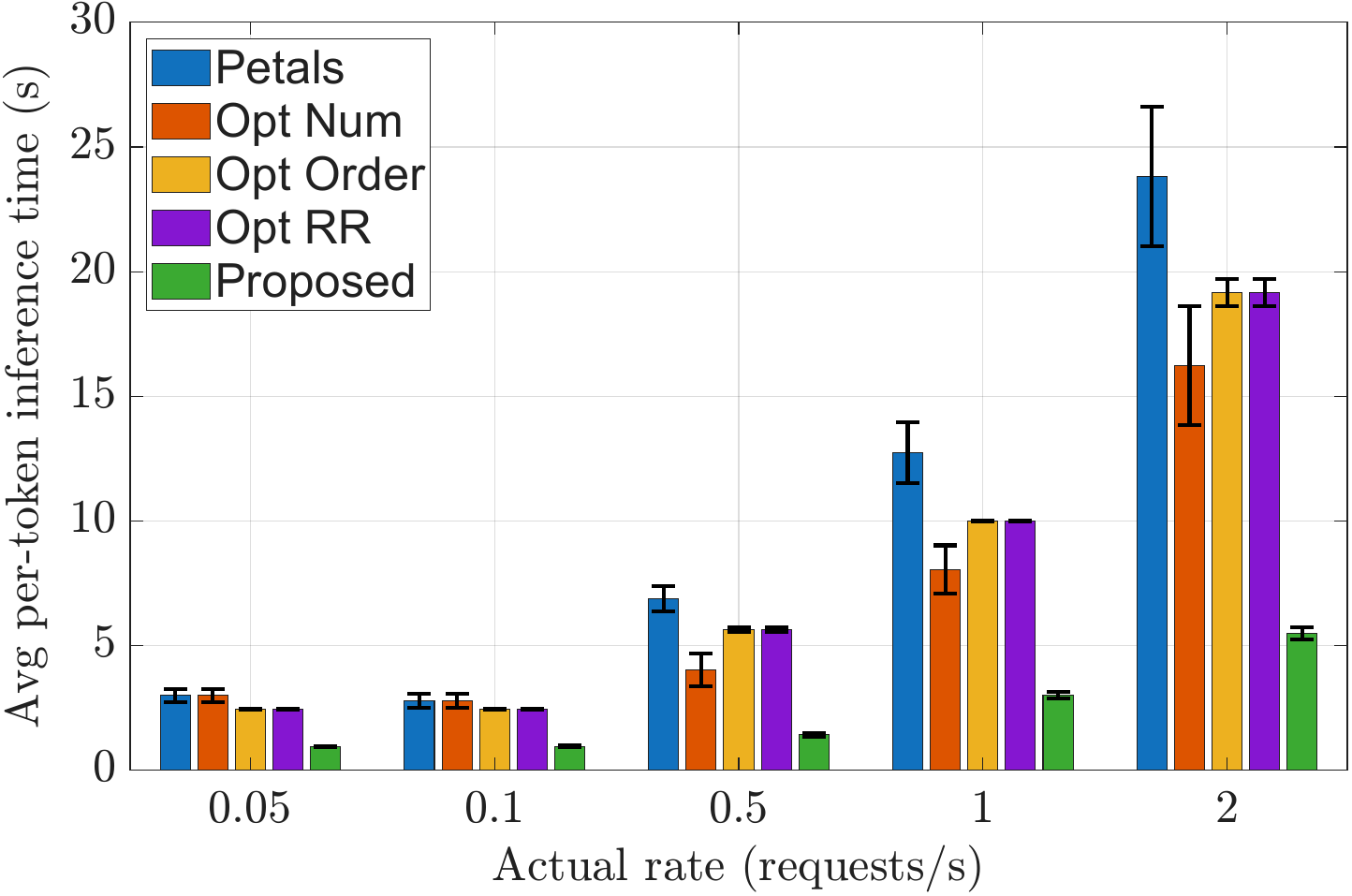}
        \caption{AboveNet}
    \end{subfigure}\hfill
    \begin{subfigure}[t]{0.31\textwidth}
        \centering
        \includegraphics[width=\linewidth]{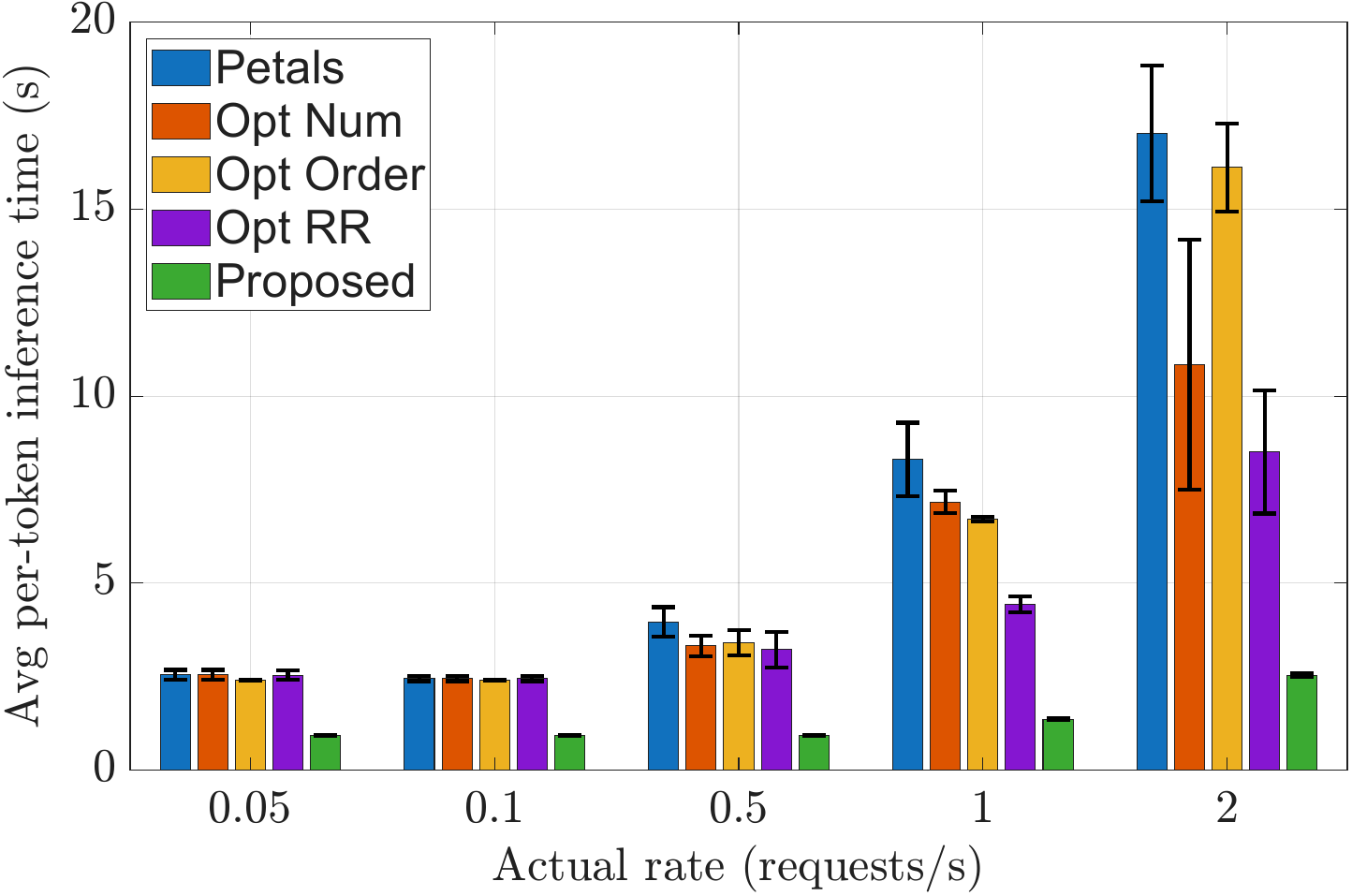}
        \caption{BellCanada}
    \end{subfigure}\hfill
    \begin{subfigure}[t]{0.31\textwidth}
        \centering
        \includegraphics[width=\linewidth]{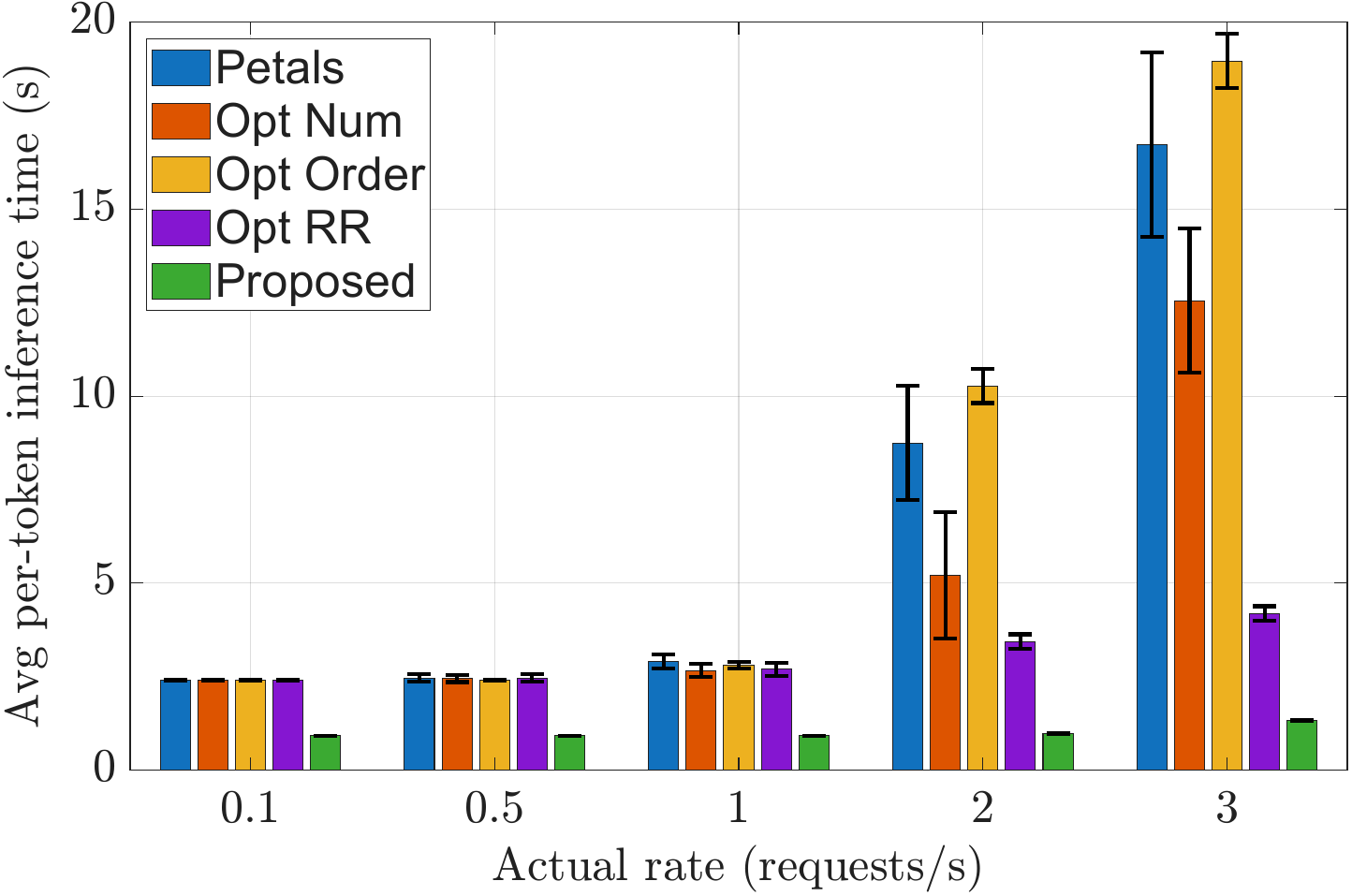}
        \caption{GTS-CE}
    \end{subfigure}

    \vspace{-0.6em}
    \caption{Inference time per token when setting $|\mathcal{R}|$ according to $\lambda_{base}$ and varying actual request rate $\lambda$ ($C = 0.4 \cdot \mbox{\#nodes}$, $\lambda_{base}=0.5$, $N_R=200\cdot \lambda$, $\eta = 0.2$, $\lmax^I=20$, $\lmax = 128$).}
    \label{fig:inference_time_vary_actual_lambda}
    \vspace{-0.4em}
\end{figure}

We further evaluate the sensitivity of our solution to the configuration of the load parameter $|\mathcal{R}|$. Under a fixed  $|\mathcal{R}|$ computed for a predicted rate $\lambda_{base}$, Fig.~\ref{fig:inference_time_vary_actual_lambda} shows the performance under different actual rates, which shows that in comparison with setting $|\mathcal{R}|$ according to the actual rates as in Fig.~\ref{fig:inference_time_vary_lambda}, using a fixed $|\mathcal{R}|$ can lead to increased inference time when the actual rate is higher than expected, due to not reserving enough memory for attention caches and thus causing some requests to incur extensive waiting. Nevertheless, the increase for our solution (`Proposed') is much smaller than that for the BPRR algorithm of PETALS (`Optimized Number'), demonstrating the robustness of our solution to load prediction.

\begin{figure}[t!]
\centering
\begin{subfigure}[t]{0.31\textwidth}
  \centering
  \includegraphics[width=\linewidth]{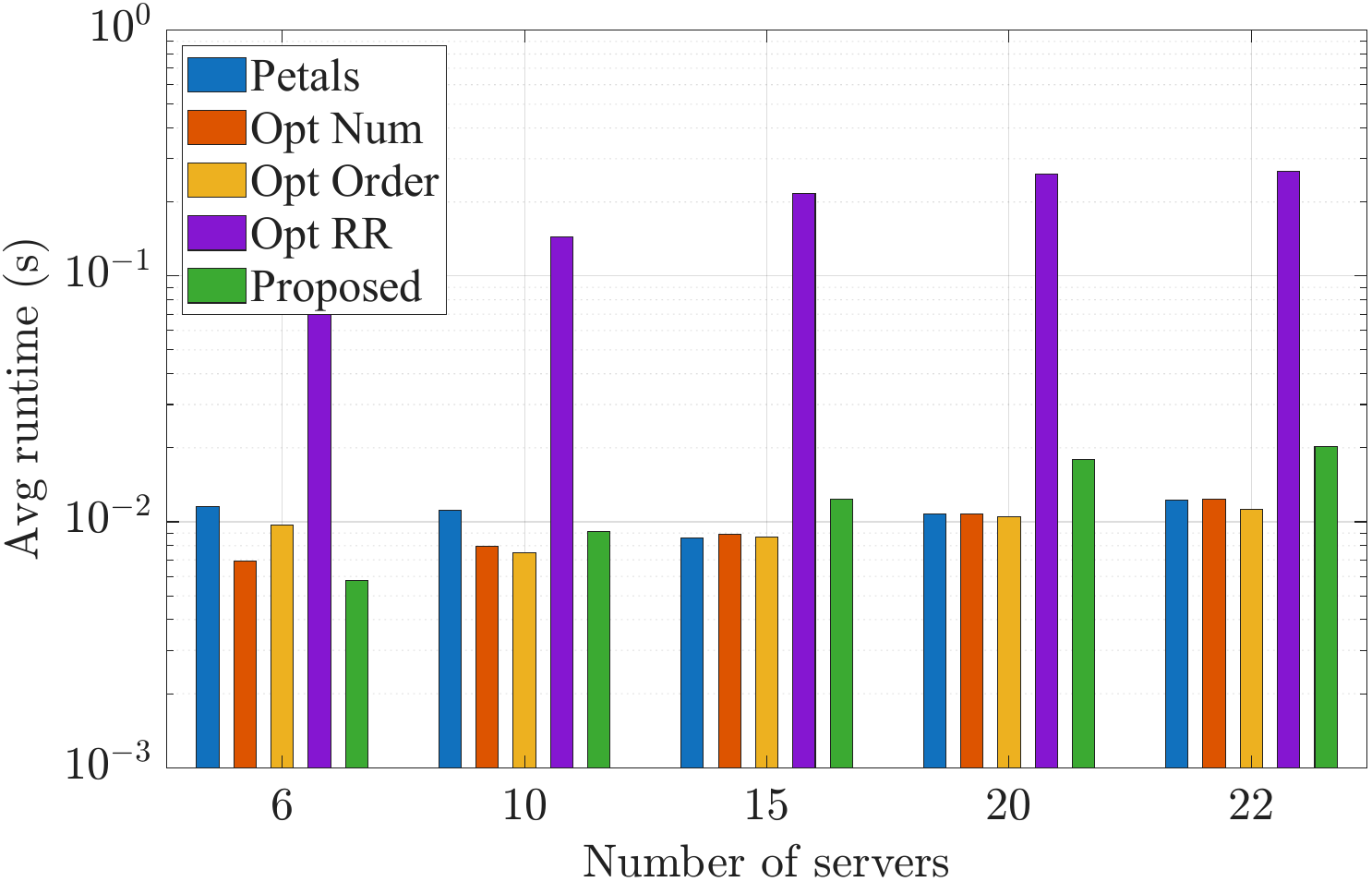}
  \caption{AboveNet}
\end{subfigure}\hfill
\begin{subfigure}[t]{0.31\textwidth}
  \centering
  \includegraphics[width=\linewidth]{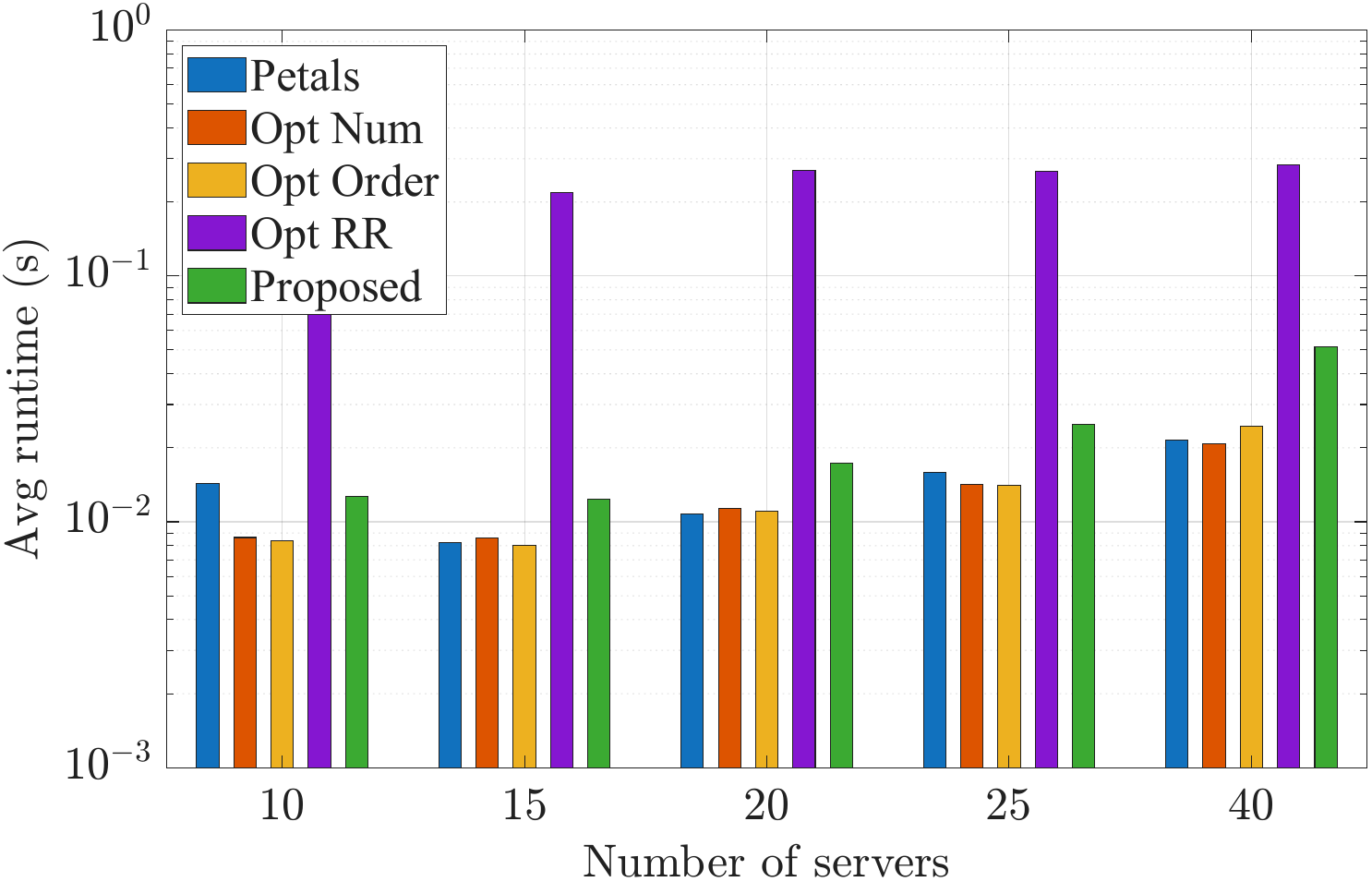}
  \caption{BellCanada}
\end{subfigure}\hfill
\begin{subfigure}[t]{0.31\textwidth}
  \centering
  \includegraphics[width=\linewidth]{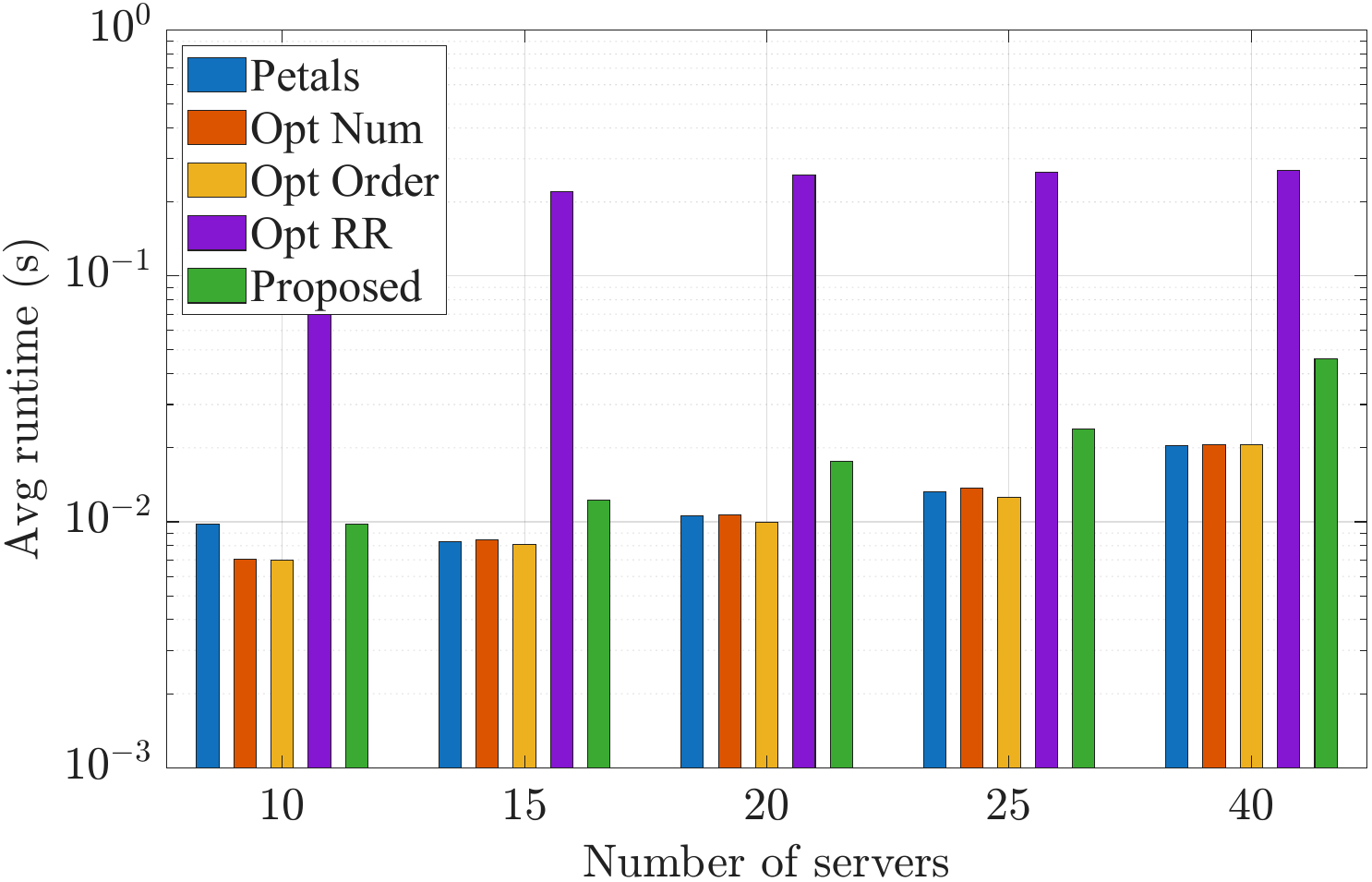}
  \caption{GTS-CE}
\end{subfigure}

\vspace{-0.6em}
\caption{Algorithm running time when varying \#servers $C$ ($\eta = 0.2$, $\lambda=0.5$, $N_R = 100$, $\lmax^I=20$, $\lmax = 128$).}
\label{fig:running_time_vary_C}
\vspace{-0.4em}
\end{figure}

\begin{figure}[t!]
\centering
\begin{subfigure}[t]{0.31\textwidth}
  \centering
  \includegraphics[width=\linewidth]{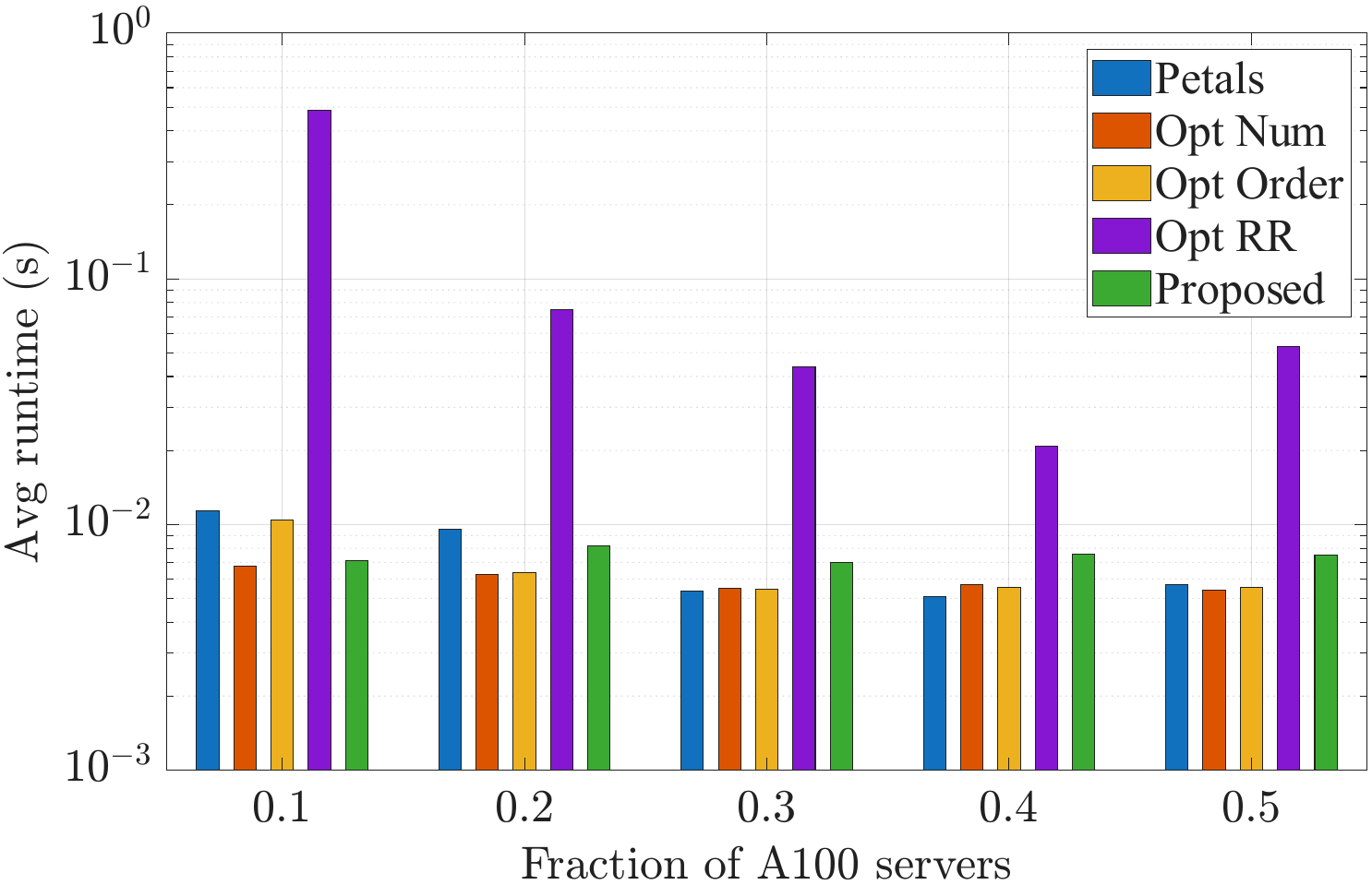}
  \caption{AboveNet}
\end{subfigure}\hfill
\begin{subfigure}[t]{0.31\textwidth}
  \centering
  \includegraphics[width=\linewidth]{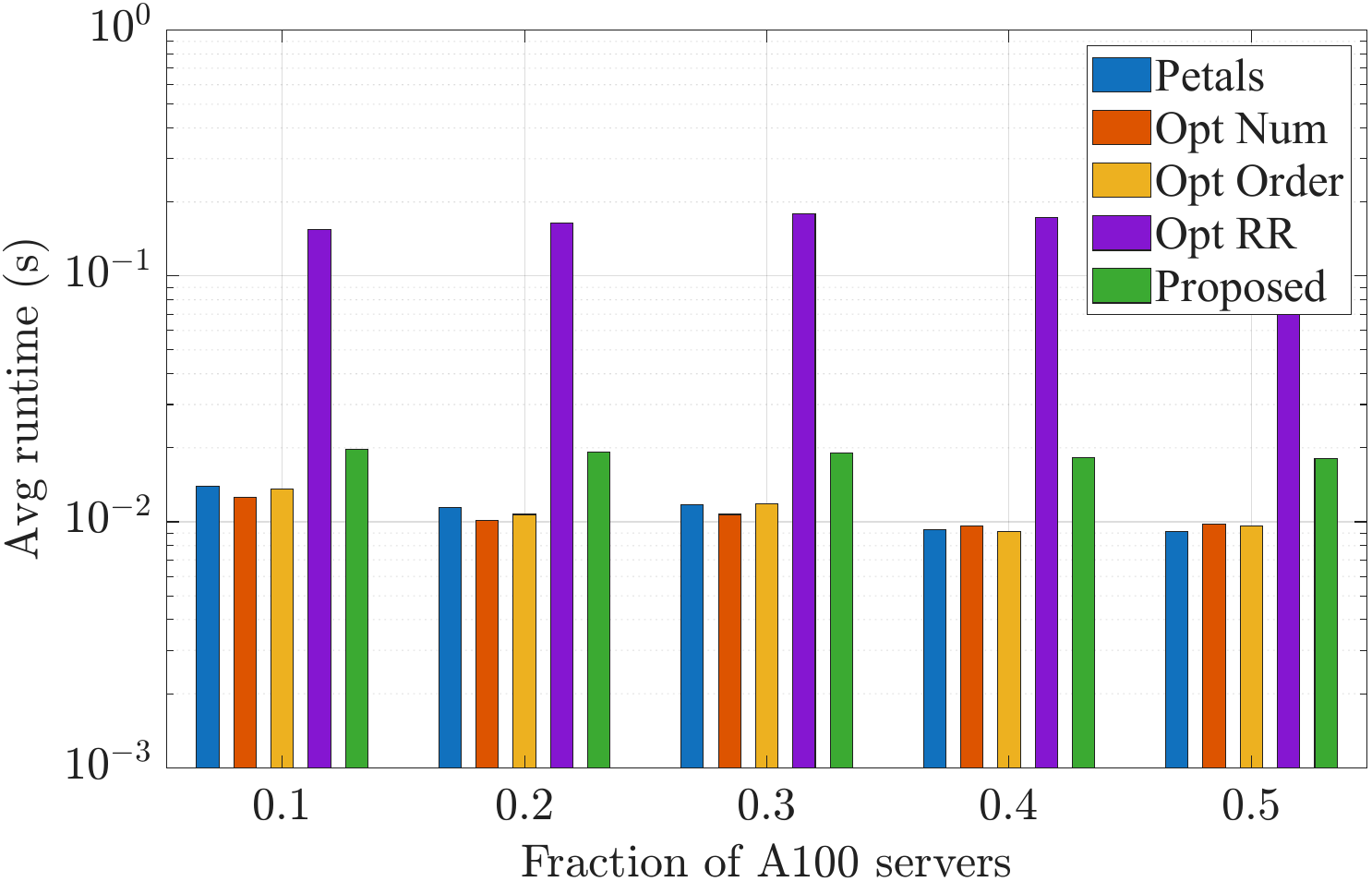}
  \caption{BellCanada}
\end{subfigure}\hfill
\begin{subfigure}[t]{0.31\textwidth}
  \centering
  \includegraphics[width=\linewidth]{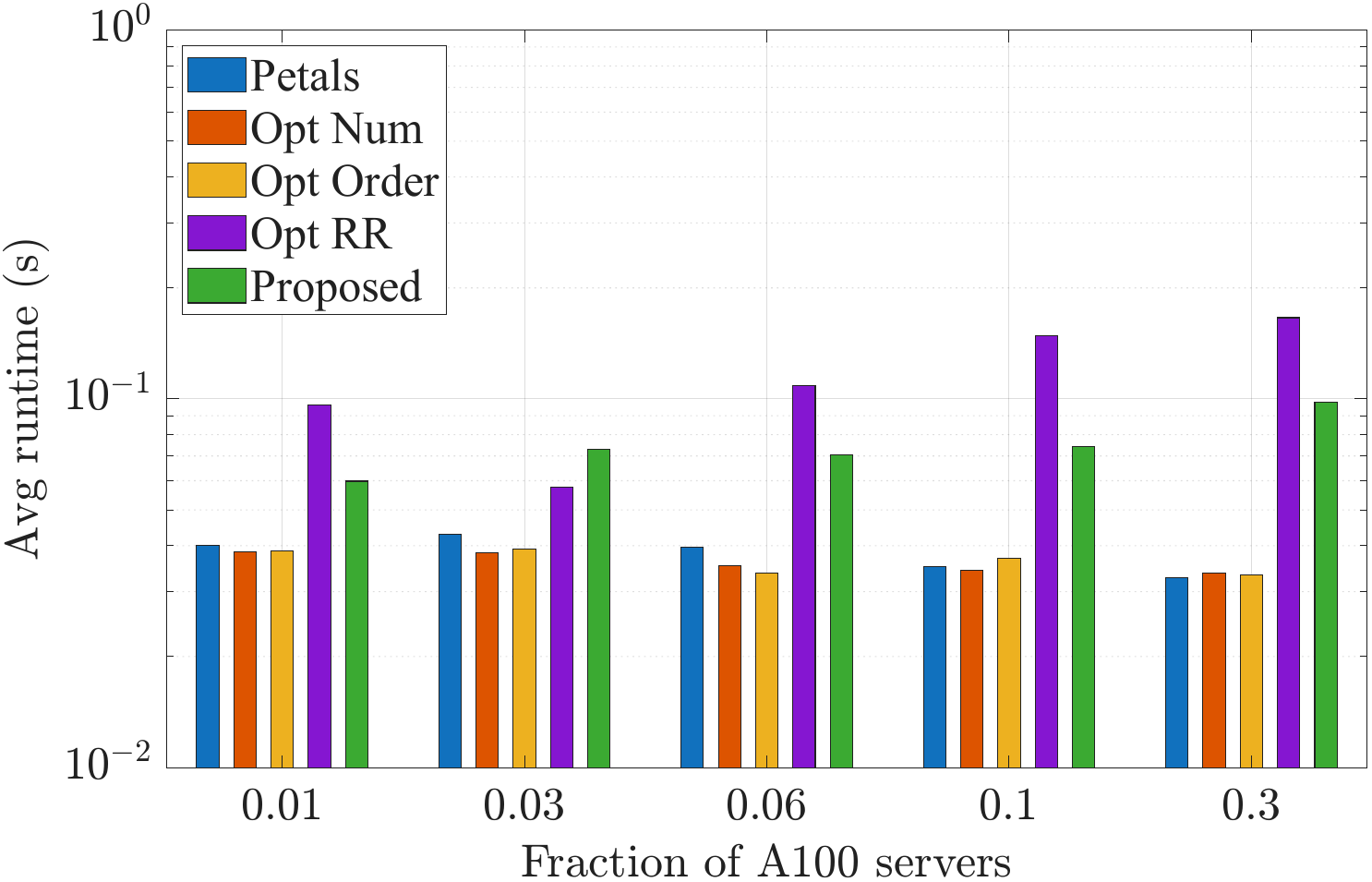}
  \caption{GTS-CE}
\end{subfigure}

\vspace{-0.6em}
\caption{Algorithm running time when varying frac. of high-performance servers $\eta$ ($C = 0.4 \cdot \mbox{\#nodes}$, $\lambda=0.5$, $N_R=100$, $\lmax^I=20$, $\lmax = 128$).}
\label{fig:running_time_vary_eta}
\vspace{-0.4em}
\end{figure}

\begin{figure}[t!]
\centering
\begin{subfigure}[t]{0.31\textwidth}
  \centering
  \includegraphics[width=\linewidth]{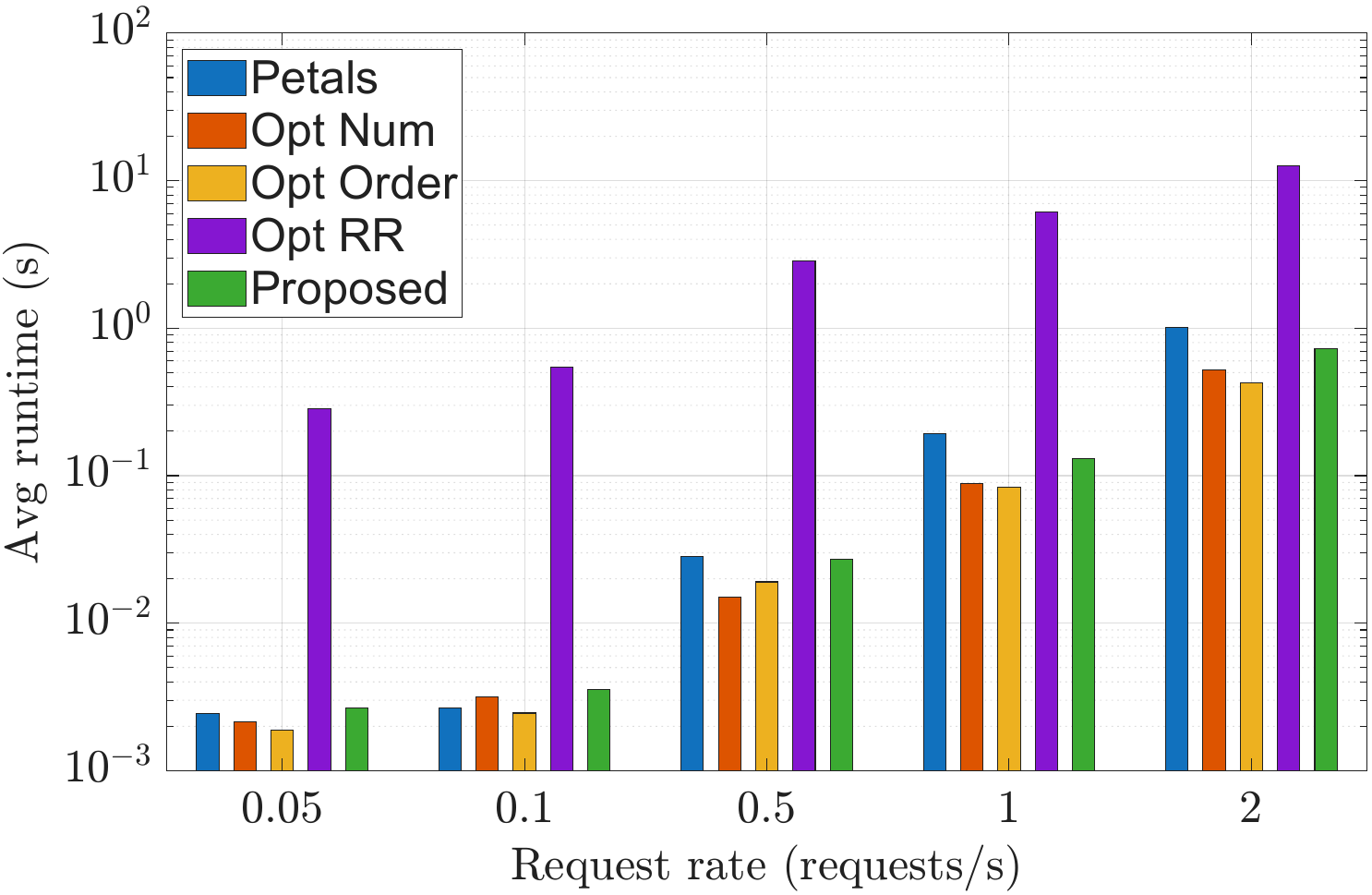}
  \caption{AboveNet}
\end{subfigure}\hfill
\begin{subfigure}[t]{0.31\textwidth}
  \centering
  \includegraphics[width=\linewidth]{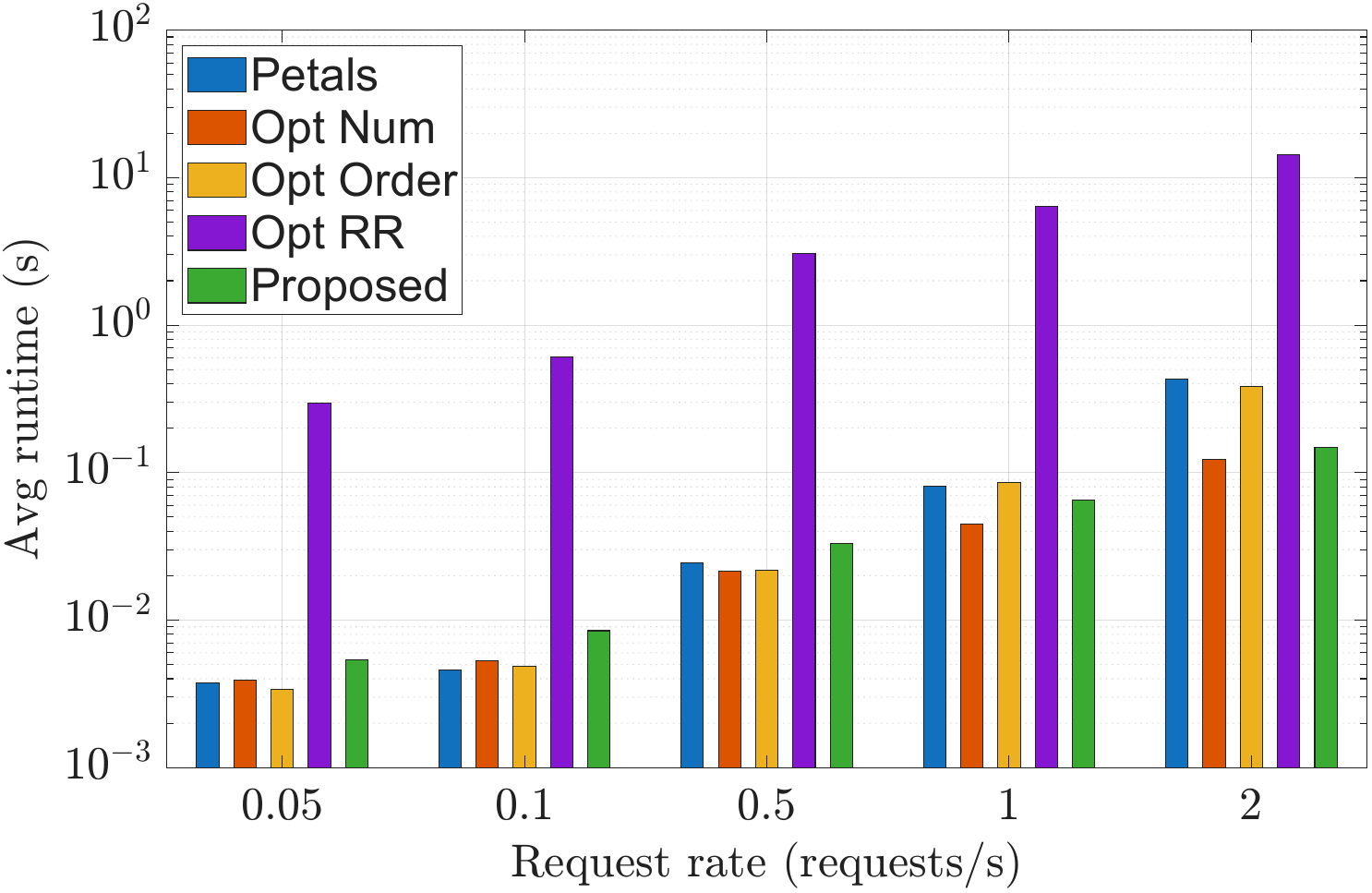}
  \caption{BellCanada}
\end{subfigure}\hfill
\begin{subfigure}[t]{0.31\textwidth}
  \centering
  \includegraphics[width=\linewidth]{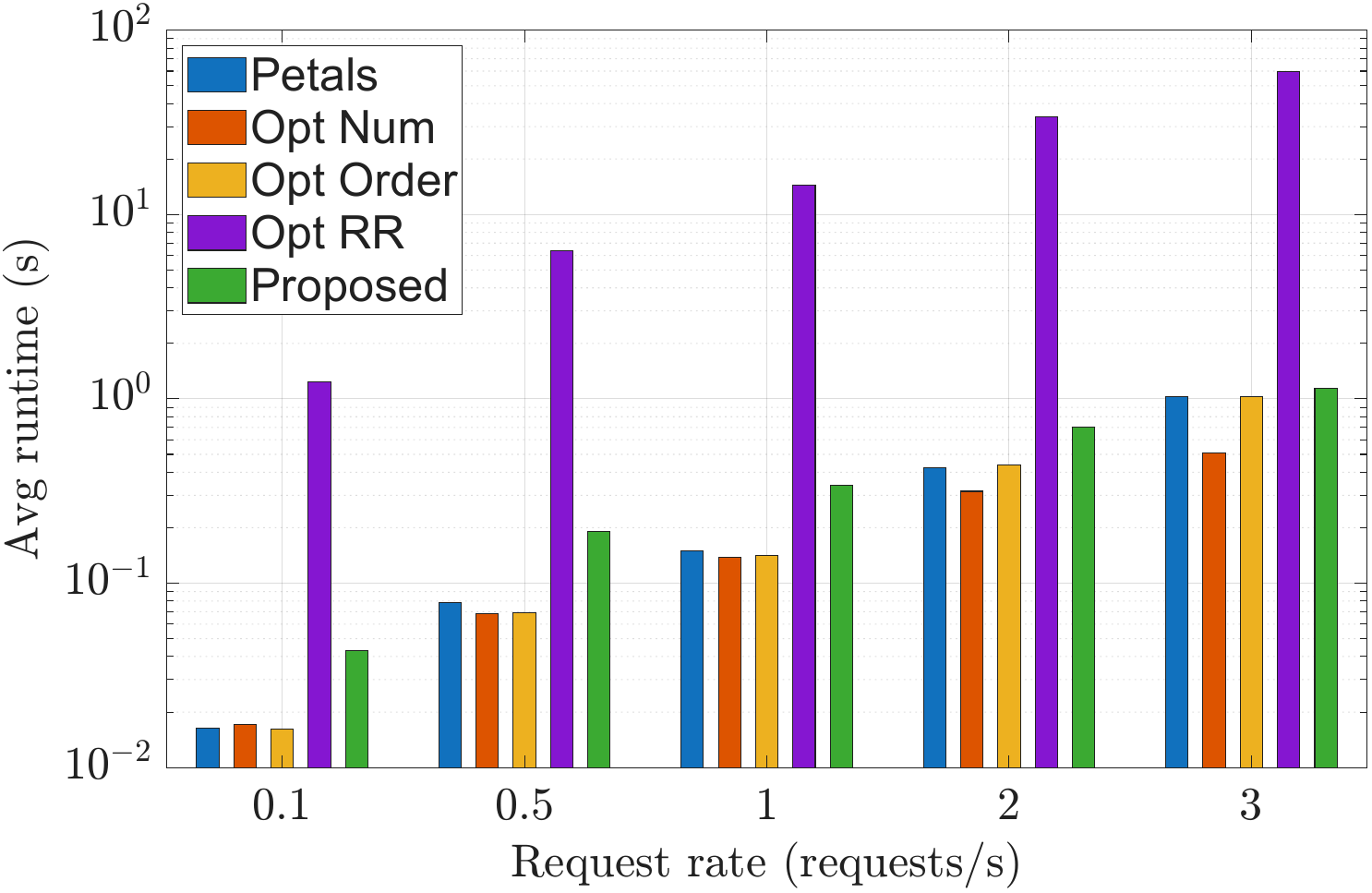}
  \caption{GTS-CE}
\end{subfigure}

\vspace{-0.6em}
\caption{Algorithm running time when varying request rate $\lambda$ with $N_R=200\cdot \lambda$ ($C = 0.4 \cdot \mbox{\#nodes}$, $\eta = 0.2$, $\lmax^I=20$, $\lmax = 128$).}
\label{fig:running_time_vary_lambda}
\vspace{-0.4em}
\end{figure}

\begin{figure}[t!]
\centering
\begin{subfigure}[t]{0.31\textwidth}
  \centering
  \includegraphics[width=\linewidth]{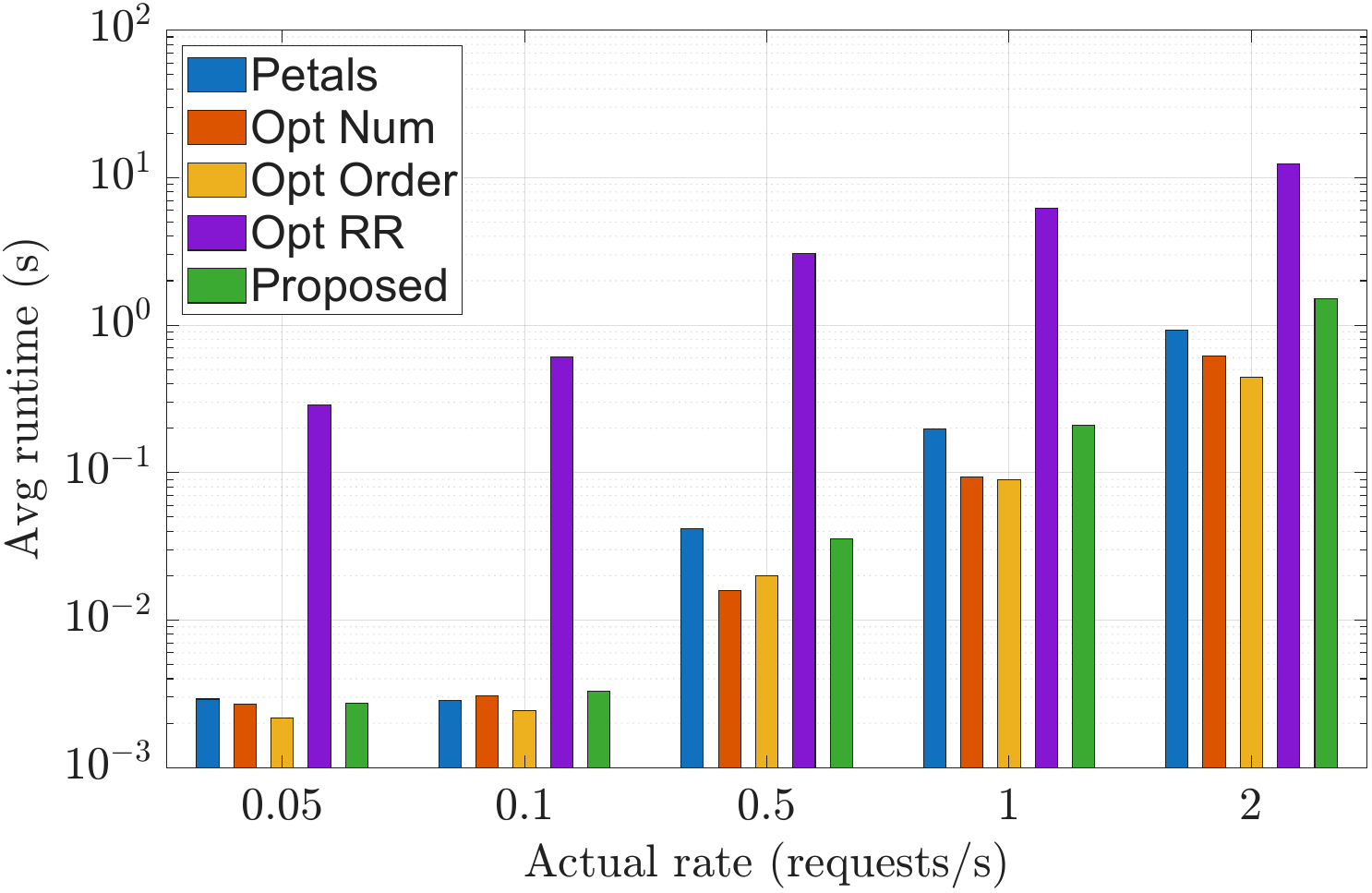}
  \caption{AboveNet}
\end{subfigure}\hfill
\begin{subfigure}[t]{0.31\textwidth}
  \centering
  \includegraphics[width=\linewidth]{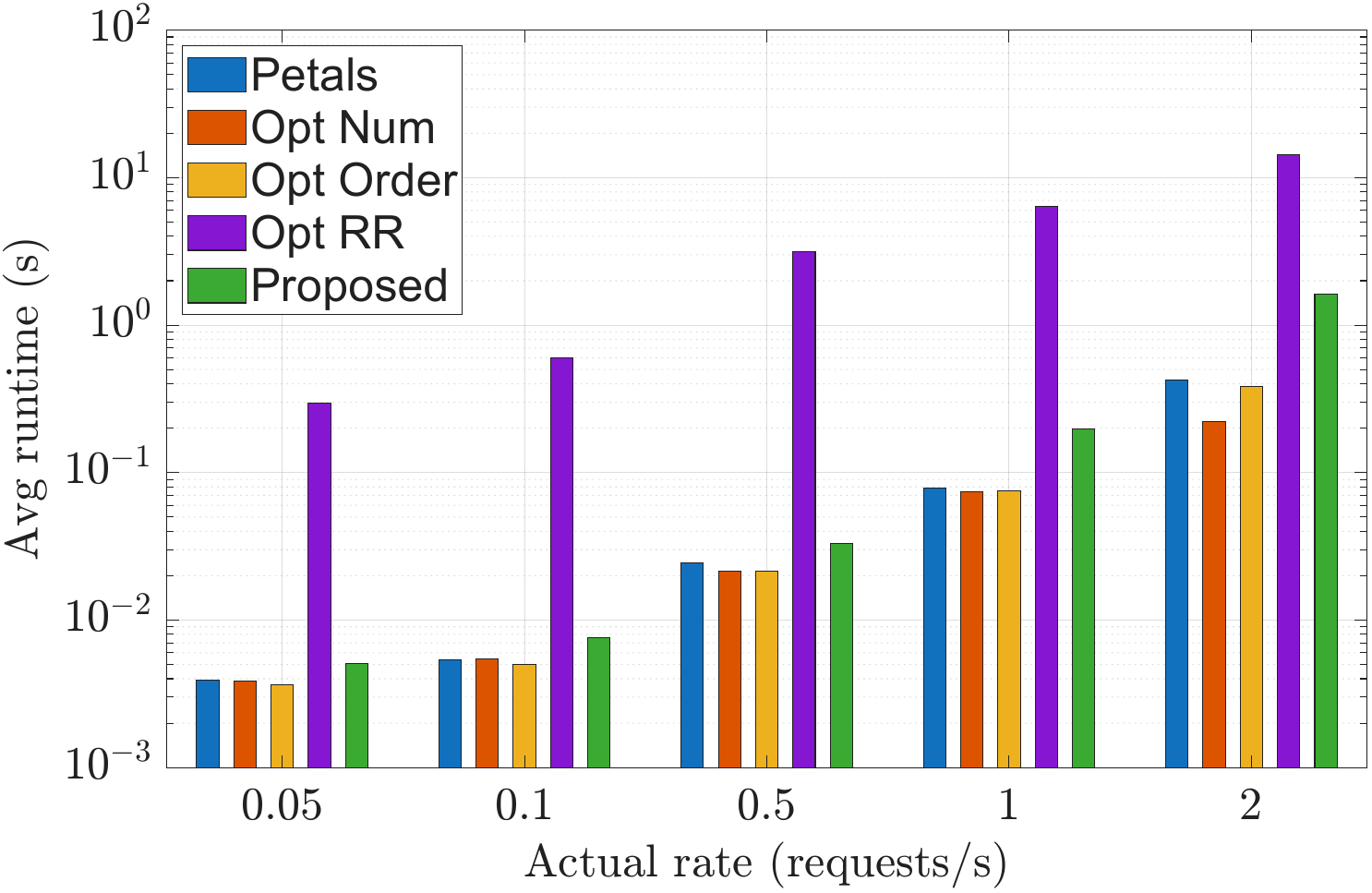}
  \caption{BellCanada}
\end{subfigure}\hfill
\begin{subfigure}[t]{0.31\textwidth}
  \centering
  \includegraphics[width=\linewidth]{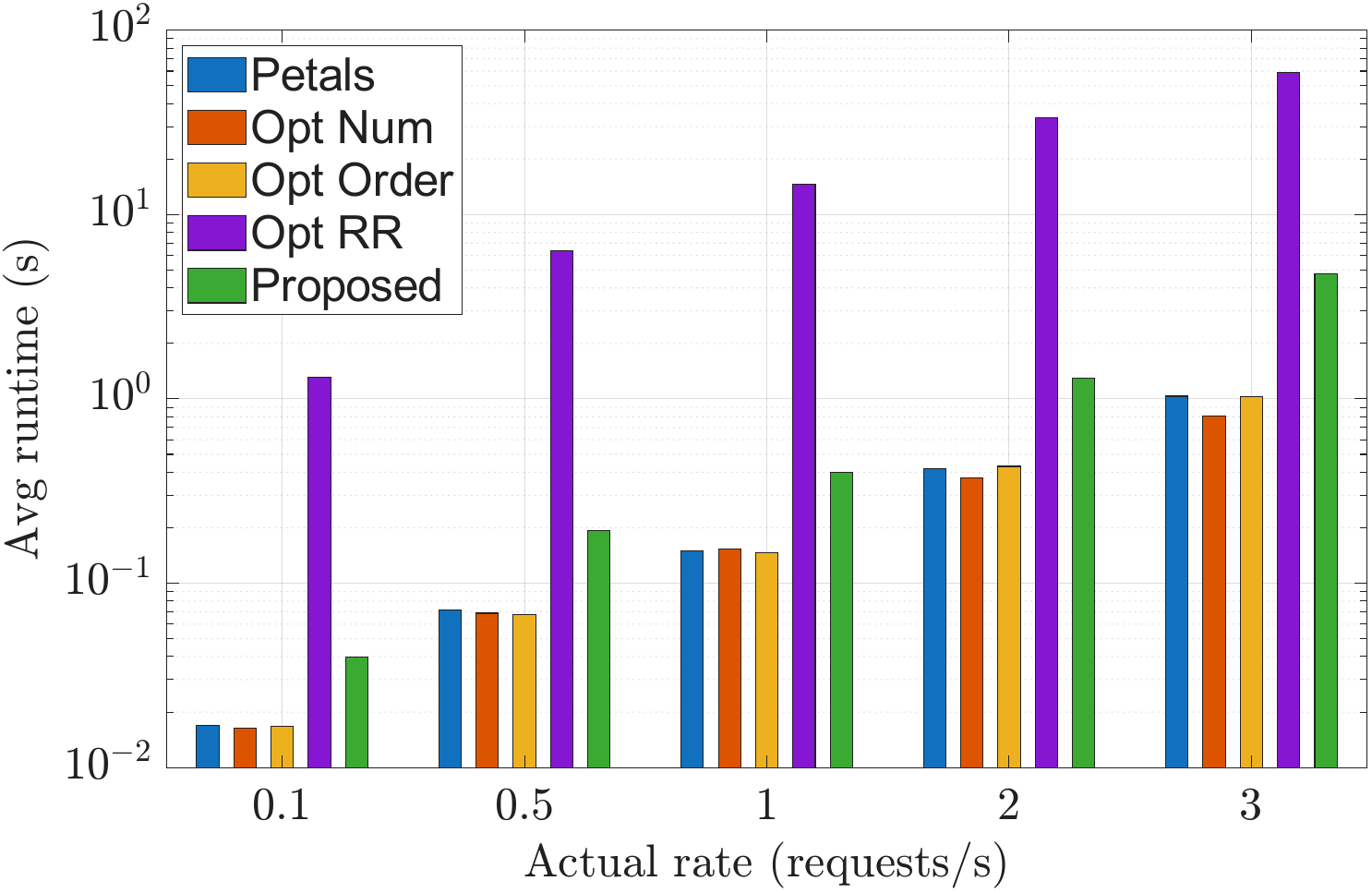}
  \caption{GTS-CE}
\end{subfigure}

\vspace{-0.6em}
\caption{Algorithm running time when varying actual rate $\lambda$ ($C = 0.4 \cdot \mbox{\#nodes}$, $\lambda_{base}=0.5$, $N_R=200\cdot \lambda$, $\eta = 0.2$, $\lmax^I=20$, $\lmax = 128$).}
\label{fig:running_time_vary_actual_lambda}
\vspace{-0.4em}
\end{figure}

\begin{figure}[t!]
\centering
\begin{subfigure}[t]{0.31\textwidth}
  \centering
  \includegraphics[width=\linewidth]{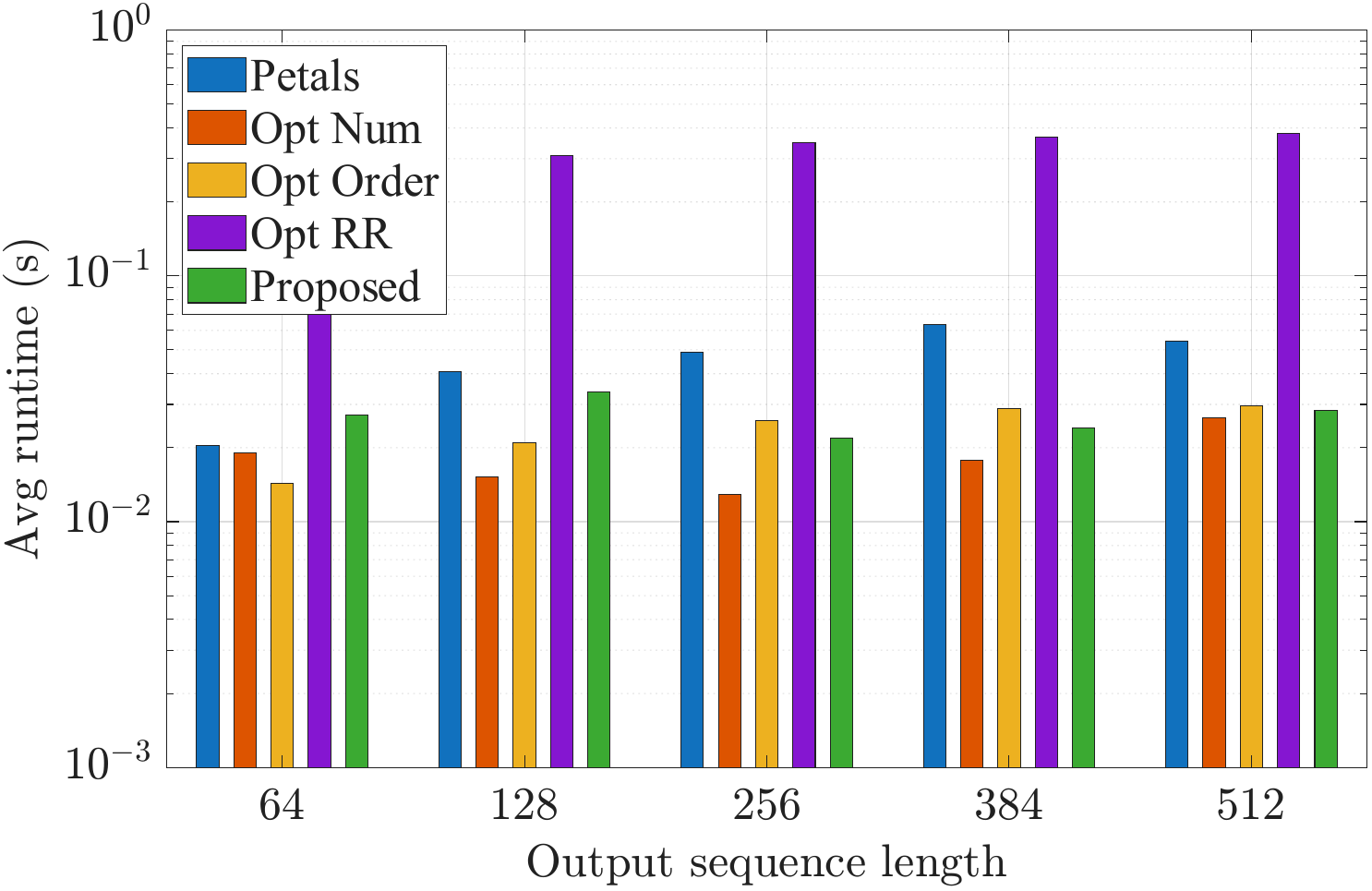}
  \caption{AboveNet}
\end{subfigure}\hfill
\begin{subfigure}[t]{0.31\textwidth}
  \centering
  \includegraphics[width=\linewidth]{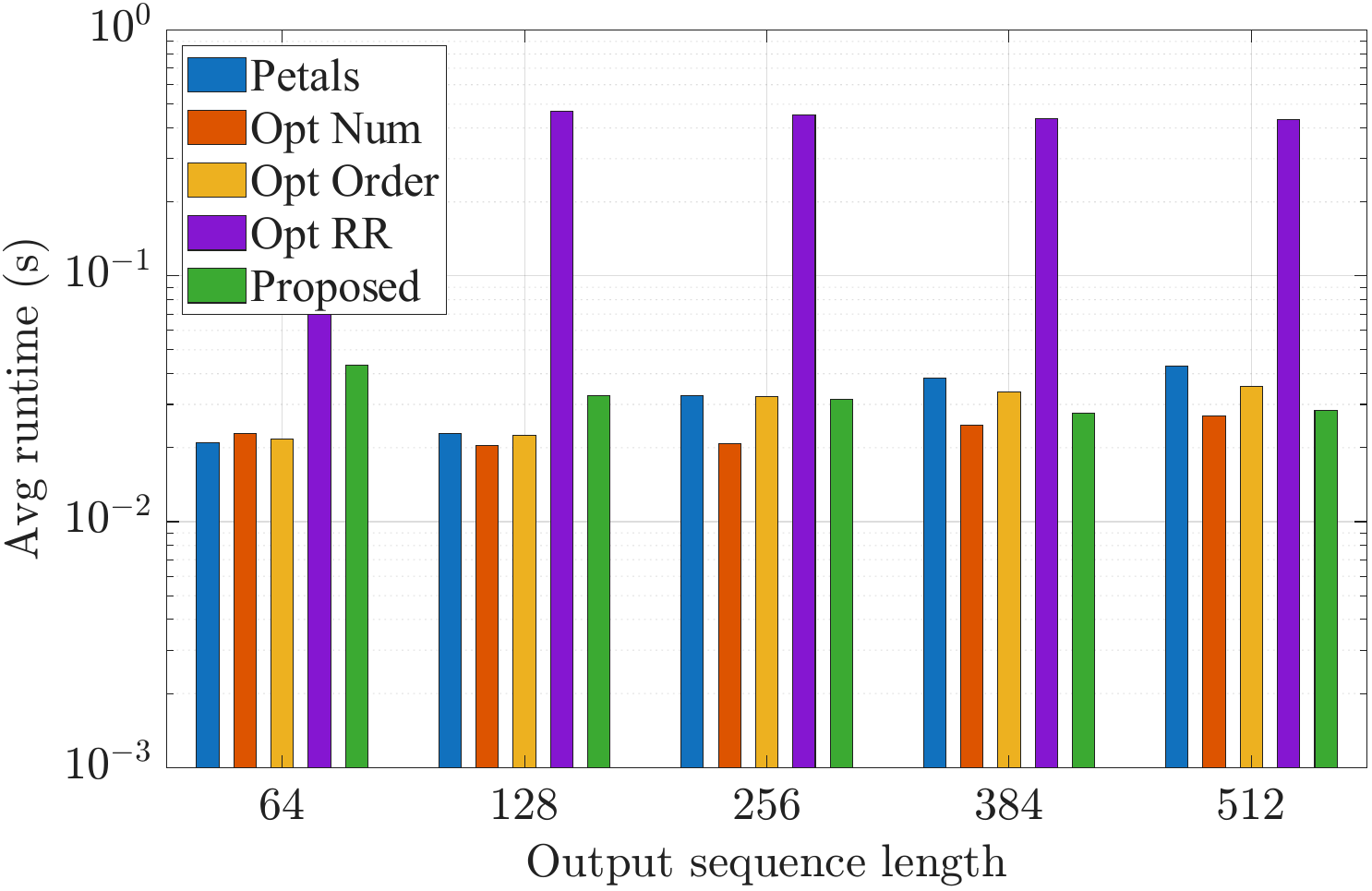}
  \caption{BellCanada}
\end{subfigure}\hfill
\begin{subfigure}[t]{0.31\textwidth}
  \centering
  \includegraphics[width=\linewidth]{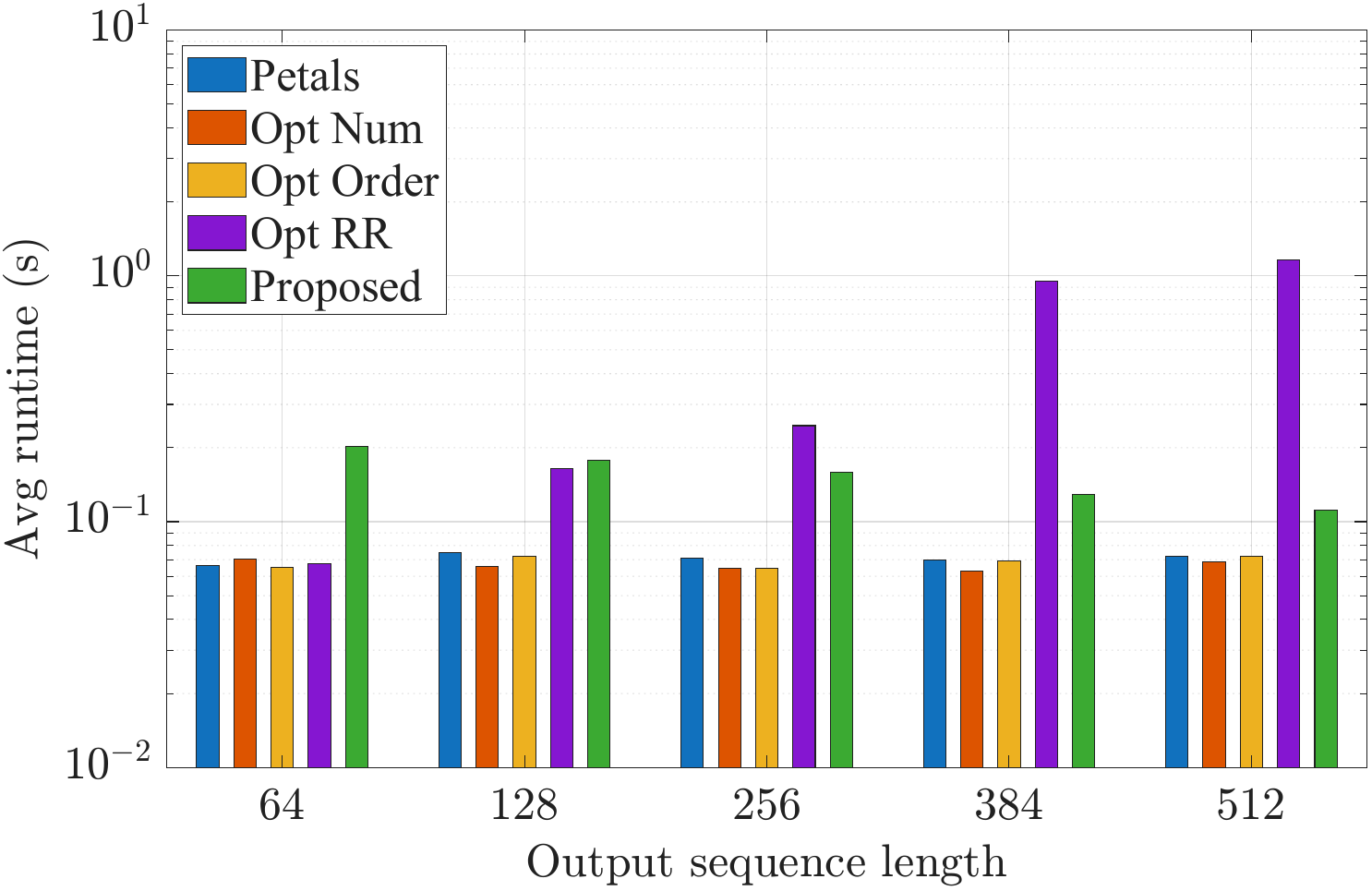}
  \caption{GTS-CE}
\end{subfigure}

\vspace{-0.6em}
\caption{Algorithm running time when varying the sequence length $\lmax$ ($C = 0.4 \cdot \mbox{\#nodes}$, $\eta=0.2$, $\lambda=0.5$, $N_R=100$, $\lmax^I=20$).}
\label{fig:running_time_vary_lmax}
\vspace{-0.4em}
\end{figure}

\begin{figure}[t!]
\centering
\begin{subfigure}[t]{0.31\textwidth}
  \centering
  \includegraphics[width=\linewidth]{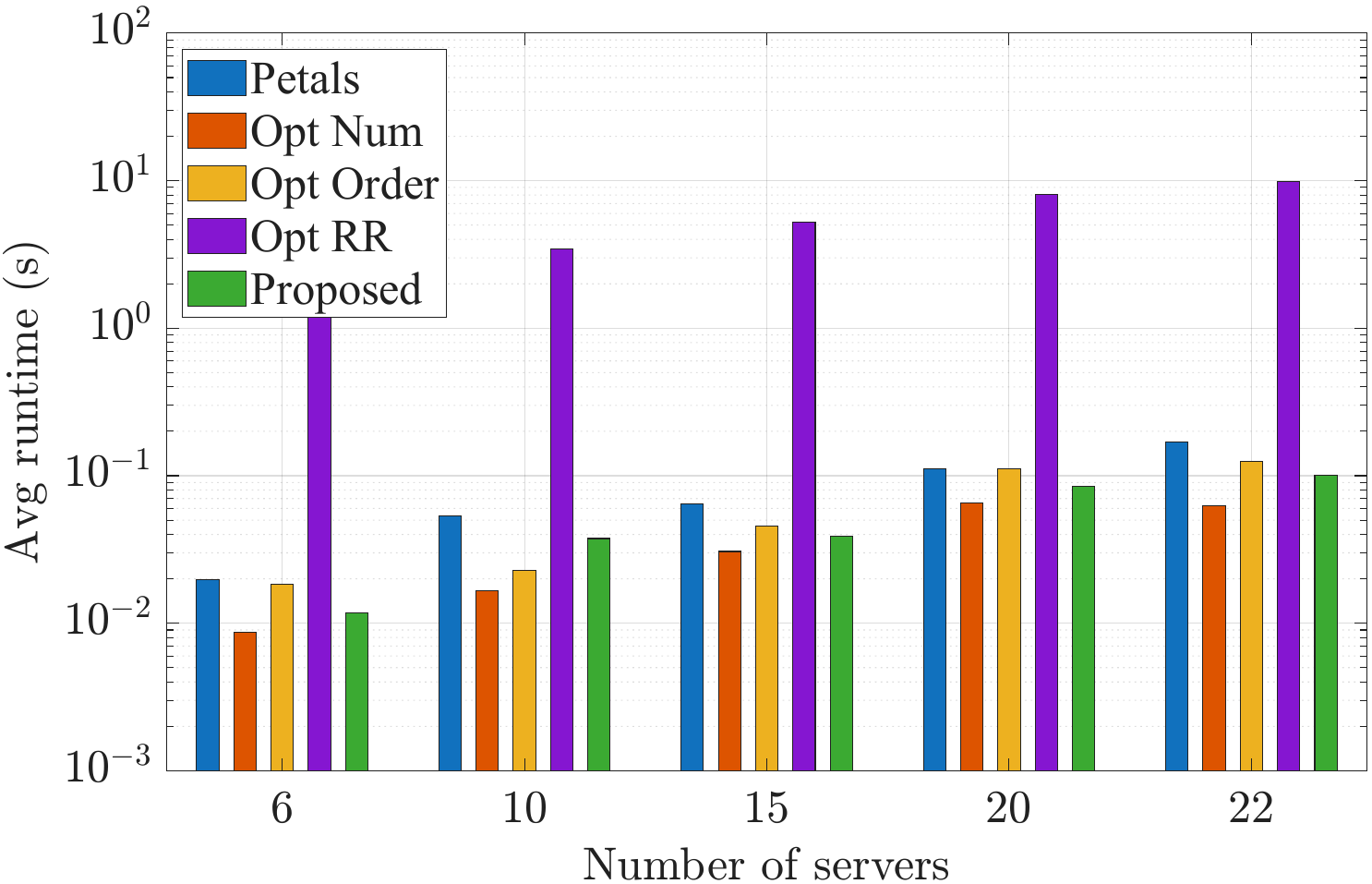}
  \caption{AboveNet}
\end{subfigure}\hfill
\begin{subfigure}[t]{0.31\textwidth}
  \centering
  \includegraphics[width=\linewidth]{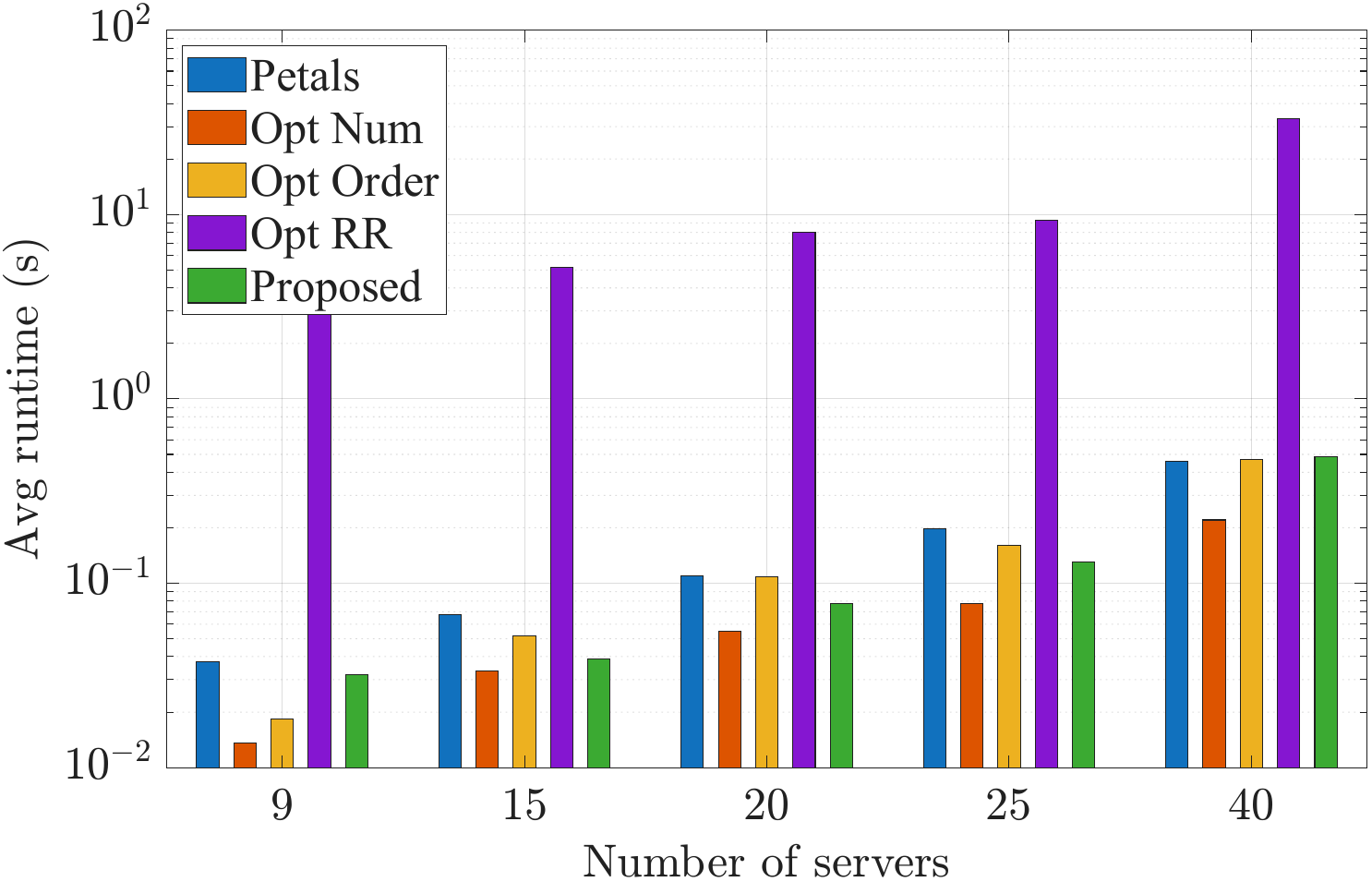}
  \caption{BellCanada}
\end{subfigure}\hfill
\begin{subfigure}[t]{0.31\textwidth}
  \centering
  \includegraphics[width=\linewidth]{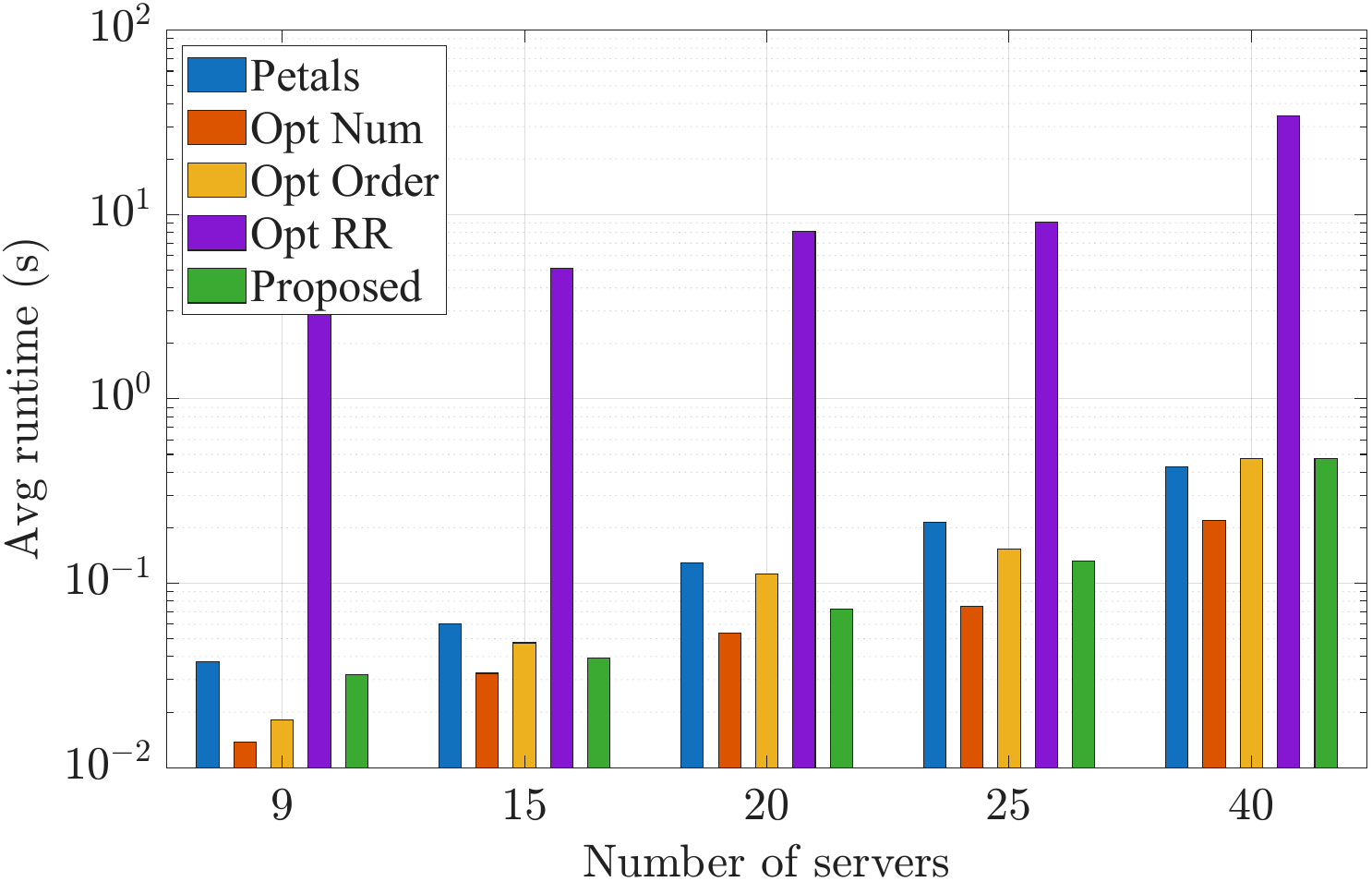}
  \caption{GTS-CE}
\end{subfigure}

\vspace{-0.6em}
\caption{Algorithm running time when varying the number of servers $C$ and rate $\lambda = {(0.1/9)}\cdot C$ ($N_R=100$, $\eta=0.2$, $\lmax^I=20$, $\lmax = 128$).}
\label{fig:running_time_vary_Clamda}
\vspace{-0.4em}
\end{figure}



Fig.~\ref{fig:running_time_vary_C}--\ref{fig:running_time_vary_lmax} show the running time of each algorithm in the simulations of Fig.~\ref{fig:inference_time_vary_C}--\ref{fig:inference_time_vary_lmax}. The results show that both the proposed algorithm (`Proposed') and the original algorithm in \cite{Borzunov23NeurIPS} (`Petals') are fast enough so that their execution costs negligible time compared to the actual inference time. The running time for `Optimized RR' is substantially higher than the others as in this case we directly solve the MILP \eqref{eq:online RR} by the Gurobi optimizer \cite{Gurobi}. The same observation holds in Fig.~\ref{fig:running_time_vary_Clamda}, which corresponds to the case in Fig.~\ref{fig:inference_time_vary_Clamda} as we proportionally increase both the number of servers and the request rate.

\end{document}